\documentclass[a4paper,USenglish,cleveref, autoref,thm-restate, numberwithinsect,pdfa]{lipics-v2021}

\pdfoutput=1 %
\hideLIPIcs  %

\newboolean{long}
\newcommand{\AppendixSymbol}{\ding{72}}
\setboolean{long}{true} %
\ifthenelse{\boolean{long}}{
	\NewDocumentEnvironment{prooflater}{m}{\begin{proof}}{\end{proof}\ignorespacesafterend}
	\NewDocumentEnvironment{proofsketch}{o +b}{}{\ignorespacesafterend}
	\newcommand{\restateref}[1]{}
	\NewDocumentEnvironment{statelater}{m}{}{}

	\NewDocumentCommand{\onlyShort}{+m}{}
	\NewDocumentCommand{\onlyLong}{+m}{#1}
	\NewDocumentCommand{\shortLong}{+m +m}{#2}
}{
	\NewDocumentEnvironment{prooflater}{m +b}{%
		\expandafter\global\expandafter\def\csname#1\endcsname{\begin{proof}#2\end{proof}}%
	}{\ignorespacesafterend}
	\NewDocumentEnvironment{proofsketch}{O{Proof sketch.}}{\begin{proof}[#1]}{\end{proof}\ignorespacesafterend}
	\usepackage{apptools}
	\newcommand{\restateref}[1]{[\IfAppendix{\hyperref[#1]{\AppendixSymbol{}}}{\hyperref[#1*]{\AppendixSymbol{}}}]}

	\NewDocumentEnvironment{statelater}{m +b}{%
		\expandafter\global\expandafter\def\csname#1\endcsname{#2}%
	}{\ignorespacesafterend}

	\NewDocumentCommand{\onlyShort}{+m}{#1}
	\NewDocumentCommand{\onlyLong}{+m}{}
	\NewDocumentCommand{\shortLong}{+m +m}{#1}
}

\graphicspath{{./graphics/}}%

\title{Revisiting Graph Modification via Disk Scaling:\\ From One Radius to Interval-Based Radii}

\titlerunning{Revisiting Graph Modification via Disk Scaling} %

\author{Thomas Depian}{Algorithms and Complexity Group, TU Wien, Austria}{tdepian@ac.tuwien.ac.at}{https://orcid.org/0009-0003-7498-6271}{Supported by the Vienna Science and Technology Fund (WWTF) under grant ICT22-029.}%

\author{Frank Sommer}{Institute of Computer Science, Friedrich Schiller University Jena, Germany}{frank.sommer@uni-jena.de}{https://orcid.org/0000-0003-4034-525X}{Supported by the Alexander von Humboldt Foundation.}

\authorrunning{T. Depian and F. Sommer} %

\Copyright{Thomas Depian and Frank Sommer} %



\category{} %

\relatedversion{} %

\nolinenumbers %

\EventEditors{John Q. Open and Joan R. Access}
\EventNoEds{2}
\EventLongTitle{42nd Conference on Very Important Topics (CVIT 2016)}
\EventShortTitle{CVIT 2016}
\EventAcronym{CVIT}
\EventYear{2016}
\EventDate{December 24--27, 2016}
\EventLocation{Little Whinging, United Kingdom}
\EventLogo{}
\SeriesVolume{42}
\ArticleNo{23}
\usepackage{todonotes}
\usepackage{mathtools}
\usepackage{xspace}
\usepackage{complexity}
\usepackage{dsfont}
\usepackage[ruled]{algorithm2e}
\usepackage{framed}
\usepackage{cite}
\usepackage{pifont}

\usetikzlibrary{arrows.meta,patterns,ipe}

\Crefname{observation}{Observation}{Observations}

\Crefname{property}{Property}{Properties}
\theoremstyle{claimstyle}
\newtheorem{branching}{Branching}
\Crefname{branching}{Branching}{Branchings}
\newtheorem{drule}{Rule}
\Crefname{drule}{Rule}{Rules}
\DeclareMathOperator{\argmax}{arg \max}

\let\oldrestatable\restatable
\def\restatable{\expandafter\oldrestatable}

\newcommand{\probname}[1]{{\normalfont\textsc{#1}}}
\newcommand{\prob}[3]{
		\begin{framed}
			\noindent \probname{#1} \\
			\textbf{Input:} #2 \\
			\textbf{Question:} #3
		\end{framed}
}

\newcommand{\instance}{\ensuremath{\mathcal{I}}\xspace}
\newcommand{\instanceLong}{\ensuremath{(\mathcal{S}, \rmin, \rmax, k)}\xspace}
\newcommand{\instanceLongRadius}{\ensuremath{(\mathcal{S}, \mathcal{T}, H, \rmin, \rmax)}\xspace}

\NewDocumentCommand{\PiScaling}{O{\ensuremath{\Pi}}}{#1\probname{-Scaling}\xspace}
\NewDocumentCommand{\PiScalingLong}{O{\ensuremath{\Pi}}}{\probname{Scaling To} #1\xspace}
\NewDocumentCommand{\PiRecognition}{O{\ensuremath{\Pi}}}{#1\probname{-Recognition}\xspace}
\NewDocumentCommand{\Cluster}{O{Cluster}}{\probname{Scaling To #1}\xspace}
\NewDocumentCommand{\ISc}{O{$c$-IS}}{\probname{Scaling To #1}\xspace}
\NewDocumentCommand{\ISt}{O{IS}}{\probname{Scaling To #1}\xspace}
\newcommand{\Complete}{\PiScalingLong[\probname{Complete}]}
\newcommand{\Connected}{\PiScalingLong[\probname{Connected}]}

\NewDocumentCommand{\GridTiling}{O{\ensuremath{>}}}{\probname{Grid Tiling With #1}\xspace}
\newcommand{\VC}{\probname{Vertex Cover}\xspace}
\newcommand{\IS}{\probname{Independent Set}\xspace}
\newcommand{\RadiusFeasibility}{\probname{ConScal}\xspace}
\newcommand{\RadiusFeasibilityLong}{\probname{Constrained Scaling To Graph}\xspace}

\newcommand{\Unit}{\ensuremath{\mathds{1}}}

\newcommand{\Size}[1]{\ensuremath{\left\vert #1 \right\vert}}
\newcommand{\Dist}[1]{\ensuremath{\left\Vert #1 \right\Vert_2}}
\newcommand{\BigO}[1]{\ensuremath{\mathcal{O}(#1)}}
\newcommand{\OO}{\mathcal{O}}
\newcommand{\col}{\ensuremath{\mathsf{col}}}
\newcommand{\red}{\ensuremath{\mathsf{red}}}
\newcommand{\blue}{\ensuremath{\mathsf{blue}}}
\newcommand{\green}{\ensuremath{\mathsf{green}}}
\newcommand{\dist}{\ensuremath{\mathsf{dist}}}
\newcommand{\far}{\ensuremath{\mathsf{far}}}
\newcommand{\clo}{\ensuremath{\mathsf{clo}}}
\newcommand{\rmin}{\ensuremath{r_{\min}}}
\newcommand{\rmax}{\ensuremath{r_{\max}}}
\newcommand{\nil}{\ensuremath{\textsf{nil}}}
\newcommand{\customEnumerateItem}[1]{\item[\textcolor{lipicsGray}{\sffamily\bfseries\upshape\mathversion{bold} #1}]}

\newcommand{\NewText}[1]{{#1}}

\begin{document}

\maketitle

\begin{abstract}
	For a fixed graph class $\Pi$, the goal of \probname{$\Pi$-Modification} is to transform an input graph $G$ into a graph~$H\in\Pi$ using at most $k$ modifications.
	Vertex and edge deletions are common operations, and their (parameterized) complexity for various $\Pi$ is well-studied.
	Classic graph modification operations such as edge deletion do not consider the geometric nature of intersection graphs such as (unit) disk graphs.
	This led Fomin et al.~[ITCS'~25] to introduce scaling as a geometric graph modification operation for unit disk graphs: For a given radius $r$, each modified disk will be rescaled to radius $r$.
	
	In this paper, we generalize their model by allowing rescaled disks to choose a radius within a given interval $[\rmin, \rmax]$ and study the (parameterized) complexity (with respect to $k$) of the corresponding problem \PiScaling.
	We show that \PiScaling is in \XP\ for every graph class~$\Pi$ that can be recognized in polynomial time. %
	Furthermore, we show that \PiScaling: (1) is \NP-hard and \FPT\ for cluster graphs, (2) can be solved in polynomial time for complete graphs, and (3) is \W[1]-hard for connected graphs.
	In particular, (1) and (2) answer open questions of Fomin et al.\ and (3) generalizes the hardness result for their variant where the set of scalable disks is restricted.

\end{abstract}

\section{Introduction}
\label{sec:introduction}

For a fixed graph class $\Pi$, the task in the \probname{$\Pi$-Modification} problem is to transform an input graph $G$ into a graph~$H\in\Pi$ using at most $k$ modifications.
Common modification operations include vertex deletion, edge deletion, and edge addition~\cite{CDFG23}, but other operations such as vertex-splitting have also been studied~\cite{ABUKHZAM2025}.
The parameterized complexity of \probname{$\Pi$-Modification} is well-understood, especially when parameterized by the number $k$ of modifications~\cite{CDFG23,CFK+.PA.2015,FLSZ19,Komusiewicz18,Yannakakis81,LY80}. %
Usually, we aim for an \FPT-algorithm, i.e., an algorithm with running time~$f(k)\cdot n^{\OO(1)}$, where $f$ is a computable function that depends only on~$k$, and $n$ is the overall instance size~\cite{CFK+.PA.2015}.
Some parameterized problems, however, are \W[$t$]-hard with respect to~$k$ for some $t \geq 1$.
This rules out an \FPT-algorithm for $k$ under common complexity assumptions, and we then aim for an \XP-algorithm, i.e., with running time~$n^{f(k)}$.

Due to their applications in wireless communication~\cite{HS95} and capabilities to model sensor networks~\cite{Santi05}, \emph{geometric intersection graphs} are an important graph family.
There, each vertex corresponds to a geometric object and two vertices are adjacent if and only if the two corresponding objects intersect.
The classical (non-parameterized) complexity of problems restricted to geometric intersection graphs is well-understood~\cite{HK01,Fishkin03}, %
and %
also their parameterized complexity has been studied recently~\cite{XZ25,FLS12,FLPSZ19,BBKMZ18,PSZ24}. %
Disk graphs are probably the most prominent family of geometric graphs.
In a disk graph $G(\mathcal{S}, r)$, each vertex corresponds to a point $p \in \mathcal{S}$ in the plane 
and two vertices $p$ and $q$ are adjacent if and only if their corresponding \emph{closed} disks with radius $r(p)$ and $r(q)$ intersect; see also \Cref{fig:example-modification} for an example.
In unit disk graphs, each disk has radius one and their geometric nature can lead to specialized, more efficient, algorithms on these graphs. %
One prominent example is the \probname{Clique} problem (find $k$~vertices which are pairwise adjacent), which is \NP-hard on general graphs~\cite{Karp72}, but can be solved in polynomial-time on unit disk graphs~\cite{CCJ90,Eppstein09,EKM.FMC.2023}.

Traditionally, vertex- and edge-based modifications were the primary modification operations for \probname{$\Pi$-Modification}, even when restricted to %
geometric intersection graphs~\cite{XZ25}. 
However, some of these operations, such as edge addition, ignore the underlying geometry as seen in \Cref{fig:example-modification}. %
This raises the question of their appropriateness for unit disk and other geometric intersection graphs.
As a response to this, two ``geometry-preserving'' modification operations were recently studied: (a)~disk dispersion~\cite{FG00Z23,DKHO24,DKHO25,DKHOW25}, where disks can move within a bounded radius, and (b)~disk scaling~\cite{FG00Z25}, where disks can receive a radius different from~1.%

\begin{figure}
	\centering
	\includegraphics[page=1]{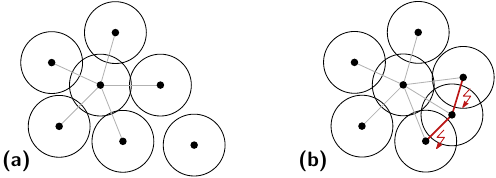}
	\caption{\textbf{\textsf{(a)}} Turning the underlying abstract graph into a $K_{1,6}$ by adding a single edge is easy. However, this is not possible while maintaining a unit disk graph as \textbf{\textsf{(b)}} demonstrates.}
	\label{fig:example-modification}
\end{figure}

We focus on the latter model, motivated by the observation that in applications such as sensor networks, the position of a sensor is fixed but its transmission range, modeled as the radius, can be adapted~\cite{Santi05}. %
More precisely, in the %
model introduced by Fomin et al.~\cite{FG00Z25}, the input contains a non-negative number~$\alpha \neq 1$, %
and each modified disk receives radius $\alpha$, i.e., all scaled disks are either \emph{shrunken} ($\alpha< 1$) or \emph{expanded} ($\alpha > 1$).
They studied disk-scaling %
for obtaining an independent set, an acyclic graph, or a connected graph, and provided \FPT-algorithms with respect to~$k$, as well as (complementing) \W[1]- and \NP-hardness results. 
A limitation of their model is that each scaled disk receives the same radius~$\alpha$.
In particular, 
all modified disks are either shrunken or expanded.\footnote{
	Fomin et al.~\cite{FG00Z25} also studied an optimization variant where one scales (at most) $k$~disks within a given interval~$[\alpha,1]$ for~$0<\alpha<1$ and minimizes the sum of all radii.
	This model also has the limitation that all disks are either shrunken or expanded.
	We do not consider this ``$\min$-variant'' in this work, but discuss it, along with other interesting generalizations, in \Cref{sec:conclusion}.}
In this work, we generalize their model by allowing each rescaled disk to choose a radius within a given interval~$[\rmin, \rmax]$, hence enabling disk shrinking and expanding simultaneously whenever~$r_{\min}<1$ and~$r_{\max}> 1$: %

\prob{%
	\PiScalingLong~(\PiScaling)
}%
{
	Set $\mathcal{S}$ of $n$ points in the plane, $0 < \rmin \leq \rmax \in \mathbb{R}$, and $k \in \mathbb{N}$.  
}%
{
	Does there exist a set $\mathcal{T} \subseteq \mathcal{S}$ and a radius assignment $r\colon \mathcal{S} \to \mathbb{R}_{> 0}$ such that $\Size{\mathcal{T}} \leq k$, $\rmin \leq r(p) \leq \rmax$ for all $p \in \mathcal{T}$ and $r(p) = 1$ otherwise, and $G(\mathcal{S}, r) \in \Pi$?
}
Intuitively, \PiScaling asks whether one can modify at most $k$ radii in a given unit disk graph $G(\mathcal{S}, \Unit)$ to values in $[\rmin,\rmax]$ such that the resulting disk graph belongs to~$\Pi$.

\subparagraph*{Our Contributions.}
We first show that \PiScaling is in \XP\ with respect to~$k$ for every graph class $\Pi$ whose recognition problem can be solved efficiently~(\Cref{thm:xp-framework}).
At first glance this might be surprising, since, although we can guess the scaled disks~$\mathcal{T}$, we cannot branch on the radius assignment $r$ or the target graph $H \in \Pi$.
To this end, we first guess for each~$p\in\mathcal{T}$ its \emph{furthest unscaled neighbor}~$\far(p)$ in $G(\mathcal{S}, r)$.
This allows us to reconstruct all neighbors of $p$ and eventually the target graph $H$ due to the geometric nature of the problem.
Afterwards, we devise a linear programming (LP) formulation (\Cref{sec:linear-programming}) that determines suitable radii once the disks to scale, i.e., $\mathcal{T}$, and the resulting intersection graph $H$ is fixed. 

Next, we focus on specific graph classes and start with cluster graphs (see \Cref{sec:preliminaries} for the definition of the graph classes). 
Cluster graphs are in particular challenging since they can require to choose a specific radius inside the interval~$[r_{\min}, r_{\max}]$; see also %
\Cref{fig:example-cluster}.
We improve the upper bound implied by \Cref{thm:xp-framework} and provide an \FPT-algorithm for \Cluster 
parameterized by~$k$ (\Cref{thm:cluster-fpt}).
The algorithm uses a branch-and-bound procedure to iteratively determine %
a new disk $p$ to scale.
However, since we do not know the final radius of $p$ beforehand, we have to further branch to determine the adjacencies of~$p$; essentially gradually fixing the target cluster graph $H$.
To ensure \FPT-time, we %
exploit the geometry within our problem.  %
In particular, we show that there is only a bounded number of ``ambiguous'' neighbors and non-neighbors of $p$: $\far(p)$, this time the furthest unscaled disk in the \emph{same} cluster as~$p$ in the solution, and $\clo(p)$, which is the \emph{closest} unscaled disk not adjacent to~$p$, i.e., $\clo(p)$ and $p$ are in \emph{different} clusters in the solution.
Once $\far(p)$ and $\clo(p)$ are determined, we can essentially ``read-off'' the (unscaled) neighbors of $p$ in~$H$.
To complement this result, we show that \Cluster is \NP-hard for any fixed rationals~$0 < r_{\min} \leq r_{\max}$ (\Cref{thm:cluster-np-hard-general}), except for the case when~$r_{\min}=1=r_{\max}$.

\begin{figure}
	\centering
	\includegraphics[page=1]{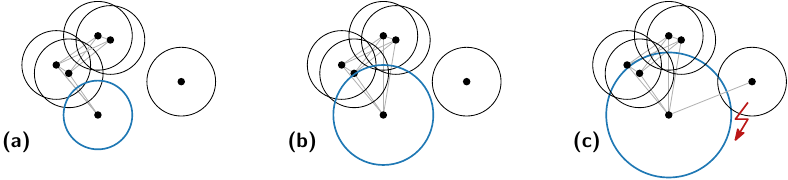}
	\caption{We can turn \textbf{\textsf{(a)}} into a cluster graph by scaling the blue disk as shown in \textbf{\textsf{(b)}}.
		However, if we enlarge the disk too much as in \textbf{\textsf{(c)}}, the resulting disk graph might no longer be a cluster graph.}
	\label{fig:example-cluster}
\end{figure}

As the second special graph class we consider complete graphs.
Here, it is optimal to scale each modified disk to~$r_{\max}$. 
We show that the corresponding problem \Complete can be solved in polynomial time (\Cref{thm:complete}) by leveraging the polynomial-time algorithm for \probname{Clique} in unit disk graphs~\cite{EKM.FMC.2023}. %
\Cref{thm:complete,thm:cluster-fpt} settle two open questions of Fomin et al.~\cite{FG00Z25};
for cluster graphs, even in our more general model with interval-based radii, for complete graphs, there is no (algorithmic) difference in the models.

As our last result, we show that the general \XP-algorithm (\Cref{thm:xp-framework,remark:xp-framework}) is for some %
graph classes~$\Pi$ the best we can hope for, i.e., cannot be improved to an \FPT-algorithm.
For example, if we want to obtain a graph that contains an independent set of size at least~$c$, the problem \ISc is \W[1]-hard for~$c$, even when~$k=0$.
This follows from the known \W[1]-hardness of detecting a size~$c$ independent set in a unit disk graph~\cite[Theorem~14.34]{CFK+.PA.2015}.
In \Cref{sec:connected}, we focus on a graph class where obtaining \W[1]-hardness is ``less-trivial''.
More concretely, we consider connected graphs, and show that \Connected is \W[1]-hard parameterized by~$k$ (\Cref{thm:connected-w1-hard}).
More precisely, in our reduction all scaled disks get expanded (with the same radius).
Hence, our result generalizes a result of Fomin et al.~\cite{FG00Z25}, who proved \W[1]-hardness %
with respect to~$k$ when~$r_{\min}=r_{\max}>1$ in the variant where the set of scalable (expandable) disks is restricted. 
This is in sharp contrast to the case when \emph{at least} $k$ disks must be shrunken to a radius~$\alpha$, which is \FPT\ with respect to~$k$~\cite[Theorem 3]{FG00Z25}.
Thus, the complexity of \PiScaling can differ drastically depending on~$r_{\min}\ge 1$ or~$r_{\max}\le 1$.

\begin{figure}
	\centering
	\input{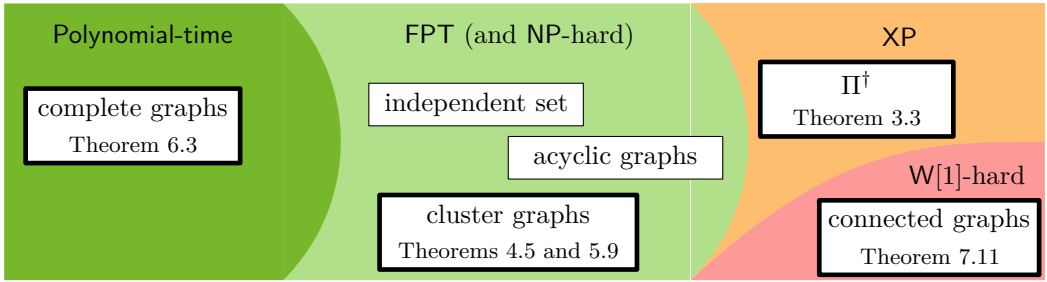}
	\caption{The (parameterized) complexity of \PiScaling. Bold boxes are results from this paper and the remaining follow from~\cite{FG00Z25}. \dag: Assuming that recognizing graphs in $\Pi$ takes polynomial time.}
	\label{fig:results}
\end{figure}

Finally, we would like to emphasize that many (algorithmic) results obtained by Fomin et al.~\cite{FG00Z25} directly transfer to our more general model.
This is due to the observation that they considered graph classes where it is always safe to either choose only~$r_{\max}$ or only~$r_{\min}$ as the radius for each scaled disk. %
We summarize our results and this observation in \Cref{fig:results}.

\onlyShort{
	\smallskip
	\noindent
	\emph{Due to space constraints, full proofs of statements marked by \AppendixSymbol{} are deferred to the appendix.}
}

\section{Preliminaries}
\label{sec:preliminaries}
We assume familiarity with the basic notions from graph theory~\cite{Die.GT4.2012} and the foundations of 
parameterized complexity theory~\cite{CFK+.PA.2015}.
For our algorithmic results, we assume the \emph{Real RAM} computational model, i.e., that we can perform arithmetic operations on real numbers in unit time~\cite{FG00Z25,EvdHM.SGN.2024}.
For an integer $p \geq 1$, we define $[p]\coloneqq\{1, 2, \ldots, p\}$.
Let $G$ be a simple and undirected graph with vertex set $V(G)$ and edge set $E(G)$.
We write $V$ and $E$ if $G$ is clear from the context. %
For a vertex set $V' \subseteq V$, we let $G[V']$ denote the subgraph of $G$ \emph{induced} by $V'$.
For a vertex $v \in V$, we let $N_G(v) = \{u \in V : uv \in E\}$ denote the \emph{neighborhood} of $v$. %

\subparagraph*{Graph Classes.}
A graph $G$ is \emph{complete} if and only if $uv \in E$ for every pair $u,v\in V$.
We call~$G$ a \emph{cluster graph} if every connected component $C$ of $G$ is a clique, also referred to as \emph{cluster}.
Equivalently,~$G$ is a cluster graph if and only if it does not contain an \emph{induced $P_3$}, i.e., three vertices $u,v,w \in V$ such that $uv, vw \in E$ but $uw \notin E$.
\NewText{Note that throughout this paper, we only} consider induced~$P_3$s.
Hence, a \emph{$P_3$-free} graph is a graph without an induced~$P_3$, i.e., a cluster graph.

\subparagraph*{Disk Graphs.}
For two points $p, p' \in \mathbb{R}^2$ in the plane, we denote with $\Dist{p - p'}$ their \emph{Euclidean distance}.
We use $x(p)$ and $y(p)$ to denote the $x$- and $y$-coordinates of $p$, respectively.
Let $\mathcal{S} = \{p_1, \ldots, p_n\}$ be a set of $n$ points in the plane. %
A \emph{radius assignment} $r\colon \mathcal{S} \to \mathbb{R}_{> 0}$ assigns to each point $p \in \mathcal{S}$ a positive \emph{radius} $r(p)$.
Given $\mathcal{S}$ and $r$, we 
define the corresponding \emph{disk graph} $G(\mathcal{S}, r) \coloneqq (\mathcal{S}, \{p_ip_j : i < j, \Dist{p_i - p_j} \leq r(p_i) + r(p_j)\})$ as the graph that contains one vertex for each point and two vertices $p_i$ and $p_j$ are adjacent if and only if their closed disks of radii $r(p_i)$ and $r(p_j)$ intersect.
If $r(p_i) = r(p_j)$ holds for all $p_i, p_j \in \mathcal{S}$, then $G(\mathcal{S}, r)$ is a \emph{unit disk graph}. %
In this case, we may assume without loss of generality $r(p) = 1$ for all $p \in \mathcal{S}$ and write $G(\mathcal{S}, \Unit)$ as shorthand. %
Throughout this paper, we identify disks in $G(\mathcal{S}, r)$ with their center point $p \in \mathcal{S}$.
Let $\instance = \instanceLong$ be an instance of \PiScaling.
We also call $k$ the \emph{budget}.
For a radius assignment $r$, we define $\mathcal{T}(r) \coloneqq \{p \in \mathcal{S} : r(p) \neq 1\}$.
We call the disks $\mathcal{T}(r)$ \emph{scaled} and all remaining disks $\mathcal{S} \setminus \mathcal{T}(r)$ \emph{unscaled}.
Note that in \PiScaling, it is sufficient to ask for a radius assignment $r$, since we can set $\mathcal{T} = \mathcal{T}(r)$.
Although \PiScaling is defined as a decision problem for complexity-theoretic reasons, we remark that all presented algorithms can output, for yes-instances, a witnessing radius assignment $r$ (i.e., a \emph{solution}).

\section{A General \textsf{XP}-Membership Framework}
\label{sec:frameworks}
Before we study \PiScaling for different graph classes $\Pi$, we first introduce a general framework that implies \XP-membership with respect to $k$ of \PiScaling for many graph classes~$\Pi$.
We condition the applicability of our framework on properties of the problem \PiRecognition, which asks, for a given graph $H$, whether $H \in \Pi$ holds.
We next present a linear programming formulation that is a key building block in our framework.

\subsection{A Linear Program for Finding Radii}
\label{sec:linear-programming}
In this section, we consider the problem \RadiusFeasibilityLong (\RadiusFeasibility), defined as follows.
Given a set $\mathcal{S} \subset \mathbb{R}^2$ of $n$ points, a set $\mathcal{T} \subseteq \mathcal{S}$, two real numbers $\rmin \leq \rmax \in \mathbb{R}_{> 0}$, and a graph $H = (\mathcal{S}, E)$, the question is whether there exists a radius assignment $r\colon \mathcal{S} \to \mathbb{R}_{>0}$ such that (a) $\rmin \leq r(p) \leq \rmax$ if $p \in \mathcal{T}$ and $r(p) = 1$ otherwise, and (b) $G(\mathcal{S}, r) = H$?
\RadiusFeasibility can be seen as a restricted version of \PiScaling, where the set of scaled disks is fixed and the target graph $H \in \Pi$ is known.
Later we will ``reduce'' \PiScaling to \RadiusFeasibility.
Let $\instance = \instanceLongRadius$ be an instance of \RadiusFeasibility. 
Without loss of generality, we assume $G(\mathcal{S} \setminus \mathcal{T}, \Unit) = H[\mathcal{S} \setminus \mathcal{T}]$; otherwise \instance is a no-instance.
\shortLong{%
	We sketch a linear program~$\mathcal{P}(\instance)$ to solve \RadiusFeasibility efficiently; see \Cref{app:linear-programming} for details.\footnote{Fomin et al.~\cite{FG00Z25} also used reduction to linear programs to determine suitable radii in their ``min-variant''.}
	
	We introduce one variable $x_p \in \mathbb{R}$ for each $p \in \mathcal{T}$ representing its radius $r(p) = x_p$.
	Furthermore, we introduce one ``slack-variable'' $\varepsilon \geq 0$ that helps us to model strict inequalities.
	To enforce $\rmin \leq r(p) \leq \rmax$, we add the constraints $\rmin \leq x_p$ and $x_p \leq \rmax$.
	Next, we consider the \emph{unscaled neighborhood} $N_H(p) \setminus \mathcal{T}$ of each $p \in \mathcal{T}$.
	We must ensure that $up \in E(G(\mathcal{S}, r))$ for every $u \in N_H(p) \setminus \mathcal{T}$.
	Since $u \notin \mathcal{T}$ implies $r(u) = 1$, this yields the constraint $x_p + 1 \geq \Dist{p - u}$.
	Similarly, for every point $u \in \mathcal{S} \setminus (N_H(p) \cup \mathcal{T})$, we must have $up \notin E(G(\mathcal{S}, r))$, resulting in the constraint $x_p + 1 \leq \Dist{p - u} - \varepsilon$.
	Observe that this encodes a strict inequality if $\varepsilon > 0$ holds.
	Finally, we consider every pair $p,q \in \binom{\mathcal{T}}{2}$.
	If $pq \in E(H)$, we introduce the constraint $x_p + x_q \geq \Dist{p - q}$, and otherwise we add $x_p + x_q \leq \Dist{p - q} - \varepsilon$.
	These constraints resemble the definitions of (non-)edges in disk graphs provided that $\varepsilon > 0$ holds.
	We ``reward'' this by maximizing $\varepsilon$ in the optimization function.
	One can now show:
}{%
	In the following, we devise a linear programming formulation~$\mathcal{P}(\instance)$ that we can use to solve \RadiusFeasibility in time $f(\Size{\mathcal{T}}) \cdot n^{\BigO{1}}$.
	At this point, we would like to note that also Fomin et al.~\cite{FG00Z25} used reduction to LPs to determine suitable radii in their algorithms for their ``$\min$-variant'' of the problem, where the amount of scaling should be minimized.
}

\begin{statelater}{linearProgrammingFormulation}
	\onlyShort{In the following, we devise a linear programming formulation~$\mathcal{P}(\instance)$ that we can use to solve \RadiusFeasibility in time $f(\Size{\mathcal{T}}) \cdot n^{\BigO{1}}$.}

	\subparagraph*{Variables.}
	We introduce one variable $x_p \in \mathbb{R}$ for every point $p \in \mathcal{T}$.
	The value of $x_p$ determines the radius for the disk~$p$ in the final radius assignment $r$, i.e., $r(p) = x_p$ will hold.
	Furthermore, we introduce an auxiliary variable $\varepsilon$ that will help us to represent strict inequalities.

	\subparagraph*{Constraints.}
	Since $\varepsilon$ should help us to model strict inequalities, we force $\varepsilon \geq 0$.
	For every point $p \in \mathcal{T}$, we introduce the following constraints.
	First, we introduce the constraints $\rmin \leq x_p$ and $x_p \leq \rmax$ to force that $\rmin \leq r(p) \leq \rmax$ holds.
	Next, we consider the neighborhood $N_H(p)$ of $p$ in $H$ and focus on points $u \in N_H(p) \setminus \mathcal{T}$, the so-called \emph{unscaled neighbors}.
	We must ensure that $up \in E(G(\mathcal{S}, r))$ for every $u \in N_H(p) \setminus \mathcal{T}$.
	Since $u \notin \mathcal{T}$ implies $r(u) = 1$, we can only control the existence of the edge via the variable $x_p$ and thus introduce the constraint $x_p + 1 \geq \Dist{p - u}$.
	We proceed similar for the points $u \in \mathcal{S} \setminus (N_H(p) \cup \mathcal{T})$, this time ensuring that the edge $up$ does not exist in $G(\mathcal{S}, r)$.
	More concretely, we introduce the constraint $x_p + 1 \leq \Dist{p - u} - \varepsilon$.
	We use the variable $\varepsilon$ as a slack-variable to model strict inequalities (for $\varepsilon > 0$).
	Later, we describe how to ensure $\varepsilon > 0$ using the objective function.
	
	Until now, we only ensured that our solution $r$ correctly represents the (non-)adjacency with respect to unscaled disks, i.e., points in $\mathcal{S} \setminus \mathcal{T}$.
	We now consider every pair $p,q \in \binom{\mathcal{T}}{2}$.
	If they are adjacent in $H$, we introduce the constraint $x_p + x_q \geq \Dist{p - q}$.
	On the other hand, if $p \notin N_H(q)$, we introduce the constraint $x_p + x_q \leq \Dist{p - q} - \varepsilon$.
	
	\subparagraph*{Objective Function.}
	We use the objective function: $\max \varepsilon$.
	This way, we consider radius assignments that also adhere to the strict inequalities modeled via our slack-variable $\varepsilon$.
	This completes the definition of the linear program $\mathcal{P}(\instance)$ and we now establish its correctness.
	
\end{statelater}

\begin{restatable}\restateref{lem:linear-program-correctness}{lemma}{lemmaLinearProgramCorrectness}
	\label{lem:linear-program-correctness}
	\instance is a yes-instance of \RadiusFeasibility if and only if
	\NewText{the linear program} $\mathcal{P}(\instance)$ admits a solution with $\varepsilon > 0$.
\end{restatable} 
\begin{prooflater}{plemmaLinearProgramCorrectness}
	We show both directions separately.
	
	\proofsubparagraph{($\boldsymbol{\Rightarrow}$)}
	Assume \instance is a yes-instance of \RadiusFeasibility and let $r$ be a radius assignment such that $G(\mathcal{S}, r) = H$ and $\rmin \leq r(p) \leq \rmax$ if $p \in \mathcal{T}$ and $r(p) = 1$ otherwise.
	We construct a variable assignment for $X$ by setting $x_p = r(p)$ for all $p \in \mathcal{T}$.
	For $\varepsilon$, we can set it to a value that is smaller than any distance between disk boundaries that do not intersect, i.e., $\varepsilon < \min_{pq \notin E(G(\mathcal{S}, r))} \Dist{p - q} - r(p) - r(q)$.
	Observe that since $r$ is a solution, we can a find a value for $\varepsilon$ such that $\varepsilon > 0$ holds.
	We now argue that this is also a solution to $\mathcal{P}(\instance)$ (with $\varepsilon > 0$).
	
	Towards a contradiction, assume that there exists an unsatisfied constraint.
	By the definition of $r$, the constraints $\rmin \leq x_p$ and $x_p \leq \rmax$ are satisfied for every $p \in \mathcal{T}$.
	Hence, one of the constraints derived from the vertex-pairs must be unsatisfied.
	Without loss of generality, assume that it is the constraint $x_p + 1 \geq \Dist{p - u}$ for some $pu \in E(H)$ (the other cases, i.e., $x_p + 1 \leq \Dist{p - u} - \varepsilon$, $x_p + x_q \geq \Dist{p - q}$, and $x_p + x_q \leq \Dist{p - q}- \varepsilon$ follow from analogous arguments together with the observation that $\varepsilon > 0$ holds).
	As the constraint is not satisfied, we have $x_p + 1 < \Dist{p - u}$.
	Since $r(p) = x_p$ and $r(u) = 1$ by definition of $r$, this means that the disks $p$ and $u$ do not intersect, i.e., $pu$ is not an edge in $G(\mathcal{S}, r)$.
	As $pu \in E(H)$, this contradicts $G(\mathcal{S}, r) = H$ and consequently the variable assignment must be a solution to $\mathcal{P}(\instance)$.
	
	\proofsubparagraph{($\boldsymbol{\Leftarrow}$)}
	Assume that $\mathcal{P}(\instance)$ admits a solution with $\varepsilon > 0$, i.e., there exists a value assignment for the variables that satisfies all constraints, and we have $\varepsilon >0$.
	We construct a radius assignment by setting $r(p) = x_p$ if $p \in \mathcal{T}$ and $r(p) = 1$ otherwise, i.e., if $p \in \mathcal{S} \setminus \mathcal{T}$.
	The constraints $\rmin \leq x_p$ and $x_p \leq \rmax$ ensure $\rmin \leq r(p) \leq \rmax$ for all $p \in \mathcal{T}$.
	We now show that $G(\mathcal{S}, r) = H$ holds and assume that there exists an edge $uv \in E(H) \setminus E(G(\mathcal{S}, r))$, i.e., the edge $uv$ is not in $G(\mathcal{S}, r)$.
	By our pre-processing check, we can assume without loss of generality that $u \in \mathcal{T}$ or $v \in \mathcal{T}$ (or both) holds.
	If $u \in \mathcal{T}$ and $v \notin \mathcal{T}$, we observe that the constraint $x_u + 1 \geq \Dist{v - u}$ is not satisfied.
	An analogous observation holds for $u \notin \mathcal{T}$ and $v \in \mathcal{T}$.
	If $u,v \in \mathcal{T}$, we observe that the constraint $x_u + x_v \geq \Dist{v - u}$ is not satisfied.
	Both cases lead to a contradiction.
	
	We now rule out the case $uv \in E(G(\mathcal{S}, r)) \setminus E(H)$ by analogous arguments.
	Assume that there exists an edge $uv \in E(G(\mathcal{S}, r))$ such that $uv \notin E(H)$.
	By our pre-processing check, we can assume $u \in \mathcal{T}$ or $v \in \mathcal{T}$ (or both).
	Consider the case $u \in \mathcal{T}$ and $v \notin \mathcal{T}$.
	As $uv \notin E(H)$, we have the constraint $x_u + 1 \leq \Dist{u - v} - \varepsilon$ in our linear program, which $x_u$ must fulfill.
	Since $\varepsilon > 0$ by assumption, this constraint actually represents $x_u + 1 < \Dist{u - v}$.
	This corresponds to $r(u) + 1 < \Dist{u - v}$, which contradicts the existence of the edge $uv$ in $G(\mathcal{S}, r)$; recall $v \notin \mathcal{T}$.
	The case $u \notin \mathcal{T}$ and $v \in \mathcal{T}$ is symmetric and for the case $u \in \mathcal{T}$ and $v \in \mathcal{T}$ we can use the constraint $x_u + x_v \leq \Dist{u - v} - \varepsilon$ to arrive at a contradiction.
	Since all cases lead to a contradiction, we conclude that $G(\mathcal{S}, r) = H$ must hold, i.e., $r$ is a solution to \instance.
\end{prooflater}

\NewText{Note that \Cref{lem:linear-program-correctness} in particular guarantees the existence of a solution when $\varepsilon$ is unbounded in the linear program $\mathcal{P}(\instance)$.}

State-of-the-art linear programming solvers are polynomial in the number of variables, constraints, \emph{and} encoding size of the involved numbers~\cite{JSWZ.fas.2021,Ren.pta.1988}. %
\NewText{We cannot obtain a bound on the encoding size of the numbers involved in $\mathcal{P}(\instance)$}. %
However, it is known that linear programs with $\eta$ variables and $m$ constraints can be optimized in $2^{\BigO{\eta\cdot\log(\eta)}} \cdot m$ 
time (in the Real RAM model)~\cite{CHAN18}.
Since the linear program~$\mathcal{P}(\instance)$ has $\Size{\mathcal{T}}+1$ variables and $\BigO{n^2}$ constraints, we get:

\begin{proposition}
	\label{prop:radius-feasibility-running-time}
	\RadiusFeasibility can be solved in $2^{\BigO{\Size{\mathcal{T}}\cdot\log(\Size{\mathcal{T}})}}\cdot n^{\BigO{1}}$
	time.
\end{proposition}

\subsection{The \textsf{XP}-Membership Framework}
We build upon \Cref{prop:radius-feasibility-running-time} and show that \PiScaling is in \XP\ with respect to $k$ for every graph class $\Pi$ for which \PiRecognition can be solved in polynomial time.

The main idea is to 
provide a Turing-reduction to %
$n^{f(k)}$ instances of %
\RadiusFeasibility.
To construct these instances, we
exhaustively consider all possibilities to construct the set~$\mathcal{T}$ of at most $k$ scaled disks and apply the following steps separately for each constructed~$\mathcal{T}$\onlyLong{---in the remainder of the paper we refer to this operation as \emph{branching} or \emph{guessing}}. 
Observe that if the radius of the modified disks would be fixed in the input, then the graph~$H$ that we obtain after scaling is solely determined by our choice of $\mathcal{T}$.
However, since each modified disk can choose a radius from $[\rmin, \rmax]$, we must further branch to determine which target graph $H$ we want to obtain.
To this end, for each possible set $\mathcal{T}$, we
determine for the scaled disks their (possible) neighborhood in $H$ by
\begin{figure}
	\centering
	\includegraphics[page=1]{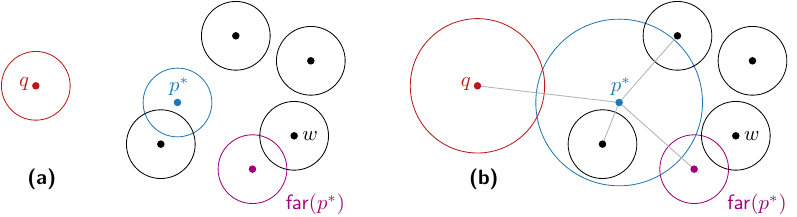}
	\caption{Let $\mathcal{T} = \{p^*,q\}$. The furthest unscaled neighbor $\far(p^*)$ (purple) of $p^*$ (blue) determines if $p^*p' \in E(H)$ for $p' \in \mathcal{S} \setminus \mathcal{T}$. For the edge to $q$ (red), we branch to fix if it exists (as here) or not. The disk graph \textbf{\textsf{(a)}} before and \textbf{\textsf{(b)}} after the scaling operations.}
	\label{fig:xp-framework-neighborhood}
\end{figure}
guessing for each $p^* \in \mathcal{T}$ its furthest unscaled neighbor $\far(p^*)$, i.e., $\far(p^*) = \argmax_{p \in (\mathcal{S} \setminus \mathcal{T}) \cap N_H(p^*)} \Dist{p - p^*}$ (if it exists).
Observe that the neighborhood of $p^*$ in the final disk graph cannot be arbitrary.
In particular, as we increase $r(p^*)$, the set of (unscaled) neighbors of $p^*$ can only increase.
Thus, $\far(p^*)$ determines all unscaled neighbors of~$p^*$:
For $p' \in \mathcal{S} \setminus \mathcal{T}$ we have $p^*p' \in E(H)$ if and only if $\Dist{p^* - p'} \leq \Dist{p^* - \far(p^*)}$; see \Cref{fig:xp-framework-neighborhood}.
Finally, we guess the adjacencies among scaled disks.
Note that this captures all relevant, i.e., possible, graphs $H$ we could obtain by scaling the disks in $\mathcal{T}$.
In particular,
each of the $2^{\BigO{k^2}}\cdot n^{\BigO{k}}$ branches corresponds to an instance of \RadiusFeasibility. %
We discard branches with $H \notin \Pi$, which can be checked efficiently since \PiRecognition can be solved in polynomial time.
\begin{restatable}\restateref{thm:xp-framework}{theorem}{theoremXPFramework}
	\label{thm:xp-framework}
	Let $\Pi$ be a graph class such that \probname{$\Pi$-Recognition} can be solved in polynomial time.
	Then, $\PiScaling$ can be solved in $2^{\BigO{k^2}}\cdot n^{\BigO{k}}$ time. %
\end{restatable}
\begin{prooflater}{ptheoremXPFramework}
	Consider a graph class $\Pi$ such that \PiRecognition can be solved in polynomial time and let $\instance = \instanceLong$ be an instance of \PiScaling.
	We now provide a Turing reduction from \instance to $n^{f(k)}$-many instances $\instanceLongRadius$ of \RadiusFeasibility, each representing a different choice for $\mathcal{T}$ and $H$.
	This reduction consists of three major steps.
	
	In the first step, we branch on the $n^k$ different sets $\mathcal{T} \subseteq \mathcal{S}$ such that $\Size{\mathcal{T}} \leq k$.
	If \instance is a yes-instance, there exists a solution $r$ and at least one of the branches will correspond to $\mathcal{T} = \mathcal{T}(r)$. 
	In the following, let us consider one of these branches, i.e., one fixed set $\mathcal{T}$.
	What remains from defining $\instanceLongRadius$ is the graph $H$.
	Steps two and three will take care of this.
	
	In the second step, we consider the disks in $\mathcal{T}$ and determine their neighborhood with respect to the unscaled disks $\mathcal{S} \setminus \mathcal{T}$.
	To this end, consider a hypothetical graph $H$.
	For each disk $p \in \mathcal{T}$ we define $\far(p) \coloneqq \argmax_{q \in (\mathcal{S} \setminus \mathcal{T}) \cap N_H(p)} \Dist{p - q}$.
	In case of ties, we break them arbitrarily.
	It could be that the set is empty, i.e., that we have $pq \notin E(H)$ for every disk $q \in \mathcal{S} \setminus \mathcal{T}$.
	In other words, in a solution $r$, $p$ either has no incident edge or is only adjacent to scaled disks in $G(\mathcal{S}, r)$.
	In such cases, we set $\far(p) = \nil$.
	We now make the following observation.
	For every unscaled disk $q \in \mathcal{S} \setminus \mathcal{T}$ we have $pq \in E(H)$ if and only if $\far(p) \neq \nil$ and $\Dist{p-q} \leq \Dist{p - \far(p)}$.
	This is because, if $\far(p) \neq \nil$, we have $r(p) \geq \Dist{p - \far(p)} - 1 \geq \Dist{p-q} - 1$ since $\far(p)$ is the furthest unscaled neighbor of~$p$.	
	In particular, observe that any edge $pq$ with $\Dist{p-q} > \Dist{p - \far(p)}$ would contradict the choice of $\far(p)$.
	Thus, once we have $\far(p)$, the set $N_H(p) \cap (\mathcal{S} \setminus \mathcal{T})$, i.e., the unscaled neighbors of $p$ in $H$, is completely determined.
	Consequently, we branch for each such $p \in \mathcal{T}$ into at most $n$ possible choices for $\NewText{\far(p) \in \mathcal{S} \setminus \mathcal{T}}$.
	Together, this yields $n^k$ different branches for all disks in $\mathcal{T}$.
	By above observation, the branching step is correct.
	
	At this point, what is missing from the construction of $H$ are edges (or non-edges) between two disks from $\mathcal{T}$.
	In the third step, we branch to determine these edges.
	Recall that $\Size{\mathcal{T}} \leq k$.
	Hence, there are at most $2^{\BigO{k^2}}$ different branches.
	Since we perform exhaustive branching, correctness is maintained.
	
	Combining Steps 1--3, we obtain $n^k \cdot n^k \cdot 2^{\BigO{k^2}} \in 2^{\BigO{k^2}} \cdot n^{2k}$ different branches, each corresponding to a different instance $\instance' = \instanceLongRadius$ of \RadiusFeasibility.
	In each branch, we check if $H \in \Pi$ holds.
	As \PiRecognition can be solved in polynomial time, this takes $n^{\BigO{1}}$ time, and we reject the branch if the check fails.
	Otherwise, it only remains to check if for the particular choice of $\mathcal{T}$ and $H$, there exists a radius assignment $r$ such that $\rmin \leq r(p) \leq \rmax$ if $p \in \mathcal{T}$ and $G(\mathcal{S}, r) = H$.
	Note that if such an $r$ exists, it witnesses that~\instance is a yes-instance of \PiScaling.
	By \Cref{prop:radius-feasibility-running-time}, we can check for the existence of~$r$ in $2^{\BigO{k \cdot \log(k)}}\cdot n^{\BigO{1}}$ time and combining everything, the theorem follows.
\end{prooflater}

\shortLong{
	Note that we only require the polynomial-time algorithm for \PiRecognition to efficiently verify that $H \in \Pi$ holds.
	Relaxing this restriction leads to different flavors of our meta-result:%
}{%
	A closer analysis of the proof of \Cref{thm:xp-framework} reveals that the existence of a polynomial-time algorithm for \PiRecognition is only required to efficiently verify that $H \in \Pi$ holds for the graph $H$ that we obtain in our branching-procedure.
	\NewText{In fact, we can replace the polynomial-time algorithm for \PiRecognition with an algorithm that runs in time $f(\kappa) \cdot n^{\BigO{1}}$ or $\BigO{n^{f(\kappa)}}$, for some computable function $f$, i.e., with an \FPT- or \XP-algorithm with respect to some parameter $\kappa$, respectively.
	This results in an additional multiplicative factor in the running time.
	More concretely, in both scenarios we can bound the overall running time by $2^{\BigO{k^2}}\cdot n^{\BigO{k + f(\kappa)}}$.
	Hence, by relaxing the prerequisite on \PiRecognition, we can obtain a weaker version of our meta-result, as summarized in the following remark.
	}
}%
\begin{remark}
	\label{remark:xp-framework}
	Let $\Pi$ be a graph class such that \probname{$\Pi$-Recognition} is \FPT\ or in \XP\ with respect to some parameter $\kappa$.
	Then, $\PiScaling$ is in \XP\ with respect to $k + \kappa$.
\end{remark}

\onlyLong{
	While \Cref{thm:xp-framework} (and \Cref{remark:xp-framework}) provide(s) us with an upper bound for many graph classes $\Pi$, it is, of course, a priori not clear if the obtained \XP-running time is the best we can achieve for specific choices of $\Pi$.
	In the upcoming sections, we take a closer look at \PiScaling for different graph classes $\Pi$ and examine if \Cref{thm:xp-framework} can be improved to a fixed-parameter or even polynomial-time algorithm. 
}

\newcommand{\showClusterEnlargeEdgeFigure}{
	\begin{figure}
		\centering
		\includegraphics[page=1]{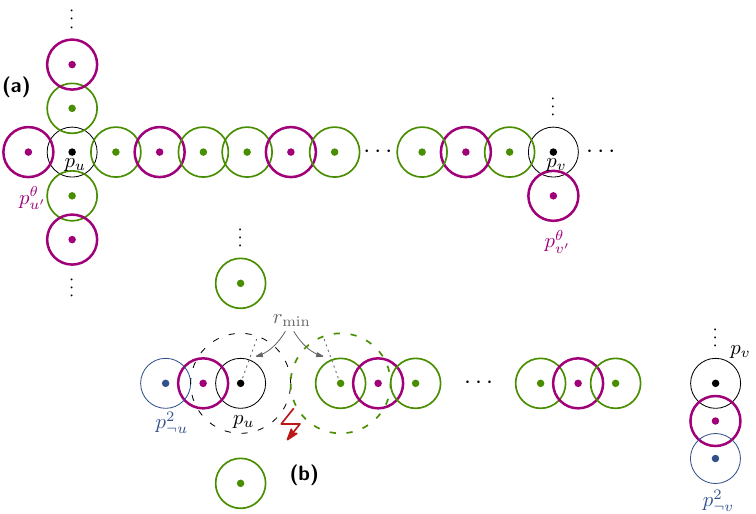}
		\caption{
			The edge $uv$ for \textbf{\textsf{(a)}} $\rmin > 1$ and \textbf{\textsf{(b)}} $\rmin < 1$. Black, blue, green, and purple disks are 1-, 2-, $\eta^2$-, and $\theta$-heavy, respectively.
			In \textbf{\textsf{(a)}}, if we enlarge $p_u$ and the left endpoint of a heavy $P_3$, then two $\theta$-heavy disks are in the same connected component but not adjacent as $\theta > k$.
		}
		\label{fig:cluster-hardness-enlarge-shrink-edge}
	\end{figure}
}

\newcommand{\showClusterHardnessExample}{
	\begin{figure}
		\centering
		\includegraphics[page=1]{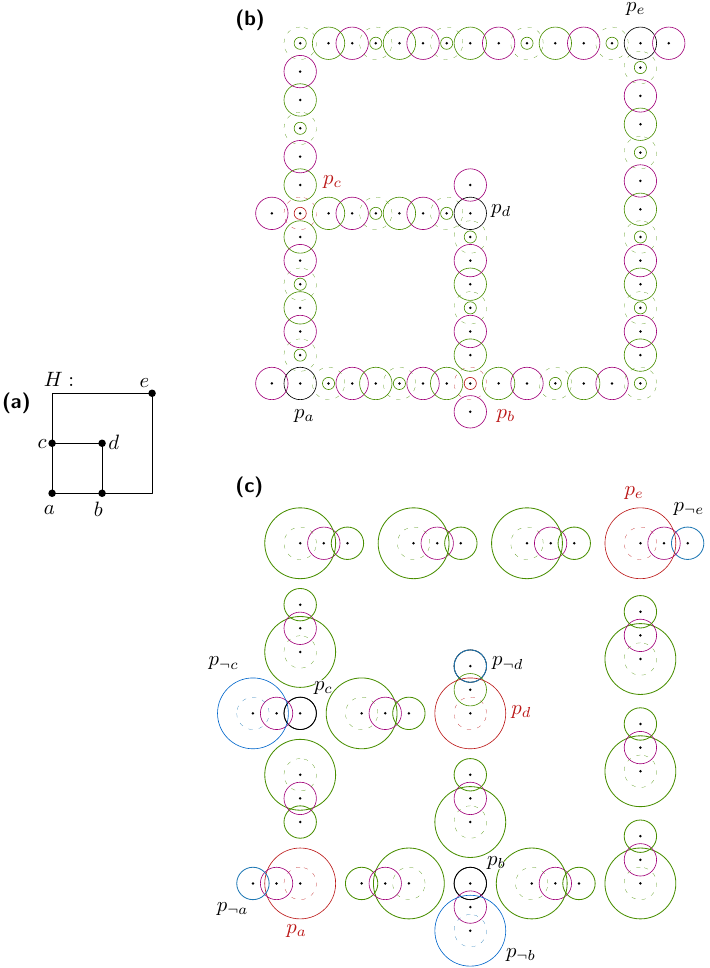}
		\caption{
			\textbf{\textsf{(a)}} A graph $H$ and the corresponding instances of \Cluster for \textbf{\textsf{(b)}} $\rmin < 1$ and \textbf{\textsf{(c)}} $\rmin > 1$. 
			If a disk representing a vertex is red, then the vertex belongs to a solution to the instance \textbf{\textsf{(b)}} $(H, 2)$ of \VC and \textbf{\textsf{(c)}} $(H, 3)$ of \IS.
			Note that we can always slightly shift some disks to make space along an edge of $H$ if required.
			We simplified $H$ and the constructed instances to ease readability.
			Colors have the same meaning as in \Cref{fig:cluster-hardness-enlarge-shrink-edge}.
		}
		\label{fig:cluster-hardness-example-combined}
	\end{figure}
}

\section{A Fixed-Parameter Tractable Algorithm for \textsc{Scaling To Cluster}}
\label{sec:cluster-fpt}
We now consider \NewText{the setting where $\Pi$ is the class of all cluster graphs} %
and show that \Cluster is \FPT\ parameterized by the \shortLong{budget $k$}{number $k$ of allowed scaling operations}.
This improves the \XP-algorithm implied by \Cref{thm:xp-framework}\shortLong{ since cluster graphs can be recognized \shortLong{efficiently}{in polynomial time} due to their characterization as being $P_3$-free}{; observe that cluster graphs can be recognized in polynomial time thanks to their characterization as being $P_3$-free}.

Our algorithm follows the overall strategy of the general \XP-approach behind \Cref{thm:xp-framework}, but uses structural properties of cluster graphs to obtain an \FPT-running time.
We provide in \Cref{fig:cluster-overview} an overview of the algorithm.
Let us fix an instance $\instance = \instanceLong$ of \Cluster.
Consider a hypothetical solution $r$.
The algorithm has two main phases.
In the first phase, we determine the disks in $\mathcal{T}(r)$\onlyLong{, i.e., Step~1 of the algorithm behind \Cref{thm:xp-framework}}.
For each such disk~$p\in\mathcal{T}(r)$, we also fix two, in some sense ``extremal'', unscaled disks that allow us to deduce the membership of $p$ in clusters from $G(\mathcal{S}, r)$:
We determine (a) the furthest \emph{unscaled} disk $\far(p)$ which is in the \emph{same} cluster as~$p$ in $G(\mathcal{S},r)$, and (b) the closest \emph{unscaled} disk $\clo(p)$ which is in a \emph{different} cluster than~$p$ in $G(\mathcal{S},r)$\onlyLong{.
	This resembles Step~2 of the algorithm behind \Cref{thm:xp-framework}}; see also \Cref{fig:cluster-overview}b--d.
With this, we can almost completely fix the neighborhood of $p$ in $G(\mathcal{S},r)$ via %
\FPT-many branches %
thanks to the structure of cluster graphs and the geometric aspect of the problem.
In Phase~2, we resolve the (non-) edges between the at most~$k$ disks in $\mathcal{T}(r)$.
\shortLong{Na{\"i}ve branching}{In principle, we could use the na{\"i}ve branching algorithm for Step~3 of our \XP-algorithm presented in \Cref{thm:xp-framework}, but this} would result in an overall running time of~$2^{\OO(k^2)}\cdot n^{\OO(1)}$.
We argue that for cluster graphs, this step can be done in $2^{\OO(k\cdot \log(k))}\cdot n^{\OO(1)}$~time.
Finally, \onlyLong{analogously to Step~3 of the algorithm behind \Cref{thm:xp-framework}, }we employ the LP from \Cref{sec:frameworks} to check if suitable radii for the scaled disks exist. %
\shortLong{We now give an overview of the phases; see \Cref{app:cluster-fpt} for details.}{In the following, we first provide in \Cref{sec:cluster-fpt-intuition} a more detailed overview of the two phases before we turn in \Cref{sec:cluster-fpt-algorithm} to a detailed algorithmic description of them. Finally, in \Cref{sec:cluster-fpt-correctness}, we establish correctness of our algorithm and bound its running time.}

\onlyLong{	
	\subsection{The Intuition of the Two Phases}
	\label{sec:cluster-fpt-intuition}
	We now present a more detailed overview of the two phases of our algorithm.
	Without loss of generality, we assume that $\Size{\mathcal{S}} > k$.
	Otherwise, we replace the following Phase~1 by %
	simply guessing the set of scaled disks, which is asymptotically faster than %
	Phase~1. %

}

\subparagraph*{\onlyShort{Intuition of }Phase 1.}
In this phase, we aim to determine a set~$\mathcal{T}$ of scaled disks and, for each disk~$p\in \mathcal{T}$, the (a) furthest \emph{unscaled} disk~$\far(p)$ which is in the same cluster as~$p$ in the solution and the (b) closest \emph{unscaled} disk~$\clo(p)$ which is not in the same cluster as~$p$ in the solution. 
Initially, we %
compute in polynomial time a maximal vertex-disjoint~$P_3$-packing~$\mathcal{P}$ of~$G(\mathcal{S}, \Unit)$, that is, $\mathcal{P}$ is a set of vertex-disjoint~$P_3$s of~$G(\mathcal{S}, \Unit)$ such that no further vertex-disjoint~$P_3$ of~$G(\mathcal{S}, \Unit)$ can be added to~$\mathcal{P}$.  
By~$P$ we denote the set of points in~$\mathcal{S}$ whose corresponding disks are contained in at least one~$P_3 \in \mathcal{P}$.
Since every~$P_3$ requires at least one scaling operation, we can safely output no if~$|\mathcal{P}|\ge k+1$.
Moreover, $G':=G((\mathcal{S}\setminus P), \Unit)$ is a cluster graph. %
Next, we create a complete edge-colored auxiliary graph~$H$ on the entire vertex set~$\mathcal{S}$. %
The colors of~$H$ model the connectivity after some scaling operations and represent edges (blue), non-edges (red), and (still) undefined vertex-pairs (green) between two scaled disks that we resolve in Phase~2; see also \Cref{fig:cluster-overview}.
More precisely, let~$\mathcal{T}_i$ denote the set of the first $i$~disks we identified for scaling (hence~$\mathcal{T}_0=\emptyset$).
The graph~$H_i$ represents the connectivity of the disk graph resulting from~$G(\mathcal{S}, \Unit)$ if~$\mathcal{T}_i$ is scaled accordingly; 
recall that we only compute the precise radii of the scaled disks at the end and now we only determine the final cluster graph~$H$, which will represent the abstract intersection graph.
The graph~$H_i$ is helpful to identify new~$P_3$s after we have already scaled the first $i$~disks~$\mathcal{T}_i$.
In particular, whenever~$H_i$ contains a \emph{fully-defined} $P_3$ $X$, i.e., two edges (blue edges in~$H_i$), one non-edge (red edge in~$H_i$), and no undefined vertex-pair in $V(X)$ (green edge in~$H_i$), we branch to determine which of the involved disks we (additionally) need to scale (and thus add to $\mathcal{T}_{i + 1}$).
Existing $P_3$s containing undefined (i.e., green) vertex-pairs will be treated at the end.

We determine~$\far(p)$ and~$\clo(p)$ via branching but exploit the structure of cluster graphs and the geometric aspect of the problem to bound the number of branches in $k$.
To obtain $\far(p)$,
let~$W$ be the $k$~closest disks to~$p$ (excluding~$w$ itself).
Due to our budget~$k$, at least one disk of~$W$ is unscaled.
Let~$\mathcal{C}$ be the set of at most $k$~clusters of~$G'$ that contain the disks in~$W$.
If the  cluster in the solution which contains~$p$ also contains disks of a cluster~$C$ of~$G'$ that ``merges'' into it, we must have~$C\in\mathcal{C}$.
There are three possibilities for the disk $\far(p)$: (a) $\far(p)$ does not exist, (b) $\far(p)$ is an unscaled disk from the $P_3$-packing $\mathcal{P}$, or (c)
$\far(p)$ is among the $k$ furthest disks in one of the clusters of~$\mathcal{C}$ due to the budget $k$.
This way, we can bound the overall number of choices for~$\far(p)$ in $\BigO{k^2}$.
Next, after we determined~$\far(p)$, we branch on the $k$~closest disks to~$p$ which are further away than~$\far(p)$ from~$p$ to determine~$\clo(p)$ (we also consider the case that~$\clo(p)$ does not exist).
After fixing~$\far(p)$ and~$\clo(p)$, there may exist disks~$q$ with~$\Dist{\far(p)-p}<\Dist{q-p}<\Dist{\clo(p)-p}$.
Since $\far(p)$ and~$\clo(p)$ is the furthest/closest \emph{unscaled} disk which is in/not in the same cluster as~$p$ in the solution, each such disk~$q$ needs to be scaled as well; see \Cref{fig:cluster-overview}b--d.
Note that due to our budget these can be at most $k$~disks.
We maintain a set~$N_i$ that keeps track of the disks that require scaling due to the above reasoning and update~$H$ accordingly.
Whenever more than $k$~scaling operations are required, i.e., $\Size{\mathcal{T}_i \cup N_i} > k$, we abort the branch.

\subparagraph*{\onlyShort{Intuition of }Phase 2.}
Assume that the set~$\mathcal{T}$ of scaled disks has been correctly determined.
We now resolve the status of vertex-pairs where both endpoints are scaled disks, i.e., the undefined edges, in $2^{\OO(k\cdot \log(k))}\cdot n^{\OO(1)}$ time, %
and distinguish between scaled disks~$p$ with for which~$\far(p)$ exists and for which it does not exist.
The latter implies that~$p$ is in a cluster with scaled disks only.
All undefined vertex-pairs incident to disks for which~$\far(p)$ does not exist can be resolved in polynomial time. %
For the remaining scaled disks, we guess a clustering of them that also resolves all incident undefined vertex-pairs. %
Hence, the final %
auxiliary graph~$H$ %
contains no undefined edges.
If at some stage of this process, graph~$H$ does not correspond to a cluster graph, we can abort.
Otherwise, we use the LP from \Cref{sec:linear-programming} to check if suitable radii for the selected disks exist.
\onlyShort{Combining all, we obtain:}

\begin{figure}
	\centering
	\includegraphics[page=1]{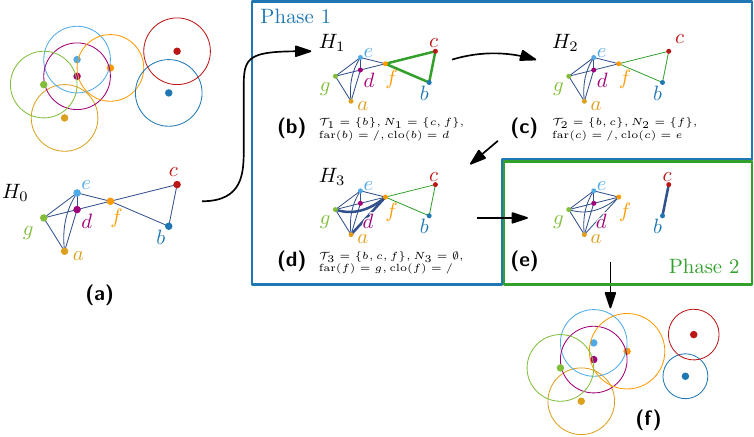}
	\caption{
		Overview of our \FPT-algorithm for \Cluster with a small example.
		\textbf{\textsf{(a)}}~An input instance and the corresponding edge-colored complete graph $H_0$ (we omit red edges for clarity).
		\textbf{\textsf{(b)}}--\textbf{\textsf{(d)}} In Phase 1, we iteratively identify disks to scale (the sets $\mathcal{T}_i$) and refine our edge-colored complete graph $H_i$ (changes highlighted in bold).
		\textbf{\textsf{(e)}} In Phase 2, we branch to fix a compatible cluster graph $H$, and check using \Cref{prop:radius-feasibility-running-time} if it can be realized.
		\textbf{\textsf{(f)}} The solution to \textbf{\textsf{(a)}}.
	}
	\label{fig:cluster-overview}
\end{figure}

\onlyLong{It remains to present the details of our algorithm and prove its correctness.}

\begin{statelater}{clusterFPTDetails}
	\subsection{The Algorithmic Description of the Two Phases}
	\label{sec:cluster-fpt-algorithm}
	In \shortLong{\Cref{sec:cluster-fpt}}{\Cref{sec:cluster-fpt-intuition}}, we gave an intuitive description of the two phases.
	Recall that we can safely assume that our instance has more than $k$~disks since otherwise the instance can be solved in the desired running time via brute force. 
	We now proceed with a more formal, algorithmic description of them.
	However, before we can do that, we first have to describe the underlying data structures.
	In the following, we define $G\coloneqq G(\mathcal{S}, \Unit)$ to be the unit disk graph of the instance \instance.
	
	\subparagraph*{Data Structures.}
	First, $\mathcal{T}_i$ is the set of disks which are scaled after $i$~steps of our algorithm.
	Thus, $|\mathcal{T}_i|=i$.
	Second, $F_i$ is a set of disks which is guaranteed to be not part of a solution after $i$~steps of our algorithm.
	More specifically, $F_i$ contains the disks~$\far(p)$ and~$\clo(p)$ which we determined as unscaled.
	Note that each disk contained in~$F_i$ is also contained in any~$F_j$ where~$j\ge i$, that is, such a disk will never be scaled (i.e., we have $F_i \subseteq F_j$). 
	Third, we use a set~$N_i$ of disks which still have to be scaled.
	Initially, $F_0=\emptyset=N_0$.
	More precisely, $N_i$ contains all disks~$q$ which are not yet scaled, but~$\Dist{\far(p)-p}<\Dist{q-p}<\Dist{\clo(p)-p}$ for some disk~$p$ which we already scaled, that is, $p\in \mathcal{T}_i$.
	Moreover, we use a set of auxiliary complete edge-colored graphs~$H_i$ to represent the adjacencies after step~$i$ of the algorithm, that is, after the disks in~$\mathcal{T}_i$ are scaled.
	More precisely, the vertex set of~$H_i$ represents the point set~$\mathcal{S}$ and there are three colors:
	\textbf{(1)} color \emph{red} represents a non-edge, 
	\textbf{(2)} color \emph{blue} represents an edge, and
	\textbf{(3)} color \emph{green} represents an undefined vertex-pair. %
	By~$\col_i(u,w)$ we denote the color of the vertex-pair~$u,w$ of~$H_i$.
	We may drop the index~$i$ if it is clear from the context.  
	Initially, before we start branching, $H_0$ corresponds to~$G$, that is, for any two points~$u,w\in\mathcal{S}$ we have~$\col(u,w)=\blue$ if~$\Dist{u-w}\le 2$, and~$\col(u,w)=\red$ if~$\Dist{u-w}> 2$.
	Thus, initially no edge of~$H_0$ is~$\green$, and thus undefined.
	We use color~$\green$ only during our algorithm for all vertex-pairs were both endpoints are from~$\mathcal{T}_i\cup N_i$.
	All $\green$~edges are resolved in Phase~2 after all disks for scaling have been determined.
	Three vertices~$u,v,w\in V(H_i)$ are a \emph{colorful~$P_3$} if~$\col_i(u,v)=\blue=\col_i(u,w)$, and~$\col_i(v,w)=\red$.
	Note since every edge in~$H_i$ where both endpoints are from~$\mathcal{T}_i\cup N_i$ is~\green, we have~$|\{u,v,w\}\cap (\mathcal{T}_i\cup N_i)|\le 1$, that is, at most one vertex of any colorful~$P_3$ of~$H_i$ is from~$\mathcal{T}_i\cup N_i$.

	\subparagraph*{Algorithmic Description of Phase~1.}
	In polynomial time, we compute a maximal $P_3$-packing~$\mathcal{P}$ of~$G$ and by~$P$ we denote the disks contained in~$\mathcal{P}$.
	If~$|\mathcal{P}|\ge k+1$, we face a no-instance.
	Consequently, we have~$|P|\le 3k$.
	Let~$G'$ be the cluster graph which is the result of removing the disks in~$P$ from~$G$. 
	Let~$\mathcal{T}_i, F_i$, and~$N_i$ be the currently scaled, forbidden, and not-yet-scaled disks, respectively.
	If~$|N_i|\ge 1$, or~$H_i$ contains a colorful~$P_3$, then we invoke our branching (see below), and otherwise we end Phase~1 and start Phase~2 (recall that initially~$\mathcal{T}_0=N_0=F_0=\emptyset$).
	
	We begin by describing our branching procedure.
	If~$|N_i\cup \mathcal{T}_i|>k$, we abort this branch.
	In the first step, we determine a disk which gets scaled.
	There are two cases: (a) $N_i\ne\emptyset$, then we pick an arbitrary disk~$w\in N_i$.
	Otherwise, (b) if~$N_i=\emptyset$, let~$Y$ be a colorful~$P_3$ of~$H_i$.
	Let~$Z:=V(Y)\setminus (\mathcal{T}_i\cup F_i)$.
	Note that~$Z\cap N_i=\emptyset$ since~$N_i=\emptyset$.
	If~$Z=\emptyset$, we abort this branch, and otherwise, we guess a vertex of~$Z$, say~$w$, which gets scaled.
	
	In the second step, we determine the edges of the solution which are incident with the selected disk~$w$ we aim to scale.
	More precisely, our aim is to determine a furthest unscaled disk~$\far(w)\in\mathcal{S}\setminus(\mathcal{T}_i\cup N_i)$ from~$w$ such that in the solution~$w$ and~$\far(w)$ are adjacent, and also a closest unscaled disk~$\clo(w)$ in~$\mathcal{S}\setminus(\mathcal{T}_i\cup N_i)$ from~$w$ such that in the solution~$w$ and~$\clo(w)$ are not adjacent. 
	Consequently, $\far(w)$ yields a lower bound for the new radius~$r(w)$, and~$\clo(w)$ yields an upper bound for~$r(w)$.
	In other words, we have~$\Dist{w-\far(w)}-1\le r_w < \Dist{w-\clo(w)}-1$.
	The precise value of~$r(w)$, however, is not yet clear due to interactions with other (so far not) scaled disks. 
	These radii will be determined at the end of Phase~2 via the LP from \Cref{sec:frameworks}.

	We first describe our branching to obtain~$\far(w)$.
	In principle, there are 3 possibilities:
	\begin{enumerate}
		\item $\far(w)$ does not exist,
		\item $\far(w)\in P$, and
		\item $\far(w)$ is contained in some cluster of~$G'$. 
	\end{enumerate}
	
	We now provide a more detailed description of the branching step.
		To do so, we need the following notation:
		For each cluster~$C$ of~$G'$ (the unit disk graph~$G'=G((\mathcal{S}\setminus P),\Unit)$ obtained after removing each unit disk in the $P_3$-packing~$\mathcal{P}$), we let~$\dist_{G'}(C,w)=\min_{p\in C\setminus{w}}\Dist{p-w}$ denote the \emph{smallest distance of any point except~$w$ corresponding to a vertex in~$C$ to~$w$}.
		We now order all cluster of~$G'$ according to~$\dist_{G'}$ and by~$\mathcal{Y}$ we denote the set of $k$ closest clusters to~$w$.
		Note that if~$w$ is its unique cluster, then cluster~$\{w\}$ is \emph{not} part of~$\mathcal{Y}$.

	\begin{branching}
		\label{branching-for-far(w)}
		We branch on the following options.
		\begin{enumerate}
			\item $\far(w)$ does not exist.
			\item $\far(w)\in P$.
			\item For each of the $k$ closest clusters~$\mathcal{Y}$ to $w$, consider the $k$ furthest disks to $w$ and branch among these $k^2$ disks.
		\end{enumerate}
	\end{branching}

	Note that \textbf{1.} represents the case that the cluster of the solution which contains~$w$ does only contain scaled disks.

	After fixing~$\far(w)$ due to the above branching, we now branch to determine~$\clo(w)$ as follows.
	We use the following notation: 
	If~$\far(w)$ exists, we set~$d_{\far(w)}:=\Dist{w-\far(w)}$, and otherwise, if~$\far(w)$ does not exist, we set~$d_{\far(w)}=0$. 
	
	\begin{branching}
		\label{branching-for-clo(w)}
		We branch on the following options.
		
		\begin{enumerate}
			\item $\clo(w)$ does not exist.
			\item 
			We sort all disks~$q$ of~$\mathcal{S}\setminus (\mathcal{T}_i\cup N_i)$ with~$d_{\far(w)}<\Dist{q-w}$, that is, all not yet (or later) scaled disks with sufficiently large distance from~$w$, ascendingly according their distance to~$w$.
			By~$D_j$ we denote the set of disks which all have the $j$-th smallest distance to~$w$.
			Now, let~$\ell$ be minimal such that~$\sum_{j=1}^{\ell}|D_j|\ge k$.
			If~$\sum_{j=1}^{\ell}|D_j|> k$, we redefine~$D_\ell$ as some arbitrary subset of~$D_\ell$ such that~$\sum_{j=1}^{\ell}|D_j|= k$.
			Let~$Q:=\bigcup_{j=1}^\ell D_j$.
			Now, branch on~$Q$.
		\end{enumerate}	
	\end{branching}

	Note that the non-existence of~$\clo(w)$ implies that the cluster of the solution which contains~$w$, contains all unscaled disks. 
	Consequently, at least one of~$\far(w)$ and~$\clo(w)$ exists, because otherwise our instance would have at most $k$~disks (which we assume is not the case).
	We use the following notation:
	If~$\clo(w)$ exists, we set~$d_{\clo(w)}:=\Dist{w-\clo(w)}$, and otherwise, if~$\clo(w)$ does not exist, we set~$d_{\clo(w)}=\infty$.
	 
	After determining~$\far(w)$ and~$\clo(w)$, we update our data structures as follows:
	We set~$\mathcal{T}_{i+1}:= \mathcal{T}_i\cup\{w\}$, and~$F_{i+1}:=F_i\cup\{\far(w), \clo(w)\}$.
	Let~$X=\emptyset$.
	To obtain~$H_{i+1}$ and~$N_{i+1}$ we do the following.
	
	\begin{drule}
		\label{rule-compute-H-i+1}
		To obtain~$H_{i+1}$, for each disk~$x\in \mathcal{S}\setminus (\mathcal{T}_{i+1}\cup N_i)$, we recolor~$xw$ as follows:
		
		\textbf{(1)} if $\Dist{w-x}\le d_{\far(w)}$, then we make~$xw$ \blue{} in~$H_{i+1}$, 
		
		\textbf{(2)} if $\Dist{w-x}\ge d_{\clo(w)}$, then we make~$xw$ \red{} in~$H_{i+1}$, and 
		
		\textbf{(3)} otherwise, we make~$xw$ \green{} in~$H_{i+1}$ and we add~$x$ to~$X$.\\
		
		Afterwards, we perform the following steps.
		
		\textbf{(a)}  We check if~$X\cap F_{i+1}\ne\emptyset$.
		If yes, we abort this branch. 
		
		\textbf{(b)}  We next set $N_{i+1}:=(N_i\setminus\{w\})\cup X$.
		
		\textbf{(c)}  Finally, we recolor every edge~$p_1p_2$ of~$H_i$ with~$p_1,p_2\in \mathcal{T}_{i+1}\cup N_{i+1}$ as \green{} in~$H_{i+1}$.
	\end{drule}	
	
	The last step (Step~(c)) in \Cref{rule-compute-H-i+1} is done to ensure that the status of each vertex-pair where both vertices are from~$\mathcal{T}_{i+1}\cup N_{i+1}$ is (still) undefined (this will be settled in Phase~2).

	\subparagraph*{Algorithmic Description of Phase~2.}
	If $H_i$ does not contain any colorful~$P_3$ anymore, $N_i=\emptyset$, and the algorithm did not output no yet, we next settle the status of all \green{} edges of~$H_i$, that is, all undefined edges of~$H_i$.
	From now on we use~$H$ as a shorthand for~$H_i$, and analogously~$\mathcal{T}, N, F$ instead of~$\mathcal{T}_i, N_i, F_i$, where $i$ is the index of the last iteration.
	Observe that since~$N_i=\emptyset$, a vertex-pair~$uw$ is \green{} in~$H$ if and only if~$u$ and~$w$ are scaled disks.
	We let $B(w)$ denote the set of \emph{\blue~neighbors}~$\{u\in N_H(w): uw \text{ is \blue{} in } H\}$ of~$w$ in~$V(H)$. %
	Let~$X$ be the subset of all scaled disks with no \blue{} neighbors, that is, $B(x)=\emptyset$ for each~$x\in X$.
	Now, consider a pair~$u,w$ of scaled disks such that the set of \blue-neighbors of at least one of them, say~$u$, is not empty, and thus $B(w)=\emptyset$ or~$B(w)\ne\emptyset$ are possible.
	This implies that $u \notin X$.
	The case that $B(u) = B(w) = \emptyset$, i.e., both have no blue neighbors, is handled later.
	We employ the following	 rule once before any edge of~$H$ is re-colored.

	\begin{drule}
		\label{rule-phase-2}
		Let~$u\in \mathcal{T}\setminus X$, that is, $B(u)\ne\emptyset$, and~$w\in \mathcal{T}$.
		We do the following.
		
		\textbf{Case 1:~$B(u)=B(w)$.} Then, $uw$ is re-colored \blue.
		
		\textbf{Case 2:~$B(u)\ne B(w)$ and~$B(u)\cap B(w)=\emptyset$.} Then, $uw$ is re-colored \red.
		
		\textbf{Case 3:~$B(u)\ne B(w)$ and~$B(u)\cap B(w)\ne\emptyset$.} In \Cref{claim-handling-green-correct}, we argue that this case is not possible.
	\end{drule}
	
	It remains to handle the vertices in~$X$, the set of scaled disks for which~$B(x)=\emptyset$. %
	In other words, in the solution, each disk of~$X$ can only be in clusters with other scaled disks.
	We guess a clustering~$\mathcal{C}$ (with at most $k$~clusters) of the disks of~$X$.
	For each pair~$u,w\in X$, if~$u,w$ are in the same cluster of~$\mathcal{C}$ we recolor~$uw$ as \blue, and otherwise, if~$u,w$ are not in the same cluster of~$\mathcal{C}$ we recolor~$uw$ \red.
	
	After this procedure ends, $H$ does not contain any colorful~$P_3$ anymore and also no \green~edge.
	More precisely, $H$ corresponds to a cluster graph (edges correspond to \blue{} edges of~$H$ and non-edges correspond to \red{} edges of~$H$).
	We now invoke the LP behind \Cref{sec:frameworks} to check whether we can find suitable radii for the set~$\mathcal{T}$ of scaled disks which yield the desired cluster graph corresponding to~$H$.
	If any of our branches yields a solution, we output yes, and otherwise, we abort this branch.
\end{statelater}

\onlyLong{
	\subsection{Combining Both Phases: Correctness and Running Time}
	\label{sec:cluster-fpt-correctness}
	\Cref{sec:cluster-fpt-algorithm} provides an algorithmic description of both main phases of our algorithm.
	We proceed by showing their correctness, i.e., that our algorithm reports that \instance is a yes-instance if and only if it admits a solution $r$.
	
}

\begin{statelater}{clusterFPTCorrectness}
	\begin{lemma}
		\label{lem:cluster-fpt-correctness}
		Our algorithm outputs yes if and only if \instance is a yes-instance.
	\end{lemma}
	\begin{proof}
		Clearly, if our algorithm does not output no, we have a solution:
		The final graph~$H$ after Phase~2 is a cluster graph and we then invoke the LP to compute suitable radii. 
		Therefore, we focus on the backward direction.
		Suppose that the input instance of \Cluster has a solution in which the disks~$\widetilde{\mathcal{T}}$ where~$|\widetilde{\mathcal{T}}|=k$ are scaled and let~$r$ be the corresponding radii assignment.\footnote{Without loss of generality, we can assume~$|\widetilde{\mathcal{T}}|=k$ instead of~$|\widetilde{\mathcal{T}}|\le k$ here.} 
		Moreover, let~$\widetilde{G}$ be the resulting cluster graph.
		We now argue that our algorithm is able to detect this solution.
		We start by showing that the output of Phase~1 of our algorithm is correct, that is, we show that after the end of Phase~1 of our algorithm we have $(1)$~$\mathcal{T}=\widetilde{\mathcal{T}}$, where~$\mathcal{T}$ is the final set of scaled disks determined by our algorithm and $(2)$~the final auxiliary complete edge-colored graph~\emph{$H$ is equivalent to~$\widetilde{G}$ up to the \green~vertex-pairs} (\emph{green-equivalent}), that is, if edge~$uw\in E(H)$ is \blue/\red{}, then~$uw$ is an edge/a non-edge in~$\widetilde{G}$.
		Recall that our algorithm computed~$\mathcal{T}$ iteratively via a sequence~$\emptyset=\mathcal{T}_0, \mathcal{T}_1, \ldots, \mathcal{T}_k=\mathcal{T}$.
		More precisely, we argue inductively that
		~\textbf{(a)}~$\mathcal{T}_i\subseteq \widetilde{\mathcal{T}}$ holds, 
		\textbf{(b)} that also~$N_i\subseteq \widetilde{\mathcal{T}}$, 
		\textbf{(c)} that~$F_i\cap\widetilde{\mathcal{T}}=\emptyset$, 
		\textbf{(d)} that~$H_i$ is green-equivalent to~$\widetilde{G}_i$, and
		\textbf{(e)} that the corresponding guesses for~$\far(w)$ and~$\clo(w)$ for each disk in~$\mathcal{T}_i$ are correct.
		Here, $\widetilde{G}_i$ is the graph obtained from the input graph~$G$ after scaling the disks of~$\mathcal{T}_i$ with their corresponding radii in the solution. 
		Afterwards, we show that~$\mathcal{T}_k=\mathcal{T}=\widetilde{\mathcal{T}}$ and that the final auxiliary graph~$H_k=H$ is green-equivalent to~$\widetilde{G}_k=\widetilde{G}$.
		Moreover, also the final set~$N_k$ is empty, and also~$F_k\cap\widetilde{\mathcal{T}}=\emptyset$, as we show.
		
		\proofsubparagraph*{Correctness of properties \textbf{(a)}--\textbf{(e)} of Phase~1 via induction.}
		Initially, when~$i=0$, we have~$\mathcal{T}_0=\emptyset$ and~$G=H_0=\widetilde{G}_0$.
		Hence, our claim holds.
		Thus, in the following, we assume that the statement is true for all values of~$j\in[0,i-1]$, that is, $\mathcal{T}_j\subseteq \widetilde{\mathcal{T}}$, $N_j\subseteq\widetilde{\mathcal{T}}$, $F_j\cap\widetilde{\mathcal{T}}=\emptyset$, and~$H_j$ is green-equivalent to~$\widetilde{G}_j$, and that the computation of~$\far(w)$ and~$\clo(w)$ for each~$w\in \mathcal{T}_j$ was correct.
		We now argue that the same is true for the next index~$i$.
		
		\proofsubparagraph*{Correctness of \textbf{(a)}.}
		Initially, we argue that the choice of the newly disk~$w$ determined by our algorithm is correct, that is, that property~$(a)$~$\mathcal{T}_i\subseteq\widetilde{\mathcal{T}}$ is fulfilled. 
		There are two possibilities for our algorithm for choosing~$w$:
		First, by selecting a disk of~$N_{i-1}$, or second by branching on the three disks of a colorful~$P_3$ of~$H_{i-1}$ (not all three disks have to be suitable candidates since some are already scaled or guessed as unscaled).
		In the first case, that is, $w\in N_{i-1}$, since~$N_{i-1}\subseteq\widetilde{\mathcal{T}}$, we have~$w\in\widetilde{\mathcal{T}}$ and thus also~$\mathcal{T}_i\subseteq\widetilde{\mathcal{T}}$.
		Otherwise, if~$N_{i-1}=\emptyset$ our algorithm searches for a colorful~$P_3$ in~$H_{i-1}$.
		If there is none, the algorithm terminates Phase~1; we will argue below that in this case~$\mathcal{T}_{i-1}=\widetilde{\mathcal{T}}$ needs to hold.
		Otherwise, assume that~$H_{i-1}$ contains a colorful~$P_3$ with vertex set~$W$.
		Since~$H_{i-1}$ and~$\widetilde{G}_{i-1}$ are \green-equivalent, every colorful~$P_3$ of~$H_{i-1}$ corresponds to a~$P_3$ in~$\widetilde{G}_{i-1}$, and since the solution~$\widetilde{G}$ is a cluster graph, that is, contains no~$P_3$, we observe that~$W$ contains at least one disk~$w\in\widetilde{\mathcal{T}}\setminus \mathcal{T}_{i-1}$.
		According to our assumptions we have~$N_{i-1}=\emptyset$ and~$F_{i-1}\cap\widetilde{\mathcal{T}}=\emptyset$.
		Consequently, $w\notin N_{i-1}\cup F_{i-1}$ and thus at least one of the at most three branches on~$W$ leads to~$w$.
		Hence, in both cases a disk~$w\in\widetilde{\mathcal{T}}\setminus \mathcal{T}_{i-1}$ is determined and thus property \textbf{(a)}~$\mathcal{T}_i\subseteq \widetilde{\mathcal{T}}$ is verified. 
		In the following, we denote~$w$ as the unique disk of~$\mathcal{T}_i\setminus \mathcal{T}_{i-1}$.
		
		\proofsubparagraph*{Correctness of remaining properties \textbf{(b)}--\textbf{(e)}.}
		To show the remaining properties~$(b)$-$(e)$, we first argue that the computation of~$\far(w)$ and~$\clo(w)$ is correct, that is, $\far(w)$ is the furthest unscaled disk which is in the same cluster as~$w$ in the solution and~$\clo(w)$ is the closest unscaled disk which is not in the same cluster as~$w$ in the solution, which implies that \textbf{(e)} is correct.
		Recall that at most one of~$\far(w)$ and~$\clo(w)$ cannot exists, because otherwise our instance would have at most $k$~disks.
		Also, recall that the non-existence of~$\far(w)$ implies that the cluster of the solution which contains~$w$ does only contain scaled disks, and the non-existence of~$\clo(w)$ implies that the cluster of the solution which contains~$w$, contains all unscaled disks. 
		
		\proofsubparagraph*{Correctness of \textbf{(e)}, that is, correctness of computation of~$\boldsymbol{\far(w)}$ and~$\boldsymbol{\clo(w)}$.}
		We begin by showing that the computation of~$\far(w)$ is correct.
		Recall that~$G'$ is the cluster graph obtained from the input graph~$G$ after the removal of the initial maximal $P_3$-packing~$\mathcal{P}$.
		Note that for the correctness of \Cref{claim-correctness-far-w} we only rely on the cluster graph~$G'$ and the initial $P_3$-packing~$\mathcal{P}$ with vertex set~$P$, that is, the sets~$\mathcal{T}_{i-1}, N_{i-1}, F_{i-1}$, and the auxiliary graph~$H_{i-1}$ are not important.

		\begin{claim}
			\label{claim-correctness-far-w}
			The computation of~$\far(w)$ is correct.
		\end{claim}

		\begin{claimproof}
			Clearly, if~$\far(w)$ does not exist or if~$\far(w)\in P$, then one of the guesses in \textbf{1.} or \textbf{2.} of our algorithm in \Cref{branching-for-far(w)} yields the correct computation.
			Hence, in the following it is sufficient to focus on the case that~$\far(w)$ is contained in a %
			cluster of~$G'$.
			
			As we scale the disk $w$, our budget allows us to scale at most $k - 1$ other disks.
			Thus, if $\far(w)$ belongs to a cluster of $G'$, this cluster must be among the $k$ closest clusters.
			Note that this cluster is among the $k$ closest clusters and not among the $k-1$ closest clusters since $w$ is a scaled disk too.
			Otherwise, we would need to scale at least $k$ other disks to merge all the clusters.
			Let $Q$ denote the~$k$ furthest disks to $w$ in the cluster that contains $\far(w)$.
			Note that we consider $Q$ in our branching.
			If $\far(w) \in Q$, we are done, so assume $\far(w) \notin Q$.
			We now show that we can find some~$q$ which can take the role of~$\far(w)$ without losing correctness.
			As $\Size{Q} = k$, there exists at least one unscaled disk in $Q$.
			Let $q \in Q$ be the unscaled disk in $Q$ that is furthest away from $w$.
			By the definition of $\far(w)$, we have $\Dist{\far(w) - w} \geq \Dist{q - w}$.
			Since $\far(w)$ is not scaled and $q$ is not scaled, we have that $w,q,\far(w)$ are in the same cluster.
			Observe that we cannot have $\Dist{\far(w) - w} > \Dist{q - w}$ by the definition of $Q$ and our choice of $q \in Q$.
			Hence, $\Dist{\far(w) - w} = \Dist{q - w}$ must hold, and we can update $\far(w) \gets q$ while maintaining the defining properties of $\far(w)$.
			Consequently, the claim is proven.
		\end{claimproof}	
		
		Now, we argue that also the computation of~$\clo(w)$ is correct.
		
		\begin{claim}
			\label{claim-correctness-clo-w}
			The computation of~$\clo(w)$ is correct.
		\end{claim}
		\begin{claimproof}
			Recall that after the computation of~$\far(w)$ by our algorithm, $\clo(w)$ is computed as described in \Cref{branching-for-clo(w)}:
			(1) either~$\clo(w)$ does not exists (\textbf{1.} of \Cref{branching-for-clo(w)}), or
			(2) $\clo(w) \in Q =\bigcup_{j=1}^\ell D_j$ (\textbf{2.} of \Cref{branching-for-clo(w)}).
			Now, we argue that~$\clo(w)$ either does not exist, is a disk contained in~$Q$, or we can redefine~$\clo(w)$ as some disk in~$Q$.
			Suppose that this is not the case.
			Due to budget constraints, there exists at least one disk~$q\in Q$ which is not scaled.
			By definition of~$Q$, we have~$\Dist{\clo(w)-w}\le\Dist{q-w}$.
			Note that~$\Dist{\clo(w)-w}<\Dist{q-w}$ is not possible due to the definition of~$Q$.
			Hence, $\Dist{\clo(w)-w}=\Dist{q-w}$.
			Now, we can redefine~$\clo(w)$ as~$q$ since~$q$ is also not scaled in the solution.
			Thus, the computation of~$\clo(w)$ is correct and hence the claim is proven.
		\end{claimproof}
		
		\proofsubparagraph*{Correct computation of~$\boldsymbol{\far(w)}$ and~$\boldsymbol{\clo(w)}$ yields correctness of \textbf{(b)}--\textbf{(d)}.}
		Now, with the correct computation of~$\far(w)$ and~$\clo(w)$ at hand, we can finally verify the correctness of \textbf{(b)}--\textbf{(d)}.
		Since~$\far(w)$ and~$\clo(w)$ are in particular unscaled disks in the solution, setting~$F_i:=F_{i-1}\cup\{\far(w), \clo(w)\}$ is correct, and thus \textbf{(c)} is shown.

		Moreover, recall that set~$N_i$ and graph~$H_i$ is computed as described in \Cref{rule-compute-H-i+1}. 
		Note that~$X\cap F_{i}\ne\emptyset$ is impossible due to the correct computation of~$\far(w), \clo(w), F_i$, and the fact that there exists a solution.
		Consequently, every disk for which \textbf{(3)} of \Cref{rule-compute-H-i+1} applies needs to be rescaled, that is, is already rescaled (contained in~$\mathcal{T}_{i-1}$), or needs to be rescaled later.
		Thus, property \textbf{(b)}, that is, the computation of the set~$N_i$, is correct. 
		
		To see that also~$H_{i}$ is correctly computed, that is, $H_i$ is green-equivalent to~$\widetilde{G}_i$, observe that every vertex-pair with a different status (edge instead of non-edge or vice versa, or undefined edge instead of edge/non-edge) in~$\widetilde{G}_i$ compared to~$\widetilde{G}_{i-1}$ contains the currently scaled disk~$w$ or a vertex of~$N_i\setminus N_{i-1}$.
		Recall that~$H_{i-1}$ is green-equivalent to~$\widetilde{G}_{i-1}$ and to obtain~$H_i$ we do two things: 
		First, some edges incident to~$w$ receive a new color, and second, some edges not-incident to~$w$ are re-colored as \green.
		Hence, it is sufficient to argue that the changes made in the first case are correct.
		Observe that the correct computation of~$\far(w)$ implies that for any so-far unscaled disk~$x$, that is, $x\in\mathcal{S}\setminus \mathcal{T}_{i-1}$ such that~$\Dist{w-x}\le \Dist{w-\far(w)}$, we have that~$xw$ is an edge in~$\widetilde{G}_{i-1}$.
		Consequently, case \textbf{(1)} of \Cref{rule-compute-H-i+1} of recoloring such edge for~$x\in\mathcal{S}\setminus(\mathcal{T}_i\cup N_{i-1})$ is correct.
		Analogously, we can argue the correctness of case~\textbf{(2)} of \Cref{rule-compute-H-i+1}.
		All remaining edges incident to~$w$ are re-colored as \green{} which corresponds to an undefined state.
		Consequently, the computation of~$H_i$ is correct and thus \textbf{(d)} is verified.
		Thus, the inductive step is proven. 	
		
		\proofsubparagraph*{Completeness of set~$\boldsymbol{\mathcal{T}}$ at end of Phase~1.}
		Now, we have to argue that at the end of Phase~1 of our algorithm we have~$\mathcal{T}=\widetilde{\mathcal{T}}$ (again, by~$\mathcal{T}, N, F$, and~$H$ we denote our final data structures).
		Moreover, let~$\widehat{G}$ be the graph after all disks of~$\mathcal{T}$ are scaled with their correct radius according to the solution.
		Recall that our inductive proof shows $\mathcal{T} \subseteq \widetilde{\mathcal{T}}$.
		Thus, we focus on $\widetilde{\mathcal{T}} \subseteq \mathcal{T}$.
		
		Suppose towards a contradiction that this is not the case, that is, there exists at least one disk~$x\in\widetilde{\mathcal{T}}\setminus \mathcal{T}$, that is, $\mathcal{T}\subsetneq \widetilde{\mathcal{T}}$.
		This implies that~$\widehat{G}\ne\widetilde{G}$, that is, the graph~$\widehat{G}$ obtained by~$G$ after all disks of~$\mathcal{T}$ received their correct radii according to the solution is not the same as the graph~$\widetilde{G}$ corresponding to the solution.
		Also, this implies that~$\widehat{G}$ contains a~$P_3$ with vertex set~$U$ including disk~$x$, that is, $x\in U$. 
		We now argue that either~$H$ and~$\widehat{G}$ are \emph{not} green-equivalent, violating the correctness of the above induction, or that~$H$ contains a colorful~$P_3$, a violation of the termination conditions of Phase~1 of our algorithm.
		Clearly, $U\subseteq \mathcal{T}$ is not possible since all disks of~$\mathcal{T}$ already received their correct radius.
		Moreover, $|U\cap \mathcal{T}|\le 1$ is also not possible since~$|U\cap \mathcal{T}|\le 1$ implies that no edge with both endpoints in~$U$ can be \green{} in~$H$ and thus the~$P_3$ in~$\widehat{G}$ with vertex set~$U$ would be a colorful~$P_3$ in~$H$, a contradiction to the termination of Phase~1 of our algorithm.
		Hence, we have~$|U\cap \mathcal{T}|=2$. 
		
		We let~$U=\{u,w,x\}$.
		By our assumptions, we have~$u,w\in \mathcal{T}\subseteq\widetilde{\mathcal{T}}$ and thus both~$u$ and~$w$ already received their correct radius in~$\widehat{G}$. 
		Now, we distinguish the two cases whether the vertex-pair~$uw$ is \textbf{(a)} an edge in~$\widehat{G}$, or \textbf{(b)} a non-edge in~$\widehat{G}$. 
		First, we focus on case \textbf{(a)}, that is, $uw$ is an edge in~$\widehat{G}$.
		Since both~$u$ and~$w$ already received their correct radius in~$\widehat{G}$, there is a unique cluster~$C$ in the solution which contains both~$u$ and~$w$. 
		Since~$U$ induces a~$P_3$ in~$\widehat{G}$, we assume without loss of generality that~$wx$ is an edge in~$\widehat{G}$ and that~$ux$ is a non-edge in~$\widehat{G}$. 
		Recall that since~$x$ is currently unscaled, any edge incident with~$x$ in~$H$ is either \red{} or \blue, but \green{} is not possible since this is used only for all vertex-pairs where both endpoints are from~$\mathcal{T}$ (since due to termination of Phase~1 we have~$N=\emptyset$). 
		Moreover, due to the green-equivalence of~$H$ and~$\widehat{G}$, we conclude that~$xw$ is \blue{} in~$H$ and that~$xu$ is \red{} in~$H$.

		We now argue that~$C$ does \emph{not} only contain scaled disks.
		Assume towards a contradiction that this is not the case.  
		This implies that both~$\far(u)$ and~$\far(w)$ do not exist.
		Now, consider the situation of Phase~1 of our algorithm when disk~$w$ is scaled: only if~$\far(w)$ exists, all unscaled disks with distance at most~$\Dist{w-\far(w)}$ are re-colored as \blue{} (Case~\textbf{(1)} of \Cref{rule-compute-H-i+1}).
		Hence, $xw$ cannot have color \blue{} at that time.
		Since~$x$ is not scaled, later in Phase~1 we also do not re-color~$xw$ again.
		Thus, $xw$ cannot be \blue{} in~$H$, a contradiction to the above observation that~$xw$ has to be \blue{} in~$H$.
		
		Hence, it is safe to assume that~$C$ contains at least one unscaled disk.
		This implies that both~$\far(w)$ and~$\far(u)$ exist. 
		By definition of~$\far(w)$, the vertex-pair~$w$ and~$\far(w)$ has to be an edge in~$\widehat{G}$ and thus is \blue{} in~$H$ due to green-equivalence of~$H$ and~$\widehat{G}$. 
		Moreover, since~$\far(w)$ is not scaled in the solution, since~$u,w$ already received their correct radii, and since~$u,w,\far(w)$ are contained in the same cluster of the solution, also the edge between~$\far(w)$ and~$u$  is \blue{} in~$H$.
		Moreover, since both~$\far(w)$ and~$x$ are currently not scaled, the edge between~$x$ and~$\far(w)$ cannot be \green{} in~$H$.
		Also, the edge between~$x$ and~$\far(w)$ cannot be \red{} in~$H$ since then~$x,w,\far(w)$ would form a colorful~$P_3$ in~$H$, a contradiction to the termination of Phase~1.
		Thus, the edge between~$x$ and~$\far(w)$ is \blue{} in~$H$ and due to green-equivalence of~$H$ and~$\widehat{G}$, it is also an edge in~$\widehat{G}$.
		But now, $u,x,\far(w)$ yield a colorful~$P_3$ in~$H$, a contradiction to the termination of Phase~1 of our algorithm.
		
		Second, we consider the case \textbf{(b)}, that is, $uw$ is a non-edge in~$\widehat{G}$, and is thus \red{} in~$H$ due to green-equivalence of~$H$ and~$\widetilde{G}$.
		This implies that~$wx$ and~$ux$ are edges in~$\widehat{G}$ and thus \blue~edges in~$H$.
		Moreover, note that~$\far(u)$ has to exist: otherwise we can argue analogously as above that~$ux$ cannot be \blue{} in~$H$, violating the green-equivalence of~$H$ and~$\widehat{G}$.
		By definition of~$\far(u)$, the edge between~$u$ and~$\far(u)$ has to be \blue{} in~$H$.
		Thus, the edge between~$x$ and~$\far(u)$ also needs to be \blue{} since otherwise~$u,x,\far(u)$ would form a colorful~$P_3$ in~$H$, violating the termination of Phase~1 of our algorithm.
		Moreover, since both~$u$ and~$w$ already received their correct radius in~$\widehat{G}$, the cluster~$C_u$ containing~$u$ in the solution is a different cluster than~$C_w$ which contains~$w$ in the solution. 
		Consequently, since~$\far(u)\in C_u$ we have~$\far(u)\notin C_w$ and thus the edge between~$\far(u)$ and~$w$ has to be \red{} in~$H$.
		But now~$x,w,\far(u)$ form a colorful~$P_3$ in~$H$, again a contradiction to the termination of Phase~1 of our algorithm. 
		Hence, all cases led to a contradiction and thus at the end of Phase~1 of our algorithm, we have~$\mathcal{T}=\widetilde{\mathcal{T}}$, $N=\emptyset$, $F\cap \widetilde{\mathcal{T}}=\emptyset$, and that~$H$ and~$\widehat{G}$ are green-equivalent.
		
		\proofsubparagraph*{Correctness of Phase~2.}
		Thus, it remains to argue that Phase~2 is also correct. 
		Recall that our algorithm initially computes the set~$B(x)$ of \blue~neighbors of~$x$ in~$H$ and then in a first step~$(1)$, resolves all \green~edges of~$H$ where at least one endpoint has a non-empty set of \blue~neighbors  and in a second step~$(2)$, all remaining \green~edges (implying both endpoints have no \blue~neighbors in~$H$) are resolved.
		
		We begin by showing the correctness of~$(1)$.
		
		\begin{claim}
			\label{claim-handling-green-correct}
			\Cref{rule-phase-2} is correct.
		\end{claim}
		\begin{claimproof}
			We individually verify that each of the three cases is handled correctly.
			
			\textbf{Case 1:} Assume towards a contradiction that in the optimal solution~$uw$ is a non-edge, that is, has to be \red{} in~$H$ due to the green-equivalence of~$H$ and~$\widetilde{G}$.
			By definition, the vertex~$\far(u)$ is the furthest non-scaled vertex from~$u$ which is in the same cluster as~$u$.
			This implies~$\far(u)\in B(u)$.
			Since~$B(u)=B(w)$, we also have~$\far(u)\in B(w)$.
			Consequently, in the optimal solution~$\widetilde{G}$, the vertex-pairs~$u\far(u)$ and~$w\far(u)$ are edges, that is, correspond to \blue~edges in~$H$.
			Hence, the graph~$\widetilde{G}$ corresponding to the solution contains the $P_3$ with disks~$u,w,\far(u)$, a contradiction to the fact that the solution is a cluster graph.
			Thus, re-coloring~$uw$ as \blue{} is correct. 
			
			\textbf{Case 2:} Assume towards a contradiction that in the optimal solution~$uw$ is an edge, that is, has to be \blue{} in~$H$ due to the green-equivalence of~$H$ and~$\widetilde{G}$.
			According to our assumption, $B(u)\ne\emptyset$ and thus~$\far(u)$ exists since~$B(u)$ only consists of unscaled disks.
			According to the definition of~$\far(u)$, we have~$\far(u)\in B(u)$ and thus~$u\far(u)$ is \blue{} in~$H$, and thus also in the graph~$\widetilde{G}$ corresponding to the solution.
			Recall that~$\far(u)$ is an unscaled disk.
			Since~$B(u)\cap B(w)=\emptyset$, we thus observe that~$w\far(u)$ is \red{} in~$H$, and thus the vertex-pair~$w$ and~$\far(u)$ is a non-edge in~$\widetilde{G}$  due to the green-equivalence of~$H$ and~$\widetilde{G}$.
			Hence, $u,w,\far(u)$ form a~$P_3$ in~$\widetilde{G}$, a contradiction to the fact that~$\widetilde{G}$ is a cluster graph.
			Thus, re-coloring~$uw$ \red{} is correct.

			\textbf{Case 3:} Since~$B_{uw}:=B(u)\cap B(w)\ne \emptyset$, there exists at least one~$x\in B(u)\cap B(w)$.
			Moreover, since~$B(u)\ne B(w)$, either~$B(u)\setminus B_{uw}\ne\emptyset$ or~$B(w)\setminus B_{uw}\ne\emptyset$.
			Without loss of generality, assume the first condition is true and let~$z\in B(u)\setminus B_{uw}$.
			By definition of~$x$ and~$z$ we have that~$wx$ and~$zx$ are \blue{} in~$H$ and thus~$wx$ and~$zx$ are edges in~$\widetilde{G}$  due to the green-equivalence of~$H$ and~$\widetilde{G}$.
			Moreover, since~$x,z\in B(u)$, both~$x$ and~$z$ are unscaled disks.
			Also, since~$z\notin B(w)$ we conclude that~$wz$ is \red\ in~$H$, and thus~$wz$ is a non-edge in~$\widetilde{G}$.
			Consequently, $w,x,z$ form a~$P_3$ in~$\widetilde{G}$, a contradiction to the fact that~$\widetilde{G}$ is a cluster graph.%
		\end{claimproof}
		
		Hence, it remains to show that the handling of the remaining \green~edges, that is, the \green~edges between vertices which do not have \blue~neighbors, is correct. 
		Let~$X=\{x\in \mathcal{T}=\widetilde{\mathcal{T}}: B(x)=\emptyset\}$.
		Observe that all \green~edges of~$H$ containing \emph{exactly} one endpoint in~$X$, are already handled above in \textbf{Case~2} of \Cref{rule-phase-2} and correctly re-colored as \red, that is, correspond to non-edges in~$\widetilde{G}$.
		In other words, it remains to show that the \green~edges between vertices with both endpoints in~$X$ are resolved correctly. 
		Recall that in our algorithm, we test all possible clusterings of~$X$.
		Clearly, one of these choices has to be correct, since for each cluster~$C$ containing at least one disk of~$X$ we have~$V(C)\subseteq X$.
		Consequently, after resolving all \green~edges of~$H$, we have that~$H$ and~$\widetilde{G}$ are equivalent (there are no \green~edges anymore).
		Finally, to obtain the correct radii for the scaled disks in~$\mathcal{T}$ we invoke the LP of \Cref{sec:frameworks} and the correctness of this follows from \Cref{prop:radius-feasibility-running-time}.
		Hence, Phase~2 of our algorithm is also correct and thus the entire correctness of our algorithm is verified.
	\end{proof}
\end{statelater}

\onlyLong{We now use \Cref{lem:cluster-fpt-correctness} to show that \Cluster is \FPT\ with respect to $k$. %
	In particular, we establish the following result.}
\begin{restatable}\restateref{thm:cluster-fpt}{theorem}{theoremClusterFPT}
	\label{thm:cluster-fpt}
	\Cluster can be solved in $2^{\OO(k\cdot \log(k))}\cdot n^{\OO(1)}$~time.
\end{restatable}
\begin{prooflater}{ptheoremClusterFPT}
	Let \instance be an instance of \Cluster.
	We have established correctness of our algorithm in \Cref{lem:cluster-fpt-correctness}.
	It remains to show its running time.
	To this end, recall that our algorithm is a branch-and-bound algorithm. 
	In Phase~1, in each step exactly one new disk gets scaled (in fact, we only 'guess' the neighborhood of the scaled disk but not the correct radius) and consequently the depth of the corresponding branching tree is bounded by~$k$.
	Moreover, after Phase~1 is completed our algorithm first invokes a polynomial-time processing to resolve some \green~edges of~$H$ and then uses branching to resolve the remaining \green~edges.
	Afterwards, the resulting LP is solved as described in \Cref{sec:frameworks}.
	According to \Cref{prop:radius-feasibility-running-time}, this requires $2^{\OO(k\cdot \log(k))}\cdot n^{\OO(1)}$~time. 
	For this branching the depth is one.
	Hence, it remains to analyze the branching numbers of both phases. 
	
	\proofsubparagraph*{Analysis of Phase~1.}
	First, we determine a disk~$w$ which gets scaled.
	If~$N_i\ne\emptyset$, then we pick an arbitrary but fixed disk of~$N_i$, and otherwise, if~$N_i=\emptyset$, we search for a colorful~$P_3$ in~$H_i$ and branch among the at most three possibilities to guess~$w$.
	Hence, in at most three branches our algorithm guesses a correct disk~$w$.
	Second, our algorithm determines~$\far(w)$ and~$\clo(w)$.
	We now analyze the number of branches to determine them.

	\begin{claim}
		\label{claim-number-of-choices-furthest}
		The number of suitable choices for~$\far(w)$ is at most%
		~$k^2 + 3k + 1$.
	\end{claim}
	
	\begin{claimproof}
		In Step~1, we bound the number of possibilities in clusters and in Step~2, we obtain the final bound by also considering the disks~$P$ in the initial $P_3$-packing~$\mathcal{P}$.

		\emph{Step~1:} Recall that for the case that~$\far(w)$ is contained in a cluster, our algorithm considers the branches described in \textbf{3.} of \Cref{branching-for-far(w)}.
		Thus, there are at most $k^2$~choices for~$\far(w)$.

		\emph{Step~2:} Moreover, observe that each disk of~$P$ corresponding to the initial $P_3$-packing~$\mathcal{P}$ is also a valid candidate for~$\far(w)$.
		This yields $3k$~possibilities for~$\far(w)$ (see \textbf{2.} of \Cref{branching-for-far(w)}).
		Finally, note that it is also possible that~$\far(w)$ does not exist, that is, in the solution~$w$ is in a cluster with only scaled disks (see \textbf{1.} of \Cref{branching-for-far(w)}).
		In total this yields at most %
		$k^2 + 3k + 1$~choices for~$\far(w)$.
	\end{claimproof}

	\begin{claim}
		\label{claim-number-of-choices-closest}
		Once~$\far(w)$ is fixed, the number of choices for~$\clo(w)$ is at most~$k+1$.
	\end{claim}

	\begin{claimproof}
		For the possibilities for branching, see \Cref{branching-for-clo(w)}.
		Clearly, this yields $k$~possibilities if~$\clo(w)$ exists.
		Since~$\clo(w)$ might not exist, there are at most $k+1$~possibilities for~$\clo(w)$.%
	\end{claimproof}	
	
	Hence, Claims~\ref{claim-number-of-choices-furthest}--\ref{claim-number-of-choices-closest} yield that the total branching number of Phase~1 is~%
	$3\cdot (k^2+3k+1)\cdot (k+1)$.
	
	\proofsubparagraph*{Analysis of Phase~2.}
	Since any clustering on at most $k$~disks can contain at most~$k$ non-empty clusters, there are at most~$k^k$ different clusterings (since the number of these partitions corresponds to the $k$th Bell number~$B_k$ and~$B_k\le k^k$).
	Hence, the branching number of Phase~2 is~$k^k$.
	
	\proofsubparagraph*{Overall running time.} 
	Hence, the overall running time of our algorithm is~$(3\cdot(k^2 + 3k + 1)\cdot (k+1))^k\cdot k^k\cdot 2^{\OO(k\cdot \log(k))}\cdot n^{\OO(1)}\in (k^3)^{\OO(k)}\cdot 2^{\OO(k\cdot \log(k))}\cdot n^{\OO(1)}\in 2^{\OO(k\cdot\log(k))}\cdot n^{\OO(1)}$.
	Combining all, the theorem follows.
\end{prooflater}

\section{Scaling to Cluster Graphs is \textsf{NP}-hard}
\label{sec:cluster-hardness} 
\shortLong{
We now justify the \FPT-algorithm from \Cref{sec:cluster-fpt} and show%
}
{%
To justify our FPT-algorithm for \Cluster from \Cref{sec:cluster-fpt}, we now show%
}
that \Cluster is \NP-hard for all fixed rationals $0 < \rmin \leq \rmax$ except $\rmin = \rmax = 1$.
We provide two different reductions, depending on whether $\rmin < 1$ or $\rmin \geq 1$,
that use the same gadget.

\shortLong{
\subparagraph*{Heavy $\boldsymbol{P_3}$s.}
}
{
\subsection{A Key Gadget: Heavy $\boldsymbol{P_3}$s}
\label{sec:cluster-hardness-heavy}	
}
For a point $p \in \mathbb{R}^2$ and an integer $\delta > 1$, we say that $p$ is \emph{$\delta$-heavy}, denoted as $p^{\delta}$, if it represents $\delta$ co-located \emph{copies} $p^{\delta}(1)$, $p^{\delta}(2)$, $\ldots$, $p^{\delta}(\delta)$. %
We define $\mathcal{S}[p] \coloneqq \{p^{\delta}(i) : i \in [\delta]\}$ and also call the corresponding disks $\delta$-heavy.
A $\delta$-heavy disk induces a $\delta$-clique in the disk graph regardless of the radius assignment $r$, and we can assume that all copies receive equal radii in $r$, in particular, all are scaled or not scaled.\footnote{For this it also suffices to place all copies in a sufficiently small region around $p$ to avoid exact co-location.}
Our main gadget is a \emph{$\delta$-$\theta$-heavy $P_3$} with $\theta > k \geq \delta$, or simply \emph{heavy $P_3$} if $\delta$ and $\theta$ are clear from the context.
It consists of three disks $p_1$, $p_2$, and $p_3$ with \NewText{$1 < \Dist{p_1 - p_2}, \Dist{p_2 - p_3} \leq 2$}, one of which is a $\theta$-heavy disk, \shortLong{often}{in most cases this will be} $p_2$, and the remaining two are $\delta$-heavy.
\NewText{Furthermore, many heavy~$P_3$s will have $\Dist{p_1 - p_2} = \Dist{p_2 - p_3} = \xi$ for some $1 < \xi \leq 2$.
Note that we allow larger distances between the heavy disks during the construction of the instance.
We will ensure that the correct distances are maintained in the final instance by appropriately scaling the entire construction.}
\NewText{Observe that a $\delta$-$\theta$-heavy~$P_3$} induces $\delta^2\cdot\theta$ distinct $P_3$s: %
\begin{observation}
\label{obs:cluster-hardness-heavy}
Let \instance be an instance of \Cluster, let $P$ be a $\delta$-$\theta$-heavy~$P_3$ in \instance, and $r$ a solution.
We have $\Size{\mathcal{T}(r)} \geq \delta$ and $\mathcal{S}[p] \subseteq \mathcal{T}(r)$ for a
$\delta$-heavy disk $p \in V(P)$. %
\end{observation}

\newcommand{\showClusterHardnessGadget}{
\begin{figure}
\centering
\includegraphics[page=1]{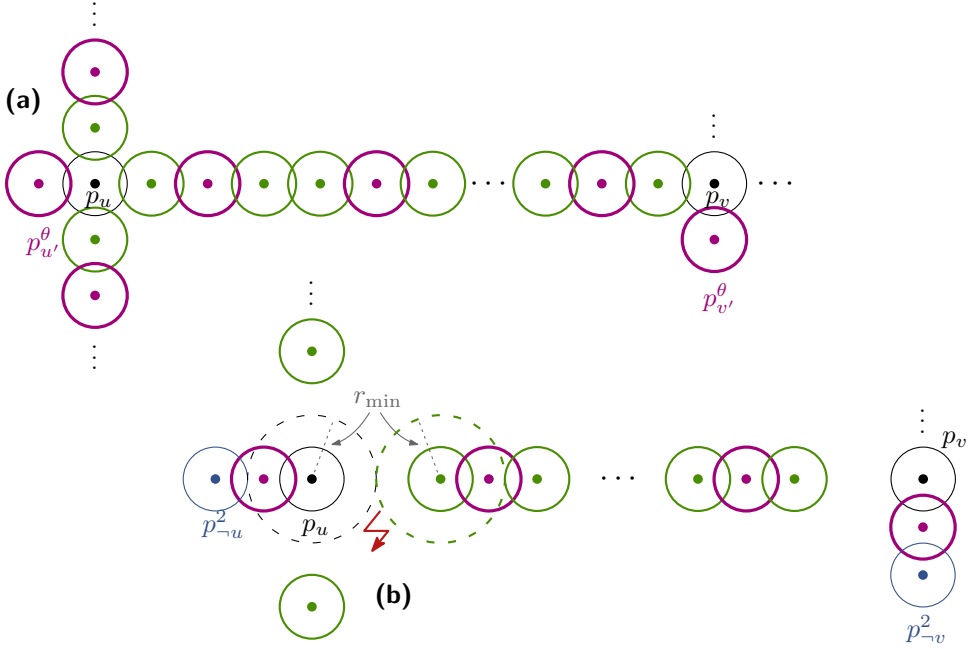}
\caption{
	The edge $uv$ for \textbf{\textsf{(a)}} $\rmin < 1$ and \textbf{\textsf{(b)}} $\rmin > 1$. Black, blue, green, and purple disks are 1-, 2-, $\eta^2$-, and $\theta$-heavy, respectively.
	In (b), if we enlarge $p_u$ and the left endpoint of a heavy $P_3$, then two $\theta$-heavy disks are in the same connected component but not adjacent as $\theta > k$.
}
\label{fig:cluster-hardness-enlarge-shrink-edge}
\end{figure}
}

\onlyShort{
We now provide an overview of the constructions; details can be found in \Cref{app:cluster-hardness}.
\showClusterHardnessGadget

\subparagraph*{Hardness for Shrinking Disks ($\mathbf{0 < \rmin < 1}$).}

The reduction is inspired by ideas from Fomin et al.~\cite{FG00Z25}.
We fix a rational value $0 < \rmin < 1$; the instance that we construct will later be scaled by a factor depending on $\rmin$.
Let $\instance_{\probname{VC}} = (H, \kappa)$ be an instance of \VC where $H$ is a planar cubic graph with $\Size{V(H)} = \eta$; \VC remains \NP-hard in this case~\cite{GJ.CIG.1979}.
We compute in polynomial time a planar rectilinear embedding $\Gamma$ of $H$ of polynomial area~\cite{LMS.LA2.1998} and scale it by an integer $\gamma > 1$ that depends on $\rmin$.
Similar to Fomin et al.~\cite{FG00Z25}, we represent edges of $H$ by evenly spaced disks such that the resulting unit disk graph encodes a suitable subdivision of $H$; see \Cref{fig:cluster-hardness-enlarge-shrink-edge}a.
Note that we can slightly compress the distance between adjacent disks to accommodate an additional disk when required.
Every vertex $v \in V(H)$ is represented by a disk $p_v$, and $v$ will be part of a vertex cover $C$ of $H$ if and only if $r(p_v) \neq 1$.
To ensure that every edge $uv \in E(H)$ is covered, %
we turn the disks along $uv$ into a series of interleaving $\eta^2$-$\theta$-heavy $P_3$s. %
Since $H$ is cubic, we can place a further $\theta$-heavy disk~$p_{v'}^{\theta}$ next to $p_v$, for every $v \in V(H)$, that does not interfere with the edges incident to~$v$; see \Cref{fig:cluster-hardness-enlarge-shrink-edge}a.
This makes $p_v$ part of four $\theta$-heavy $P_3$s.
If we do not scale~$p_v$, these $P_3$s, together with our budget $k$ specified below, force the disks $p_u$ for all vertices $u \in N_H(v)$ to be scaled (otherwise a $P_3$ remains along~$uv$).
We set $k = k_{\text{fix}} + \kappa$, where $k_{\text{fix}}$ is the number of enforced scaling operations to remove all heavy $P_3$s along edges of $H$, and $\theta = k + 1$.
For the value of $\rmax$, we observe that removing a $P_3$ can only be done via shrinking some disks, since initially multiple $\theta$-heavy disks are in the same connected component; recall $\theta > k$.
Thus, we can set any fixed rational $\rmax \geq \rmin$ in the construction.

}

\begin{statelater}{clusterHardnessShrinkDetails}
\subsection{Hardness-Construction for Shrinking Disks ($\mathbf{0 < \rmin < 1}$)}
\label{sec:cluster-hardness-shrink}
In this section, we show that \Cluster is \NP-hard for any fixed rational $0 < \rmin < 1$.
To establish \NP-hardness, we use a construction inspired by the reduction by Fomin et al.~\cite{FG00Z25} to obtain \NP-hardness for \PiScalingLong[\probname{Independence}].
In their work, Fomin et al.\ reduced from \probname{Independent Set}, which remains \NP-hard even if the graph $H$ is planar and cubic.
On a high-level, they considered a sufficiently scaled rectilinear embedding of $H$ and represented every edge with evenly spaced disks such that the resulting unit disk graph corresponds to a graph $H'$ obtained by subdividing every edge of $H$ an even number of times.
As this operation preserves the existence of an independent set, and $G(\mathcal{S}, r)$ should be independent, correctness follows.
In contrast to Fomin et al.~\cite{FG00Z25}, subdividing edges is no longer sufficient to obtain a correct reduction to \Cluster, since $G(\mathcal{S}, r)$ can be a cluster graph even though it is not edge-less.
To overcome this obstacle, we no longer place evenly spaced disks along edges but evenly spaced $P_3$s.
Since a solution must be $P_3$-free, we must remove the $P_3$ by scaling at least one disk.
By arranging the $P_3$s correctly, we can use this observation to forward information along an edge, in particular, which of its endpoints is (not) part of the solution.

In the following, we formalize the above-sketched idea.
Let us fix a rational, positive value $0 < \rmin < 1$ and define $\alpha \coloneqq \lceil \frac{2}{1 - \rmin'} \rceil \geq 2$, where $\rmin' = \rmin + \varsigma$ for some arbitrary small constant $0 < \varsigma \ll (1 - \rmin)$.\footnote{We use $\varsigma$ to ensure that two disks do not intersect after scaling one of them down to $\rmin$; recall that we consider closed disks.}
Let $\instance_{\probname{VC}} = (H, \kappa)$ be an instance of \VC, where $H$ is a planar cubic graph on $\eta = \Size{V(H)}$ vertices.
We note that \VC remains \NP-hard in this case~\cite{GJ.CIG.1979}.
Furthermore, let $\Gamma$ be a rectilinear planar embedding of $H$ of polynomial area and where each edge has at most three bends.
We can compute $\Gamma$ in polynomial time~\cite{LMS.LA2.1998}.
Note that every edge $e \in E(H)$ corresponds to a sequence of (up to three) segments in $\Gamma$.
We scale $\Gamma$ by an integer $\gamma \geq 1$ such that (1) the vertical/horizontal distance between any pair of vertical/horizontal segments is greater than $4\alpha$ and (2) any segment $s$ has an integer length $\Dist{s} > 0$ which is divisible by \NewText{$12\beta(2\alpha - 1)$}, where $\beta \coloneqq 2\lceil \frac{2\alpha - 2}{3} \rceil \geq \alpha \geq 2$.
The reasoning behind these values will become clear when we construct the instance, in particular, when we replace edges of $H$ with interleaving heavy $P_3$s.
We can assume that $\Gamma$ is drawn on an integer grid.
Hence, such a $\gamma$ must exist.
Let $\Gamma^{\gamma}$ denote the scaled drawing.
Recall our overall approach, which consists of representing each edge $uv \in E(H)$ with a series of $P_3$s to ensure that $u$ or $v$ are part of a vertex cover $C$ of $H$.
We now formalize this idea using $\delta$-$\theta$-heavy~$P_3$s, where we specify the values of $\delta$ and $\theta$ in the end.
Initially, the following construction is scaled up by a factor of $\alpha$ \NewText{in order to} ensure that we place every disk on integer coordinates.
In the end, we can scale it down to obtain an instance of \Cluster.

\subparagraph*{Constructing the Instance.}
Consider an edge $e = uv \in E(H)$.
Let $s(e)$ denote the segments of $e$ in $\Gamma^{\gamma}$ and let $s \in s(e)$ be such a segment. 
For the following discussion, order the endpoints of $s$ along the $u$-to-$v$ traversal of $e$.
Observe that $\Dist{s} = \lambda_s \cdot (3\beta(2\alpha - 1))$ must hold for some integer \NewText{$\lambda_s \geq 4$}; observe that $\lambda_s$ ``contains'' a factor \NewText{four} of $\Dist{s}$, which we will later use when placing disks along $s$.
We now differentiate between two cases, depending on the second endpoint of $s$ being $v$ or not.
Intuitively, if $v$ is the second endpoint of $s$, we have traversed the entire edge $uv$.
In this case, we introduce an additional disk which will later represent the vertex $v$.
But let us first discuss the easier case, where $v$ is not the second endpoint of $s$.
If this is the case, we place $3\lambda_s'+1$-many disks $p_{s,1}, p_{s,2}, \ldots, p_{s,3\lambda_s'+1}$ along~$s$, starting at the first endpoint of $s$, where $\lambda_s' = \lambda_s\cdot\beta$.
We distribute the disks along $s$ such that $\Dist{p_{s, i + 1} - p_{s,i}} = 2\alpha - 1$ for all $i \in [3\lambda_s']$; see \Cref{fig:cluster-hardness-enlarge-shrink-edge}a.
Observe that the distance between $p_{s,3\lambda_s'+1}$ and $p_{s,1}$ equals $(3\lambda_s')\cdot(2\alpha - 1) = \lambda_s \cdot (3\beta(2\alpha - 1))$.
Later, we will introduce $\lambda_s'$-many vertex-disjoint $\delta$-$\theta$-heavy $P_3$s along $s$, each composed of three consecutive disks.
The first disks on $s$ either represents the vertex $u$, if $u$ is the first endpoint of $s$, or will be used to connect different segments together.
\onlyLong{
\showClusterHardnessGadget
}

In the second case, where $v$ is the second endpoint of $s$, we want to place one additional disks on $s$, which later will represent the vertex $v$.
To this end, we use a similar trick as Fomin et al.~\cite{FG00Z25}.
More \NewText{concretely}, let \NewText{$\mu \coloneqq \frac{4\alpha + 6 \beta}{6 \beta + 2}$} and observe $1 < \mu < 2$ since $\beta \geq \alpha \geq 2$.
We place the first $3\beta + 1$-many disks $p_{s,1}, p_{s,2}, \ldots, p_{s,{3\beta+1}}$ along $s$ as above, i.e., such that $\Dist{p_{s,{i + 1}} - p_{s,i}} = 2\alpha - 1$ for all $i \in [3\beta]$.
For the next \NewText{$6\beta+2$-many} disks \NewText{$p_{s,3\beta+2}, p_{s,3\beta+3}, \ldots, p_{s,{9\beta+3}}$}, we now maintain $\Dist{p_{s,{i + 1}} - p_{s,i}} = 2\alpha - \mu$ for all \NewText{$3\beta + 1 \leq i < 9\beta + 3$}.
We observe that the distance between \NewText{$p_{s,{9\beta + 3}}$} and $p_{s,3\beta + 1}$ equals \NewText{$(6\beta + 2)(2 \alpha - \mu)$, which we can simplify to
\begin{align*}
(6\beta + 2)(2 \alpha - \mu) &= (6\beta + 2) \cdot 2\alpha - ((6\beta+2) \cdot \mu) = (6\beta + 2) \cdot 2\alpha - (4\alpha + 6\beta)\\
&=
12\alpha\beta + 4\alpha - 4\alpha - 6\beta = 6\beta \cdot (2\alpha - 1).
\end{align*}}%
Since $\Dist{s} = \lambda_s \cdot (3\beta(2\alpha - 1))$, for some integer \NewText{$\lambda_s \geq 4$}, and the first disks occupy a distance of \NewText{$9\beta(2\alpha - 1)$}, we observe that there remains on $s$ a stretch of length \NewText{$(\lambda_s - 3)(3\beta(2\alpha - 1))$}.
We place on this remaining stretch $3\beta(\lambda_s-2) + 1$-many disks $p_{s,{x + 1}}, p_{s,{x + 2}}, \ldots, p_{s,{x + 3\beta(\lambda_s-2) + 1}}$, where \NewText{$x = 9\beta +3$}, such that $\Dist{p_{s,{x + i + 1}} - p_{s,{x + i}}} = 2\alpha - 1$ for every \NewText{$i \in [(\lambda_s-3)3\beta]$}.
The distance between $p_{s,1}$ and \NewText{$p_{s,{x + 3\beta(\lambda_s-3) + 1}}$} equals
\NewText{$3\beta \cdot (2\alpha - 1)) + 6\beta \cdot (2\alpha - 1) + 3\beta\cdot(\lambda_s-3)\cdot(2\alpha - 1) = \lambda_s \cdot (3\beta(2\alpha - 1))$}, where we derived the first two terms above.
Hence, in the end we have covered the entire segment $s$ with (mostly evenly) distributed disks; see \Cref{fig:cluster-hardness-enlarge-shrink-edge}a.

We repeat this process for every segment $s$ of the edge $uv$ and every edge $e \in E(H)$.
Note that if the position of two disks coincides, we only create one of them.
The above construction forms the base layout of our reduction.
Let $\mathcal{S}'$ denote the point set constructed so far.
We observe that $\mathcal{S}'$ contains for every vertex $v \in V(H)$ one disk $p_v$ on the position of~$v$ in $\Gamma^{\gamma}$, i.e., with $\Gamma^{\gamma}(v) = p_v$.
In the following, we say that $p_v$ \emph{represents} the vertex $v$.
We will use the following equivalence between a vertex cover $C \subseteq V(H)$ of $H$ and the radius assignment $r(p_v)$ of $p_v$: $v \in C$ if and only if $r(p_v) \neq 1$, i.e., if and only if $p_v \in \mathcal{T}(r)$.

Observe that every edge $uv$ is represented by a series of \NewText{$3\cdot q(uv)$} disks along $uv$, ignoring $p_u$ and $p_v$, where \NewText{$q(uv) \coloneqq \sum_{s \in s(e)} \lambda_s\beta + 1$}.
The distance between any two consecutive disks along $uv$, including the disks $p_u$ and $p_v$, is at most $2\alpha - 1$.
The distance between any other pair of points in $\mathcal{S}'$ is (strictly) greater than $2\alpha$.
Note that this also holds along bends of $uv$ since $\alpha \geq 2$.
Consequently, we have $G(\mathcal{S}', r') = H'$, where $r'(p) = \alpha$ for every $p \in \mathcal{S}'$ and $H'$ is a subdivision of $H$.
So far, the construction is similar to the one by Fomin et al.~\cite{FG00Z25}.
However, as already hinted at in the high-level overview, we must make two adaptions to ensure correctness.

First, for every edge $uv \in E(H)$, let \NewText{$p_u, p_1, p_2, p_3, \ldots, p_{3\cdot q(uv)}, p_v$} be the disks along the edge, where we recall that the disks $p_u$ and $p_v$ represent the vertices $u$ and $v$, respectively.
We %
turn every disk $p_{3i - 1}$, \NewText{$i \in [q(uv)]$}, into a $\theta$-heavy disk $p_{3i - 1}^{\theta}$, i.e.,~$p_{3i - 1}$ is replaced by a set $\mathcal{S}[p_{3i - 1}]$ of $\theta$ disks, where we specify $\theta$ in the end.
Moreover, we turn every disk $p_{3i}$ and $p_{3i - 2}$, \NewText{$i \in [q(uv)]$}, into $\eta^2$-heavy disks $p_{3i}^{\eta^2}$ and $p_{3i - 2}^{\eta^2}$, respectively, where we recall that $\eta = \Size{V(H)}$.
Observe that this transforms the (unit) disks along the edge~$uv$ into a series of interleaving $\eta^2$-$\theta$-heavy $P_3$s %
between the disks.
Furthermore, observe that we also have $q$ vertex-disjoint $\eta^2$-\NewText{$\theta$}-heavy $P_3$s along $uv$.
Moreover, only the disks $p_v$ for $v \in V(H)$ are not heavy.

Second, for every vertex $v \in V(H)$, let $p_v \in \mathcal{S}'$ be the disk that represents it.
Recall that $H$ is cubic, i.e., every vertex, and in particular $v$, has degree three in $H$.
Hence, in the rectilinear drawing $\Gamma^{\gamma}$ of $H$, there is exactly one direction of top, bottom, left, and right without any incident edge.
Say that there is no incident edge that leaves $v$ to the top.
We place a $\theta$-heavy disk $p_{v'}^{\theta}$ exactly $2\alpha - 1$ to the top of $p_v$.
For the other cases we perform symmetrically.
Observe that $p_v$ is involved in three $\theta$-heavy $P_3$s, one for each incident edge; see again \Cref{fig:cluster-hardness-enlarge-shrink-edge}a.
Hence, if we do not scale $p_v$, we must scale all copies of the three overlapping $\eta^2$-heavy disks that represent (the start of) the three incident edges.
This completes the construction of the point set $\mathcal{S}$ and by scaling down the entire instance by a factor of $\alpha$, we obtain a unit disk graph.

\subparagraph*{The Complete Reduction.}
Recall that $\instance_{\probname{VC}} = (H, \kappa)$ is an instance of \VC where $H$ is a cubic and planar graph with $\eta = \Size{V(H)}$.
To complete the construction of an instance $\instance$ of \Cluster, we still need to specify the values for $\rmin$, $\rmax$, and~$k$.
Furthermore, we need to specify the value of $\theta$.
We define \NewText{$k_{\text{fix}} \coloneqq \eta^2\sum_{e \in E(H)}q(e)$}.
\NewText{Note that}~$k_{\text{fix}}$ is the minimum number of scaling operations required ``within'' the edges of~$H$. %
Scaling the disk $p_v$ should correspond to having $v \in C$, for a vertex cover $C \subseteq V(H)$ of $H$ (of size $\kappa$).
Thus, we set $k = k_{\text{fix}} + \kappa$.
Since we want to prevent that $\theta$-heavy disks get scaled in any solution to our constructed instance, we can set $\theta = k + 1$.

Recall that we specified in the beginning a value $0 < \rmin < 1$.
Furthermore, recall that every pair of overlapping disks in $\mathcal{S}'$ has a distance of either $2\alpha - 1$ or $2\alpha - \mu$ with $1 < \mu \leq 2$.
Hence, in the reduced instance, the distance between any pair of overlapping disks from $\mathcal{S}'$ is at \NewText{least} $\frac{2\alpha - 2}{\alpha}$.
Consequently, to remove such an edge, we must scale one of its endpoints $p$ to a radius smaller than $r(p) = \frac{2\alpha - 2}{\alpha} - 1$.
Recalling $\alpha = \lceil\frac{2}{1 - (\rmin + \varsigma)}\rceil$ and $\rmin' = \rmin + \varsigma$, we can make the following observation using basic calculus.
\begin{align*}
r(p) = \frac{2(\alpha - 1)}{\alpha} - 1 \geq \left(\frac{2}{1 - \rmin'} - 1\right)(1 - \rmin') - 1 = \rmin'
\end{align*}
Recall $\rmin < \rmin' = r(p)$.
Hence, scaling the endpoint $p$ to $\rmin$ is sufficient to remove the edge.
Observe that enlarging the radius of a disk is never beneficial:
On the one hand, we can only remove a $P_3$ by shrinking one of its disks.
On the other hand, if we enlarge a disk too much, then we risk connecting two disjoint $\theta$-heavy disks, which can never lead to a solution due to our budget $k$.
Consequently, for $\rmax$ we can choose any rational $\rmax > \rmin$.
With this at hand, we have completed the construction of our instance $\instance = (\mathcal{S}, \rmin, \rmax, k = k_{\text{fix}} + \kappa)$ of \Cluster.

We now make some useful observations for showing correctness of the reduction.
First, recall that for an edge $e \in E(H)$, $s(e)$ denotes all segments of $e$ in the drawing $\Gamma^{\gamma}$.
Furthermore, the length of every segment $s \in s(e)$ in $\Gamma^{\gamma}$ is $\Dist{s} = \lambda_s \cdot (3\beta(2\alpha - 1))$ for some integer \NewText{$\lambda_s \geq 4$}.
We defined \NewText{$q(uv) \coloneqq \sum_{s \in s(e)} \lambda_s\beta + 1$} and placed on the segments $s$
exactly \NewText{$q(e)$-many} vertex-disjoint $\eta^2$-$\theta$-heavy $P_3$, all of which need at least $\eta^2$ scaling operations to be removed.
Moreover, every segment $s$ contains $\lambda_s$-many $\theta$-heavy disks, one per vertex-disjoint $\eta^2$-$\theta$-heavy $P_3$.
We observe:
\begin{observation}
\label{obs:cluster-hardness-scaling-on-edge}
Let $uv \in E(H)$ be an edge of $H$ and let \NewText{$p_u, p_1, p_2, p_3, \ldots, p_{3\cdot q(uv)}, p_v$} %
be the disks along $uv$.
In every solution $r$, we have $\mathcal{S}[p_{3i - 2}] \subseteq T(r)$ or $\mathcal{S}[p_{3i}] \subseteq T(r)$, i.e., we scale all $\eta^2$-copies of the $\eta^2$-heavy disks $p_{3i - 2}^{\eta^2}$ or $p_{3i}^{\eta^2}$, for every \NewText{$i \in [q(uv)]$}.
Moreover, this amounts to scaling at least \NewText{$\eta^2q(uv)$} disks along the edge \NewText{$uv$}.
\end{observation}
Note that scaling a single copy of \NewText{an} $\eta^2$-heavy disk does not remove \NewText{an} $\eta^2$-$\theta$-heavy $P_3$ it is involved in.
Together with $k = k_{\text{fix}} + \kappa < k_{\text{fix}} + \eta^2$, we can use this to show the following property, which strengthens \Cref{obs:cluster-hardness-scaling-on-edge}.
\begin{lemma}
\label{lem:cluster-hardness-scaling-on-edge}
Let $uv \in E(H)$ be an edge of $H$ and let \NewText{$p_u, p_1, p_2, p_3, \ldots, p_{3q(uv)}, p_v$}, %
	be the disks along $uv$.
For any solution $r$ to our construction, either $\mathcal{S}[p_{3i - 2}] \subseteq T(r)$ or $\mathcal{S}[p_{3i}] \subseteq T(r)$ (but not both) holds for every \NewText{$i \in [q(uv)]$}.
\end{lemma}
\begin{proof}
Towards a contradiction, assume that there exists a radius $r$ assignment that violates the lemma statement.
In particular, there exists an edge $uv \in E(H)$ represented by the disks \NewText{$p_u, p_1, p_2, p_3, \ldots, p_{3q(uv)}, p_v$}, %
such that $\mathcal{S}[p_{3i - 2}] \subseteq T(r)$ and $\mathcal{S}[p_{3i}] \subseteq T(r)$ holds for some \NewText{$i \in [q(uv)]$}.
This removes the $\eta^2$-$\theta$-heavy $P_3$ on the disks $p_{3i - 2}$, $p_{3i - 1}$, and $p_{3i}$.
The disks along $uv$ form $q-1$ remaining $\eta^2$-$\theta$-heavy $P_3$s, namely on the disks $p_{3j - 2}$, $p_{3j - 1}$, and $p_{3j}$, for all \NewText{$j \in [q(uv)]$}, $j \neq i$; recall \Cref{fig:cluster-hardness-enlarge-shrink-edge}a.
These $P_3$ are disjoint from the one on the disks $p_{3i - 2}$, $p_{3i - 1}$, and $p_{3i}$.
Hence, by \Cref{obs:cluster-hardness-scaling-on-edge}, we must scale at least \NewText{$\eta^2(q(uv) - 1)$} additional disks along $uv$.
In total, this equals \NewText{$\eta^2(q(uv) + 1)$} scaling operations along $uv$, where we recall that $q = \sum_{s \in s(uv)} \lambda_s$.

Likewise, for every edge $u'v' \in E(H)$ with $u'v' \neq uv$, we must also scale at least \NewText{$\eta^2q(u'v')$} disks due to \Cref{obs:cluster-hardness-scaling-on-edge}.
Summing all up, this amounts to \NewText{$\Size{T(r)} \geq \eta^2\sum_{e \in E(H)}q(e) + \eta^2$} scaling \NewText{operations}.	
Recall \NewText{$k_{\text{fix}} = \eta^2\sum_{e \in E(H)}q(e)$} and observe $\kappa \leq \eta$.
Thus, $\Size{T(r)} \geq k_{\text{fix}} + \eta^2 > k_{\text{fix}} + \kappa = k$.
Since any solution, and in particular~$r$, can scale at most $k$ disks,
this implies a contradiction with the existence of $r$: either $r$ scales strictly more than $k$ disks or there exists at least one $P_3$ in $G(\mathcal{S}, r)$ along an edge of $H$.
\end{proof}

\onlyLong{We show in \Cref{fig:cluster-hardness-example-combined}b a small (simplified) example of our construction.
Using \Cref{obs:cluster-hardness-scaling-on-edge,lem:cluster-hardness-scaling-on-edge}, we now establish correctness of our reduction.
}
\end{statelater}

\begin{restatable}\restateref{prop:cluster-hardness-shrink}{proposition}{propositionClusterHardnessShrink}
\label{prop:cluster-hardness-shrink}
\Cluster is \NP-hard for all fixed rational $0 < \rmin < 1$.
\end{restatable}
\begin{prooflater}{ppropositionClusterHardnessShrink}
Let $\instance_{\probname{VC}} = (H, \kappa)$ be an instance of \VC where $H$ is a cubic and planar.
Furthermore, for a fixed rational $0 < \rmin < 1$, let $\instance = (\mathcal{S}, \rmin, \rmax, k = k_{\text{fix}} + \kappa)$ be the instance of \Cluster  (for an arbitrary fixed rational $\rmax \geq \rmin$) constructed as described above.
We can observe that $\instance$ can be constructed in time that is polynomial in the size of $\instance_{\probname{VC}}$.
Since we placed all our disks on rational coordinates that depend on the value of $\rmin$, the size of $\instance$ is also polynomial in the size of $\instance_{\probname{VC}}$, and it remains to show correctness of the reduction.

\proofsubparagraph*{($\boldsymbol{\Rightarrow}$)}
Assume that $\instance_{\probname{VC}}$ is a yes-instance of \VC and let $C \subseteq V(H)$ be a vertex cover of $H$ of size $\Size{C} \leq \kappa$.
We now construct a radius assignment $r$ that is a solution to~$\instance$.
To this end, we first set $r(p_v) = \rmin$ for every $v \in C$.
For every edge $uv \in E(H)$, we have $u \in C$ or $v \in C$.
Without loss of generality, assume that $u \in C$ holds and consider the disks $p_u, p_1, p_2, p_3, \ldots, p_{3q}, p_v$ along the edge $uv$; recall $q = \sum_{s \in s(e)} \lambda_s$.
Second, we set $r(p_{3i}^{\eta^2}(j)) = \rmin$ for every copy $p_{3i}^{\eta^2}(j)$ of the $\eta^2$-heavy disk $p_{3i}$ for every $j\in [\eta^2]$ and $i \in [q]$.
Third, we set $r(p) = 1$ for all remaining disks $p \in \mathcal{S}$; see also \Cref{fig:cluster-hardness-example-combined}b.
Observe that for every edge $uv \in E(H)$, we scale exactly $\eta^2\sum_{s\in s(uv)} \lambda_s$ disks (ignoring $p_u$ and $p_v$), thus matching the lower bound from \Cref{obs:cluster-hardness-scaling-on-edge}.
Overall, this amounts to $\eta^2\sum_{e \in E(H)}\sum_{s \in s(e)} \lambda_s + \kappa$ scaling operations, thus $\Size{\mathcal{T}(r)} \leq k$.
What remains to show is that $G(\mathcal{S}, r)$ is a cluster graph, i.e., does not contain an induced $P_3$.
Towards this, consider again the edge $uv \in E(H)$ and the disks $p_u, p_1, p_2, p_3, \ldots, p_{3q}, p_v$ along the edge $uv$ and recall that we assume, without loss of generality, $u \in C$.
We observe that for every $i \in [q]$, at least one of the disks $p_{3i - 2}$, $p_{3i - 1}$, and $p_{3i}$ is $\eta^2$-heavy and all of its copies are scaled.
Furthermore, also from the three disks $p_{u'}^{\theta}$, $p_u$, and $p_1$, the disk $p_u$ is scaled, where $p_{u'}^{\theta}$ is the $\theta$-heavy disk that we placed next to~$p_u$; recall \Cref{fig:cluster-hardness-enlarge-shrink-edge}a and \Cref{fig:cluster-hardness-example-combined}b.
Similarly, also from $p_{3q}$, $p_v$, and $p_{v'}^{\theta}$, all copies of the $\eta^2$-heavy disk $p_{3q}^{\eta^2}$ are scaled.
Hence, there is no induced $P_3$ along the edge $uv$ (including its endpoints).
Applying the argument to every edge of $H$, we conclude that $G(\mathcal{S}, r)$ has no induced $P_3$, i.e., is a cluster graph.
Consequently, $r$ is a solution with $\Size{\mathcal{T}(r)} \leq k$.

\proofsubparagraph*{($\boldsymbol{\Leftarrow}$)}
Assume that $\instance$ is a yes-instance of \Cluster and let $r$ be a solution, i.e., a radius assignment such that $\Size{\mathcal{T}(r)} \leq k$ and $G(\mathcal{S}, r)$ is a cluster graph.
We construct \NewText{$C = \{v \in V(H) : p_v \in \mathcal{T}(r)\}$} and show now that it is a vertex cover of $H$ of size at most $\kappa$.

Recall that by \Cref{obs:cluster-hardness-scaling-on-edge}, we must scale at least $\eta^2q$ disks for every edge $e \in E(H)$, with $q = \sum_{s \in s(e)} \lambda_s$.
Summing this up for all $e \in E(H)$, we scale at least $k_{\text{fix}}$ disks along the edges of $H$ alone.
Since $\Size{\mathcal{T}(r)} \leq k = k_{\text{fix}} + \kappa$, we conclude $\Size{C} \leq \kappa$.

Second, to show that $C$ is a vertex cover of $H$, assume that this is not the case.
Thus, there exists an edge $uv \in E(H)$ such that $u,v \notin C$.
Consider the disks $p_u, p_1, p_2, p_3, \ldots, p_{3q}, p_v$ along the edge $uv$.
We now show that $G(\mathcal{S}, r)$ contains a $P_3$ involving some of these disks.
First, consider the disks $p_{u'}^{\theta}$, $p_u$, and $p_1$, where we recall that $p_1$ is actually \NewText{an} $\eta^2$-heavy disk~$p_1^{\eta^2}$.
Observe that they form a series of $P_3$s, each involving a $p_u$ and a copy of $p_{1}^{\eta^2}$ and~$p_{u'}^{\theta}$, respectively.
Since $\eta^2 \leq k < \theta$ and $u \notin C$, we conclude that $\mathcal{S}[p_{1}] \subseteq T(r)$ must hold; recall also \Cref{obs:cluster-hardness-scaling-on-edge}.
Now, consider the disks $p_1$, $p_2$, and $p_3$, and recall they are $\eta^2$- and $\theta$-heavy, respectively, i.e., we have $p_1^{\eta^2}$, $p_2^{\theta}$, and $p_3^{\eta^2}$.
By \Cref{obs:cluster-hardness-scaling-on-edge}, we know that $\mathcal{S}[p_{1}] \subseteq T(r)$ or $\mathcal{S}[p_{3}] \subseteq T(r)$ must hold.
However, \Cref{lem:cluster-hardness-scaling-on-edge} tells us that $\mathcal{S}[p_{3}] \not\subseteq T(r)$ must hold, since we have already concluded that $\mathcal{S}[p_{1}] \subseteq T(r)$ holds.
Hence, there exists at least one copy $p_3^{\eta^2}(j)$, for $j \in [\eta^2]$ such that $r(p_3^{\eta^2}(j)) = 1$ holds.
Together with the (heavy) disks $p_4^{\eta^2}$ and $p_5^{\theta}$, they form $P_3$ and to remove them, we must scale all copies of $p_4^{\eta^2}$ by similar arguments as above.
Repeating above arguments, we conclude that there exists one copy $p_6^{\eta^2}(j')$, for $j' \in [\eta^2]$ such that $r(p_6^{\eta^2}(j')) = 1$ holds.
By iterative application, we conclude that there exists at least one disk $p_{3q}^{\eta^2}(j'')$, for $j'' \in [\eta^2]$ such that $r(p_{3q}^{\eta^2}(j'')) = 1$ holds.
Consider now the disks $p_{3q}^{\eta^2}(j'')$, $p_v$, and $p_{v'}^{\theta}$.
Since $v \notin C$, we have $r(p_v) = 1$.
Furthermore, since $k < \theta$, there exists at least one copy $p_{v'}^{\theta}(i)$, $i \in [\theta]$, such that $r(p_{v'}^{\theta}(i)) = 1$ holds.
These three disks form a $P_3$ in $G(\mathcal{S}, r)$, contradicting the assumption that $r$ is a solution.
\end{prooflater}

\onlyLong{
Note that the construction behind \Cref{prop:cluster-hardness-shrink} heavily uses the fact that $\rmin < 1$ holds.
In the following section, we present a different reduction for the case $1 \leq \rmin$ (and $\rmax > 1$), thus complementing \Cref{prop:cluster-hardness-shrink}.
}

\begin{figure}
\centering
\includegraphics[page=1]{cluster-hardness-example-combined}
\caption{
\textbf{\textsf{(a)}} A graph $H$ and the corresponding instances of \Cluster for \textbf{\textsf{(b)}} $\rmin < 1$ and \textbf{\textsf{(c)}} $\rmin > 1$. 
If a disk representing a vertex is red, then the vertex belongs to a solution to the instance \textbf{\textsf{(b)}} $(H, 2)$ of \VC and \textbf{\textsf{(c)}} $(H, 3)$ of \IS.
Note that we can always slightly shift some disks to make space along an edge of $H$ if required.
We simplified $H$ and the constructed instances to ease readability.
Colors have the same meaning as in \Cref{fig:cluster-hardness-enlarge-shrink-edge}.
}
\label{fig:cluster-hardness-example-combined}
\end{figure}

\onlyShort{
\subparagraph*{Hardness for Enlarging Disks ($\mathbf{1 \leq \rmin}$).}

We fix a rational value $1 < \rmin$ (the case $\rmin = 1$ is handled in the end).
Let $\instance_{\probname{IS}} = (H, \kappa)$ be an instance of \IS, where $H$ is a planar cubic graph with $\Size{V(H)} = \eta$; also \IS remains \NP-hard in this case~\cite{GJ.CIG.1979}.
We start with an appropriately scaled planar rectilinear embedding $\Gamma^{\gamma}$ 
of~$H$, where the scaling factor $\gamma > 1$ depends on $\rmin$.
To represent an edge $uv \in E(H)$, we use a series of disjoint %
$\eta^2$-$\theta$-heavy $P_3$s. %
The distance between endpoints of each heavy $P_3$ is $2\xi = \min(\rmin, 1.5) + 1$
and the spacing $d$ between consecutive heavy $P_3$s along $uv$ satisfies $1 + \rmin < d \leq 2\rmin$.
We also maintain a distance of $2\rmin$ to the start- and endpoint of~$uv$; see \Cref{fig:cluster-hardness-enlarge-shrink-edge}b.
By \Cref{obs:cluster-hardness-heavy}, we must scale all copies of one endpoint of each %
heavy~$P_3$.
Since $d \leq 2\rmin$, we must consistently scale either the left or the right endpoint of each heavy~$P_3$ along $uv$.
Otherwise, two $\theta$-heavy disks belong to the same cluster but do not intersect; see the dashed disks in \Cref{fig:cluster-hardness-enlarge-shrink-edge}b.
As $\theta > k$, this cannot lead to a solution. %
We represent each vertex $v \in V(H)$ by an extra disk $p_v$. %
Since $H$ is cubic, there is a free ``direction'' next to $v$ in $\Gamma^{\gamma}$.
There, we place an additional $\theta$-heavy disk and a $2$-heavy disk $p_{\lnot v}^{2}$ that together with $p_v$ form a heavy $P_3$; see \Cref{fig:cluster-hardness-enlarge-shrink-edge}b.
To remove this $P_3$, we have $p_v \in \mathcal{T}(r)$ or $\mathcal{S}[p_{\lnot v}] \subseteq \mathcal{T}(r)$, encoding the choice of having $v$ in an independent set $C$ of $H$, i.e., $v\in C$ if and only if $p_v \in \mathcal{T}(r)$.
Scaling $p_u$ and $p_v$ for an edge $uv \in E(H)$ results in ``merging'' two disjoint $\theta$-heavy disks into one connected component, which cannot lead to a solution.

To complete the construction, we set $k = k_{\text{fix}} + \kappa + 2(\eta - \kappa) = k_{\text{fix}} + 2\eta - \kappa$ and $\theta = k + 1$.
The value of $k_{\text{fix}}$ denotes the number of scaling operations required to remove all heavy $P_3$s along the edges of $H$.
Moreover, the term $2(\eta - \kappa)$ accounts for the $\eta - \kappa$ vertices $v \notin C$ for which we have to scale (both copies of) $p_{\lnot v}^2$.
Observe that enlarging a disk beyond $\rmin$ is never helpful, as we eventually connect two $\theta$-heavy disks which cannot lead to a solution; recall $\theta > k$.
Therefore, we can use any fixed rational $\rmax \geq \rmin$ in the above construction.

}

\begin{statelater}{clusterHardnessEnlargeDetails}
\subsection{\textsf{NP}-hardness Construction for Enlarging Disks ($\mathbf{1\leq \rmin} < \rmax$)}
\label{sec:cluster-hardness-enlarge}
In this section, we show that \Cluster remains \NP-hard for any fixed rational $1 \leq \rmin$. %
Note that by the definition of cluster graphs, the problem can be trivially decided in polynomial time for the case $\rmin = \rmax = 1$.
In the following, we describe the hardness construction for the case $1 < \rmin$.
In \Cref{sec:cluster-hardness-combined}, we discuss how to adapt it for $\rmin = 1$ and $1 < \rmax$.

Recall that in our hardness construction from \Cref{sec:cluster-hardness-shrink} we used overlapping (heavy)~$P_3$s.
To remove them, we had to (consistently) shrink one of the two end vertices of the heavy $P_3$.
Now, for $1 < \rmin$, we can no longer shrink disks but must enlarge them.
This requires us to use a different building block that internally again uses (heavy) $P_3$s, however, this time they are isolated in $G(\mathcal{S}, \Unit)$.
By choosing the right distance between these $P_3$s, we can again use them to forward information along an edge; in our case, which vertices are part of an \emph{independent set}.
In the following, we formalize the above-sketched idea; this time, we directly construct a correctly-scaled instance.

Let us fix a rational, positive value $1 < \rmin$.
We define $\alpha \coloneqq \min(\rmin, 1.5) + 1$, $\beta \coloneqq 2\rmin$, and \NewText{$\mu = \lceil\frac{\beta}{\min(\alpha, \rmin - 1)}\rceil + 1$}.
Let $\instance_{\probname{IS}} = (H, \kappa)$ be an instance of \IS, where~$H$ is a planar cubic graph on $\eta = \Size{V(H)}$ vertices.
As for \VC, \IS remains \NP-hard in this case~\cite{GJ.CIG.1979}.
Furthermore, let $\Gamma$ be a rectilinear planar embedding of~$H$ of polynomial area and where each edge has at most three bends; we recall that $\Gamma$ can be computed in polynomial time~\cite{LMS.LA2.1998}.
Similar to \Cref{sec:cluster-hardness-shrink}, we let $\Gamma^{\gamma}$ denote the drawing $\Gamma$ scaled by an integer $\gamma \geq 1$ such that (1) the vertical/horizontal distance between any pair of vertical/horizontal segments is greater than $2(\alpha + \beta)$ and (2) any segment $s$ has an integer length $\Dist{s} > 0$ which is divisible by $2(\mu + 1) (\alpha + \beta)$.

\subparagraph*{Constructing the Instance.}
Consider an edge $e = uv \in E(H)$, a segment $s \in s(e)$ and order the endpoints of $s$ along the $u$-to-$v$ traversal of $e$.
Observe that $\Dist{s} = \lambda_s(\alpha + \beta)$ must hold for some integer $\lambda_s \geq 2(\mu + 1)$.
Similar to the construction for the previous reduction, we consider two cases, depending on whether $v$ is the second endpoint of $s$.

In the first case, $v$ is not the second endpoint of $s$.
We place $3\lambda_s$-many points along $s$, let them be $p_{s,1}, p_{s,2}, \ldots, p_{s,3\lambda_s}$.
The disk $p_{s,1}$ is placed a distance of $\beta$ away from the first endpoint of $s$.
The remaining disks are distributed along $s$ such that $\Dist{p_{s, i + 1} - p_{s,i}} = d$ for $i \in [3\lambda_s - 1]$, where $d = \beta$ if \NewText{$i = 3j$} for some $j \in [\lambda_s - 1]$ and $d = \alpha / 2$ otherwise.
Observe that the distance between $p_{s,3\lambda_s}$ and the first endpoint of $s$ equals $\beta + \lambda_s \cdot \alpha + (\lambda_s - 1)\beta = \lambda_s \cdot (\alpha + \beta)$.
Later, we will introduce $\lambda_s$-many $\delta$-$\theta$-heavy $P_3$s along $s$, each composed of three consecutive disks.
On the first endpoint on $s$, we will either place a disk representing the vertex $u$, if $u$ is the first endpoint of $s$, or we have already placed the last disk from a different segment there.

In the second case, $v$ is the second endpoint of $s$.
Observe that in the above-described first case, we place the last disk at the second endpoint of $s$.
This time, we want to be able to later place there an additional disk $p_v$, which will represent the vertex $v$.
Moreover, the distance between $p_v$ and the last disk on $s$ should be $\beta$, i.e., the same distance as the first disk has to the first endpoint of $s$.
To this end, we use a similar, but this time conceptually simpler, trick as in \Cref{sec:cluster-hardness-shrink}.
We place the first three disks $p_{s,1}$, $p_{s,2}$, and $p_{s,3}$ as above, i.e., the distance between $p_1$ and the first endpoint of $s$ is $\beta$ and the distance between the remaining disks is $\alpha / 2$ each.
Afterwards, we place the next $3\mu$ disks $p_{s,4}, \ldots, p_{s,3(\mu+1)}$ such that $\Dist{p_{s, i} - p_{s,i - 1}} = d$ for $4 \leq i \leq 3(\mu+1)$, where $d = \beta - \frac{2\rmin}{\mu}$ if $i = 3j + 1$ for some $j \in [\mu]$ and $d = \alpha / 2$ otherwise.
Observe that we have \NewText{$2\rmin - \alpha < \beta - \frac{2\rmin}{\mu}$}.
\NewText{To see this, observe that $2\rmin - \alpha < \beta - \frac{2\rmin}{\mu}$ is equivalent to $\frac{\beta}{\alpha} < \mu$ and we have $\mu \geq \frac{\beta}{\alpha} + 1$.
This guarantees that if we (consistently) scale the leftmost or rightmost heavy disk of two adjacent~$P_3$s, then the scaled disks will not intersect.
We also have $\rmin + 1 < \beta - \frac{2\rmin}{\mu}$, which is equivalent to $\frac{\beta}{\rmin - 1} < \mu$ and we have $\mu \geq \frac{\beta}{\rmin - 1} + 1$.
This ensures that a scaled disk of a heavy $P_3$ does not intersect disks of another heavy $P_3$ which are not scaled.}
Moreover, observe that the distance between $p_{s,3(\mu+1)}$ and the first endpoint on $s$ equals $\beta + \alpha + \mu\alpha + \mu\cdot (\beta - \frac{2\rmin}{\mu}) = (\mu + 1)(\alpha + \beta) - \beta$, since $\beta = 2\rmin$.
As $\lambda_s \geq 2(\mu + 1)$, we can complete the construction by placing $ 3\cdot(\lambda_s - (\mu + 1))$ further disks \NewText{$p_{s,3(\mu+1) + 1}, \ldots, p_{s, 3\lambda_s}$} along~$s$.
This time, we use the same construction as in the first case i.e., $\Dist{p_{s,3(\mu+1) + 1}, p_{s,3(\mu+1)}} = \beta$, and $\Dist{p_{s, 3(\mu+1) + i + 1} - p_{s,3(\mu+1) + i}} = d$ for \NewText{$i \in [3(\lambda_s - (\mu+1)) - 1]$}, where $d = \beta$ if \NewText{$i = 3j$} for some $j \in [\lambda_s - 1], j \geq 3(\mu+1)$, and $d = \alpha / 2$ otherwise.
The distance between $p_{s,1}$ and \NewText{$p_{s, 3\lambda_s}$} equals
$(\mu + 1)(\alpha + \beta) - \beta + (\lambda_s - (\mu + 1))\cdot(\alpha + \beta) = \lambda_s \cdot (\alpha + \beta) - \beta$.
We repeat this process for every segment $s$ of the edge~$uv$ and every edge $e \in E(H)$.
Afterwards, we add one disk $p_v$ for every vertex $v \in V(H)$ which \emph{represents} the vertex $v$.
We place $p_v$ on the position~$\Gamma^{\gamma}(v)$ of $v$ in the scaled drawing $\Gamma^{\gamma}$.  
In the end, \NewText{we} will use the following equivalence between an independent set $C \subseteq V(H)$ of $H$ and the radius assignment $r(p_v)$ of $p_v$: $v \in C$ if and only if $r(p_v) \neq 1$.

The above construction forms the base layout of our reduction.
Similar to \Cref{sec:cluster-hardness-shrink}, we have to make two adaptions to the base layout.
First, for every edge $uv \in E(H)$, let \NewText{$p_u, p_1, p_2, p_3, \ldots, p_{3q(uv)}, p_v$} be the disks along the edge, where we recall that the disks $p_u$ and~$p_v$ represent the vertices $u$ and $v$, respectively\NewText{, and we define $q(uv) \coloneqq \sum_{s\in s(uv)} \lambda_s$}.
Observe that the disks form \NewText{$q(uv)$} disjoint $P_3$s, separated by a distance of \NewText{at most $\beta = 2\rmin$}.
For every \NewText{$i \in [q(uv)]$}, we turn the disk~$p_{3i - 1}$ into a $\theta$-heavy disk~$p_{3i - 1}^{\theta}$ and the disks $p_{3i}$ and $p_{3i - 2}$ into $\eta^2$-heavy disks~$p_{3i}^{\eta^2}$ and~$p_{3i - 2}^{\eta^2}$, respectively, where recall $\eta = \Size{V(H)}$ and specify $\theta > k$ in the end.
This transforms the above-mentioned $P_3$s into $\eta^2$-$\theta$-heavy $P_3$s, this time using $\xi = \alpha / 2$.
Since $\Dist{p_{3i}-p_{3i-2}} > 0$ for every \NewText{$i \in [q(uv)]$}, we must scale all copies of $p_{3i}^{\eta^2}$ or $p_{3i - 2}^{\eta^2}$ to remove the~$P_3$, which leads to the following observation; see also \Cref{fig:cluster-hardness-enlarge-shrink-edge}b.
\begin{observation}
\label{obs:cluster-hardness-enlarging-on-edge}
Let $uv \in E(H)$ be an edge of $H$ and let \NewText{$p_u, p_1, p_2, p_3, \ldots, p_{3q(uv)}, p_v$}, %
be the disks along $uv$.
In every solution $r$, we have $\mathcal{S}[p_{3i - 2}] \subseteq T(r)$ or $\mathcal{S}[p_{3i}] \subseteq T(r)$.
Moreover, this amounts to scaling at least \NewText{$\eta^2q(uv)$} disks along the edge \NewText{$uv$}.
\end{observation}

Second, for every vertex $v \in V(H)$, let $p_v$ be the disk that represents it.
As before, we make the observation that there is exactly one direction of top, bottom, left, and right without an edge incident to $v$ in $\Gamma^{\gamma}$.
For the following description, say that no edge leaves~$v$ to the left.
We place a $\theta$-heavy disk $p_{v'}^{\theta}$ exactly $\alpha/2$ to the left of $p_v$.
Moreover, we place $\alpha/2$ left of $p_{v'}^{\theta}$ a $2$-heavy disk $p_{\lnot v}^{2}$; see also \Cref{fig:cluster-hardness-enlarge-shrink-edge}b.
For the other cases we perform symmetrically.
The disks $p_{\lnot v}^{2}$, $p_{v'}^{\theta}$, and $p_v$ form a $P_3$.
To remove this $P_3$, we have to scale $p_v$ or both copies of $p_{\lnot v}^{2}$; recall that we will set $\theta > k$.
In particular, scaling the disk $p_v$ will correspond to including the vertex $v \in V(H)$ in an independent set of $H$.
This completes the construction of the point set $\mathcal{S}$.
In the following, we briefly summarize the construction before we show its correctness.

\subparagraph*{The Complete Reduction.}
Recall that $\instance_{\probname{IS}} = (H, \kappa)$ is an instance of \IS where $H$ is a cubic and planar graph with $\eta = \Size{V(H)}$.
We now specify the value of $\theta$, $\rmin$,~$\rmax$, and $k$ to complete the construction of an instance $\instance$ of \Cluster.
To this end, recall \Cref{obs:cluster-hardness-enlarging-on-edge}, which states that we must scale at least $\eta^2\sum_{s \in s(e)} \lambda_s$ disks along every edge $e \in E(H)$, where we recall that the length of a segment $s$ in the drawing $\Gamma^{\gamma}$ is $\Dist{s} = \lambda_s(\alpha + \beta)$ for some integer \NewText{$\lambda_s \geq 2(\mu + 1)$}.
As a consequence, we defined \NewText{$k_{\text{fix}} \coloneqq \eta^2\sum_{e \in E(H)}q(e)$}.
Moreover, since we are interested in an independent set $C \subseteq V(H)$ of size (at least) $\kappa$, and for every $v \in V(H)$ we have $v \in C$ if and only if $p_v \in \mathcal{T}(r)$ (and hence $v \notin C$ if and only if $\mathcal{S}[p_{\lnot v}] \subseteq \mathcal{T}(r)$), we set $k = k_{\text{fix}} + \kappa + 2(\eta - \kappa) = k_{\text{fix}} + 2\eta - \kappa$.
This allows us to set $\theta = k + 1$.
Since enlarging a disk above $\rmin$ is never beneficial (eventually, we would have two $\theta$-heavy disks in the same connected component), we can choose any value $\rmax \geq \rmin$.
This completes the construction of the instance $\instance = (\mathcal{S}, \rmin, \rmax, k = k_{\text{fix}} + 2\eta - \kappa)$ of \Cluster, since we specified in the beginning a value $1 < \rmin$.

We now aim to strengthen \Cref{obs:cluster-hardness-enlarging-on-edge}.
Recall that $\alpha = \min(\rmin, 1.5) + 1$ and $\beta =~2\rmin$.
Let $uv \in E(H)$ be an edge of $H$ and let \NewText{$p_u, p_1, p_2, p_3, \ldots, p_{3q(uv)}, p_v$} be the disks along $uv$.
Observe that $p_1$, $p_2$, and $p_3$ and $p_4$, $p_5$, and $p_6$ form two $P_3$s in $G(\mathcal{S}, \Unit)$.
We have $\Dist{p_3 - p_1} = \Dist{p_6 - p_4} = \alpha$ and $\Dist{p_4 - p_3} \leq \beta$; recall also \Cref{fig:cluster-hardness-enlarge-shrink-edge}b.
\NewText{By \Cref{obs:cluster-hardness-enlarging-on-edge}, we must scale all $\eta^2$-copies of~$p_1^{\eta^2}$ or all of $p_3^{\eta^2}$ to remove the heavy $P_3$ on (the copies of) $p_1$, $p_2$, and $p_3$.}
Moreover, setting $r(p_i^{\eta^2}(j)) = \rmin$ for $i \in \{1,3\}$ and all $j \in [\eta^2]$ suffices, as $\rmin + 1 \geq \min(\rmin, 1.5) + 1 = \alpha= \Dist{p_3 - p_1}$.
The same holds for the $P_3$ on $p_4$, $p_5$, and $p_6$.
However, if we would have $\mathcal{S}[p_{3}], \mathcal{S}[p_{4}] \subseteq T(r)$, then the copies of $p_{3}^{\eta^2}$ and $p_4^{\eta^2}$ intersect; recall that we have $r(p_3^{\eta^2}(i)), r(p_3^{\eta^2}(j)) \geq \rmin$ and $\Dist{p_3^{\eta^2}(i) - p_3^{\eta^2}(j)} \leq 2\rmin$ for every $i,j \in [\eta^2]$.
Thus, the two $\theta$-heavy disks $p_2^{\theta}$ and $p_5^{\theta}$ are in the same connected component.
Since they cannot be (scaled to be) adjacent since $k < \theta$, this would imply the existence of a $P_3$ in $G(\mathcal{S},r)$.
\NewText{Consequently, if we scale (all copies of) $p_3$, we also have to scale (all copies of) $p_6$.
We can iterate this argument to conclude that $\mathcal{S}[p_{3}] \subseteq T(r)$ implies $\mathcal{S}[p_{3j}] \subseteq T(r)$ for all $j \in [q(uv)]$.
By similar arguments we can derive an analogous consequence of $\mathcal{S}[p_{3q(uv)} - 2] \subseteq T(r)$.}
Consequently, we can make the following observation, which strengthens \Cref{obs:cluster-hardness-enlarging-on-edge}.
\begin{observation}
	\label{obs:cluster-hardness-enlarging-on-edge-strong}
	Let $uv \in E(H)$ be an edge of $H$ and let \NewText{$p_u, p_1, p_2, p_3, \ldots, p_{3q(uv)}, p_v$} be the disks along $uv$.
	\NewText{For every solution $r$ it holds that $\mathcal{S}[p_{3}] \subseteq T(r)$ implies $\mathcal{S}[p_{3j}] \subseteq T(r)$ and $\mathcal{S}[p_{3q(uv)} - 2] \subseteq T(r)$ implies $\mathcal{S}[p_{3j - 2}] \subseteq T(r)$ for all $j \in [q(uv)]$.}
\end{observation}

\onlyLong{
We show in \Cref{fig:cluster-hardness-example-combined}c a small (simplified) example of our construction and are now ready to show correctness of the reduction.
}
\end{statelater}

\begin{restatable}\restateref{prop:cluster-hardness-enlarge}{proposition}{propositionClusterHardnessEnlarge}
\label{prop:cluster-hardness-enlarge}
\Cluster is \NP-hard for all fixed rational $1 < \rmin$.
\end{restatable}
\begin{prooflater}{ppropositionClusterHardnessEnlarge}
Let $\instance_{\probname{IS}} = (H, \kappa)$ be an instance of \IS where $H$ is a cubic and planar graph on $\eta = \Size{V(H)}$ vertices.
Furthermore, for a fixed $1 < \rmin$ (and chosen $\rmax \geq \rmin$), let $\instance = (\mathcal{S}, \rmin, \rmax, k = k_{\text{fix}} + 2\eta - \kappa)$ be the instance of \Cluster constructed as described above.
Note that $\instance$ can be constructed in polynomial time and it remains to show correctness of the reduction.

\proofsubparagraph*{($\boldsymbol{\Rightarrow}$)}
Assume that $\instance_{\probname{IS}}$ is a yes-instance of \IS and let $C \subseteq V(H)$ be an independent set of $H$ of size $\Size{C} \geq \kappa$.
We now construct a radius assignment $r$ that is a solution to $\instance$.
To this end, we first set $r(p_v) = \rmin$ for every $v \in C$ and $r(p_u^2(1)) = r(p_u^2(2)) = \rmin$ for every $u \notin C$.
For every edge $uv \in E(H)$, we have $u \notin C$ or $v \notin C$ (or both).
Without loss of generality, assume that $u \notin C$ holds and consider the disks \NewText{$p_u, p_1, p_2, p_3, \ldots, p_{3q(uv)}, p_v$} along the edge $uv$; recall \NewText{$q(uv) = \sum_{s \in s(e)} \lambda_s$}.
Second, we set $r(p_{3i - 2}^{\eta^2}(j)) = \rmin$ for every copy $p_{3i - 2}^{\eta^2}(j)$ of the $\eta^2$-heavy disk $p_{3i - 2}$ for every $j\in [\eta^2]$ and \NewText{$i \in [q(uv)]$}.
Third, we set $r(p) = 1$ for all remaining disks $p \in \mathcal{S}$; see also \Cref{fig:cluster-hardness-example-combined}c.
Observe that for every edge $uv \in E(H)$, we scale exactly $\eta^2\sum_{s\in s(uv)} \lambda_s$ disks (ignoring $p_u$ and $p_v$), thus matching the lower bound from \Cref{obs:cluster-hardness-enlarging-on-edge}.
Moreover, among the disks representing vertices of $H$, we scale (at least) $\kappa$ times disks $p_v$ and (at most) $\eta - \kappa$ times both copies of $p_{\lnot v}$.
Summing this up, this yields at most $\eta^2\sum_{e \in E(H)}\sum_{s \in s(e)} \lambda_s + \kappa + 2(\eta - \kappa)$ scaling operations, thus $\Size{\mathcal{T}(r)} \leq k$.
We now show that $G(\mathcal{S}, r)$ does not contain a $P_3$, i.e., is a cluster graph.
For every $v \in V(H)$, we observe that we either scale $p_v$ or both copies of $p_{\lnot v}^2$.
In both cases, we remove the $P_3$ that involves the disk $p_v$.
Consider now the edge $uv \in E(H)$ and the disks \NewText{$p_u, p_1, p_2, p_3, \ldots, p_{3q(uv)}, p_v$} along~$uv$.
Assume $u \notin C$; the case $v \notin C$ is symmetric.
Since we set $r(p_{3i - 2}^{\eta^2}(j)) = \rmin$ for every \NewText{$i \in [q(uv)]$}, $j \in [\eta^2]$, we remove all the $P_3$ along the edge~$uv$.
In particular, observe that we do not introduce a $P_3$ along the edge.
To this end, recall that $u \notin C$ implies $r(p_u) = 1$ and \NewText{$\Dist{p_{3q(uv)} - p_v} > \rmin + 1$}.
Applying the argument to every edge of $H$, we conclude that $G(\mathcal{S}, r)$ has no induced $P_3$, i.e., is a cluster graph.
Consequently,~$r$ is a solution with $\Size{\mathcal{T}(r)} \leq k$.

\proofsubparagraph*{($\boldsymbol{\Leftarrow}$)}
Assume that $\instance$ is a yes-instance of \Cluster and let $r$ be a solution, i.e., a radius assignment such that $\Size{\mathcal{T}(r)} \leq k$ and $G(\mathcal{S}, r)$ is a cluster graph.
We construct $C = \{v \in V(H) : p_v \in \mathcal{T}(r)\}$ and show in the following that it is \NewText{an} independent set of $H$ of size at least $\kappa$.

Recall that by \NewText{\Cref{obs:cluster-hardness-enlarging-on-edge}}, we must scale at least \NewText{$\eta^2q(e)$} disks for every edge $e \in E(H)$. %
Summing this up for all $e \in E(H)$, we scale $k_{\text{fix}}$ disks along the edges of $H$ alone.
Since $\Size{\mathcal{T}(r)} \leq k = k_{\text{fix}} + 2\eta - \kappa$, we have $2\eta - \kappa$ scaling operations remaining for the $P_3$s that involve disks $p_v$ representing vertices $v$ of $H$.
For every such disk~$p_v$, we observe that $p_v \notin \mathcal{T}(r)$ implies $p_{\lnot v}^2(1), p_{\lnot v}^2(2) \in \mathcal{T}(r)$.
As we have $\eta$ such $P_3$s, and the remaining budget is $2\eta - \kappa$, we can conclude $\Size{C} \geq \kappa$.

We now show that $C$ is an independent set of $H$ and assume, towards a contradiction, that there exists an edge $uv \in E(H)$ with $u,v \in C$.
Consider the disks \NewText{$p_u, p_1, p_2, p_3, \ldots, p_{3q(uv)}, p_v$} along the edge $uv$.
In the following, we argue that there exists \NewText{in $G(\mathcal{S}, r)$} a connected component involving $p_u$ or $p_v$ that is not a clique, implying the existence of a $P_3$ in $G(\mathcal{S}, r)$.

\NewText{
Observe that $u \in C$ implies $r(p_u) \geq \rmin$.
We now argue that we have $\mathcal{S}[p_{3}] \subseteq T(r)$.
Assume not.
By \Cref{obs:cluster-hardness-enlarging-on-edge}, this implies $\mathcal{S}[p_{1}] \subseteq T(r)$.
Recall $\Dist{p_u - p_1} \leq 2\rmin$.
Thus, for $\mathcal{S}[p_{1}] \subseteq T(r)$, we have that the two $\theta$-heavy disks $p_{u'}^{\theta}$ and $p_{2}^{\theta}$ are in the same connected component; observe that these disks already overlap with $p_u$ and $p_1$ in $G(\mathcal{S}, \Unit)$, respectively.
As $k < \theta$, we can find $i,i' \in [\theta]$ such that $r(p_{u'}^{\theta}(i)) = r(p_{2}^{\theta}(i')) = 1$.
Observe that $p_{u'}^{\theta}(i)$, $p_u$, $p_1$, and $p_{2}^{\theta}(i')$ are in the same connected component but $p_{u'}^{\theta}(i)p_{2}^{\theta}(i') \notin E(G(\mathcal{S}, r))$ as $\Dist{p_{u'}^{\theta}(i) - p_{2}^{\theta}(i')} > 2$.
Hence, this connected component is not a clique and thus witnesses the existence of a $P_3$ in $G(\mathcal{S}, r)$.
Consequently, $\mathcal{S}[p_{1}] \subseteq T(r)$ is not possible, and we must have $\mathcal{S}[p_{3}] \subseteq T(r)$.
By \Cref{obs:cluster-hardness-enlarging-on-edge-strong}, this implies $\mathcal{S}[p_{3q(uv)}] \subseteq T(r)$.
Since $v \in C$ implies $r(p_v) \geq \rmin$, we can conclude that the two $\theta$-heavy disks $p_{3q(uv)}^{\theta}$ and $p_{v'}^{\theta}$ are in the same connected component.
Thus, by the same arguments as above, we can find a $P_3$ in $G(\mathcal{S}, r)$, which contradicts the assumption that $r$ is a solution.
}
\end{prooflater}

\onlyShort{
\subparagraph*{Combining Both Constructions.}
\Cref{prop:cluster-hardness-shrink,prop:cluster-hardness-enlarge} together imply \NP-hardness of \Cluster for all fixed rationals $\rmin$ with $0 < \rmin < 1$ or $1 < \rmin$.
We give a simplified example in
\Cref{fig:cluster-hardness-example-combined}.
For $\rmin = 1$ and $\rmax > 1$, we can modify the $P_3$s' placement in the construction $\rmin > 1$.
For $\rmin = \rmax = 1$, we simply need to check if $G(\mathcal{S}, \Unit)$ is a cluster graph, which can be done in polynomial time due to their characterization. %
}

\begin{statelater}{clusterHardnessCombinedDetails}
\subsection{Combining Both Hardness-Constructions}
\label{sec:cluster-hardness-combined}
Observe that \Cref{prop:cluster-hardness-shrink,prop:cluster-hardness-enlarge} together establish \NP-hardness of \Cluster for all fixed rational $\rmin$ with  $1 < \rmin$ or $0 < \rmin < 1$.
For $\rmin = 1$ and $\rmax > 1$, we can use the construction from \Cref{sec:cluster-hardness-enlarge} with $\alpha = \min(\rmax, 1.5) + 1$, $\beta = 2\alpha-2$, and $\mu = \lceil\frac{\beta}{\alpha}\rceil$ to obtain an instance $\instance = (\mathcal{S}, \rmin, \rmax, k = k_{\text{fix}} + 2\eta - \kappa)$ of \Cluster.
Intuitively speaking, we %
treat \rmin\ as $\min(\rmax, 1.5)$.
Correctness can be shown with arguments analogous to \Cref{prop:cluster-hardness-enlarge} and two additional observations.
First, to remove \NewText{an} $\eta^2$-$\theta$-heavy $P_3$, the sum of radii of any pair of \NewText{disks in two $\eta^2$-heavy disks} must be at least \NewText{$\alpha > \rmin + 1 = 2$}.
Second, our budget $k$ is not sufficient to scale all copies of two $\eta^2$-heavy disks in \NewText{an} $\eta^2$-$\theta$-heavy $P_3$.
Hence, for one of the $\eta^2$-heavy disks, there exists at least one copy~$p$ with $r(p) > 1.5$.
This copy allows for an observation similar to \Cref{obs:cluster-hardness-enlarging-on-edge-strong}.
Hence, \Cluster remains \NP-hard for $\rmin = 1$ and $\rmax > 1$.%
\begin{remark}
\label{rem:cluster-hardness-enlarge}
\Cluster is \NP-hard for $\rmin = 1$ and all fixed rationals values $1 < \rmax$. 
\end{remark}
Finally, we would like to recall that all our arguments for \Cref{prop:cluster-hardness-shrink,prop:cluster-hardness-enlarge,rem:cluster-hardness-enlarge} are with respect to the (fixed) value for $\rmin$; the actual value for $\rmax$ was irrelevant.
More concretely, for \Cref{prop:cluster-hardness-shrink}, removing a $P_3$ can only be done via shrinking operations since $G(\mathcal{S}, \Unit)$ contains multiple $\theta$-heavy disks in the same connected component.
For \Cref{prop:cluster-hardness-enlarge} (and \Cref{rem:cluster-hardness-enlarge}), observe that scaling beyond $\rmin$ is never beneficial: In the worst case, we would connect two $\theta$-heavy disks, which can never be adjacent in $G(\mathcal{S}, r)$ due to $k < \theta$.
Hence, we can use any value for \rmax\ in the constructions from \Cref{sec:cluster-hardness-shrink,sec:cluster-hardness-enlarge} as long as we fulfill $\rmin \leq \rmax$.

Finally, observe that for $\rmin = \rmax = 1$, \Cluster is equivalent to checking if $G(\mathcal{S}, \Unit)$, which is our input graph, is a cluster graph.
This can be done in polynomial time due to the characterization via $P_3$-freeness.
Combining all arguments, we can thus finally conclude:
\end{statelater}
\begin{theorem}
\label{thm:cluster-np-hard-general}
\Cluster %
is \NP-hard for all fixed rationals $0 < \rmin \leq \rmax$ except for the case $\rmin = \rmax = 1$ where it can be solved in polynomial time.
\end{theorem}

\section{Scaling To Complete Graphs in Polynomial Time}
\label{sec:complete}
We now turn our attention to complete graphs.
Let us fix an instance $\instance=\instanceLong$ of \Complete and consider a hypothetical solution $r$.
Since additional edges can only be created by enlarging disks, we can assume without loss of generality that $\rmax > 1$ (otherwise, it suffices to test if $G(\mathcal{S}, \Unit)$ is already complete) and that $r(p) = \rmax$ for all $p \in \mathcal{T}(r)$.
Observe that for any pair $p,q \in \mathcal{S}$ with $pq \notin E(G(\mathcal{S}, \Unit))$, at least one of $p$ and $q$ must be in $\mathcal{T}(r)$. 
We now partition the set $E' = \binom{\mathcal{S}}{2}$ of all point pairs based on their distance:

\begin{itemize}
\item $E_0' \coloneqq \{pq \in E' : \Dist{p - q} \leq 2\}$, i.e., pairs that already form an edge in $G(\mathcal{S}, \Unit)$;
\item $E_1' \coloneqq \{pq \in E' : 2 < \Dist{p - q} \leq 1 + \rmax\}$, i.e., pairs where scaling one disk suffices;
\item $E_2' \coloneqq \{pq \in E' : 1 + \rmax < \Dist{p - q} \leq 2\rmax\}$, i.e., pairs where we must scale both~disks;
\item $E_{\bot}' \coloneqq \{pq \in E' : 2\rmax < \Dist{p - q}\}$, i.e., pairs that can never intersect.
\end{itemize}

Let $\mathcal{S}[E_i'] \coloneqq \bigcup_{pq \in E_i'} \{p, q\}$ for $i = 0, 1, 2$.
Furthermore, define $\mathcal{T}' \coloneqq \mathcal{S}[E_2']$. %
We observe:
\begin{observation}
\label{obs:scale-complete-trivial-observations}
It holds $\mathcal{T}' \subseteq \mathcal{T}(r)$.
Also, \instance is a no-instance if $E_{\bot}' \neq \emptyset$ or  $\Size{\mathcal{T}'} > k$.
\end{observation}

While scaling all disks in $\mathcal{T}'$ is necessary, it is not yet sufficient.
In particular, for $pq \in E_1'$ with $p,q \notin \mathcal{T}'$, the disks do not form an edge in $G(\mathcal{S}, r)$ unless we also scale $p$ or $q$. %
Let $\mathcal{X} \coloneqq \mathcal{S}[E_1'] \setminus \mathcal{T}'$ and $n_{\mathcal{X}} \coloneqq \Size{\mathcal{X}}$.
Consider the unit disk graph $G(\mathcal{X}, \Unit)$.
Every non-edge $pq \notin E(G(\mathcal{X}, \Unit))$ corresponds to one such pair $pq \in E_1'$ where we must still scale $p$ or $q$. %
Hence, we must select (at least) one endpoint from every non-edge in $G(\mathcal{X}, \Unit)$, which corresponds to a vertex cover (of size at most $k' \coloneqq k - \Size{\mathcal{T}'}$) in the complement graph $\overline{G(\mathcal{X}, \Unit)}$.
Computing this vertex cover is equivalent to finding a clique (of size $n_{\mathcal{X}} - k'$) in $G(\mathcal{X}, \Unit)$.\onlyShort{\footnote{Although slower, it would also suffice to find a clique of size $n - k'$ in the disk graph $G(\mathcal{S}, r')$ with $r'(p) = \rmax$ for $p \in \mathcal{T}'$ and $r'(p) = 1$ otherwise, which we can do in $\BigO{n^{6.5}}$ time~\cite{KM.MCP.2025,HK.nAM.1973}.}}
\begin{figure}
\centering
\includegraphics[page=1]{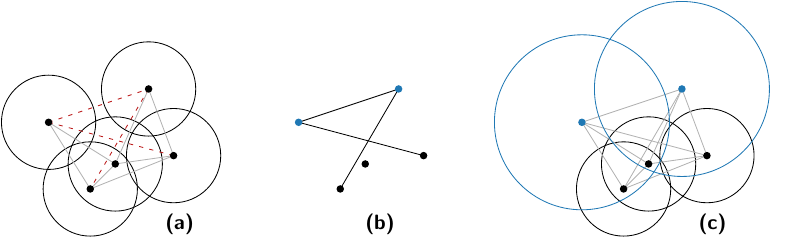}
\caption{
To turn \textbf{\textsf{(a)}} into a complete graph, we must scale at least one endpoint of every dashed non-edge. 
\textbf{\textsf{(b)}} A vertex cover (blue) of the complement graph \textbf{\textsf{(c)}} directly translates into a solution. 
}
\label{fig:complete-vertex-cover}
\end{figure}

\begin{restatable}\restateref{lem:scale-complete-clique}{lemma}{lemmaScalingCompleteClique}
\label{lem:scale-complete-clique}
\instance is a yes-instance if and only if %
\NewText{$E_{\bot}' = \emptyset$, $\Size{\mathcal{T}'} \leq k$, and}
$G(\mathcal{X}, \Unit)$ has a clique of size $n_{\mathcal{X}} - k'$.
\end{restatable}
\begin{prooflater}{plemmaScalingCompleteClique}
We argue both directions separately; see also \Cref{fig:complete-vertex-cover}.
Recall for the following arguments that 
$\mathcal{T}' = \mathcal{S}[E_2']$, 
$\mathcal{X} = \mathcal{S}[E_1'] \setminus \mathcal{T}'$,
$n_{\mathcal{X}} = \Size{\mathcal{X}}$, and $k' = k - \Size{\mathcal{T}'}$.

\proofsubparagraph*{($\boldsymbol{\Rightarrow}$)}
For the forward direction, we assume that \instance is a yes-instance, i.e., there exists a radius assignment $r$ such that $\Size{\mathcal{T}(r)} \leq k$ and $G(\mathcal{S}, r)$ is complete. %
Without loss of generality, we can assume that $\mathcal{T}(r)$ consists of exactly $k$ disks.
\NewText{This is because scaling up additional disks can never destroy a solution.
In particular, if $k > n$, then \instance is a yes instance if and only if $E_{\bot}' = \emptyset$ (in which case $\Size{\mathcal{T}'} \leq k$ trivially holds as $\Size{\mathcal{T}'} \leq n$).}
\NewText{By \Cref{obs:scale-complete-trivial-observations}, we have $E_{\bot}' = \emptyset$ and $\Size{\mathcal{T}'} \leq k$.}
\NewText{Assume $n_{\mathcal{X}} - k' > 0$; otherwise there is nothing to show.}
We now claim that the set $\mathcal{Y} \coloneqq \mathcal{X} \setminus \mathcal{T}(r)$ is a clique in $G(\mathcal{X}, \Unit)$ of size (at least) $n_{\mathcal{X}} - k'$.
The size of $\mathcal{Y}$ follows directly from its construction: 
By \Cref{obs:scale-complete-trivial-observations}, we know that $\mathcal{T}' \subseteq \mathcal{T}(r)$ holds.
Note that $\mathcal{X} \cap \mathcal{T}' = \emptyset$ by definition of $\mathcal{X}$, which implies $\mathcal{Y} = \mathcal{X} \setminus (\mathcal{T}(r) \setminus \mathcal{T}')$.
Since $\Size{\mathcal{T}(r)} = k$ and $k' + \Size{\mathcal{T}'} = k$, we have $\Size{(\mathcal{T}(r) \setminus \mathcal{T}')} = k'$.
Consequently, we obtain $\Size{\mathcal{Y}} \geq \Size{\mathcal{X}} - \Size{(\mathcal{T}(r) \setminus \mathcal{T}')} = n_{\mathcal{X}} - k'$.
Note that $\Size{\mathcal{Y}} > \Size{\mathcal{X}} - \Size{(\mathcal{T}(r) \setminus \mathcal{T}')}$ is possible due to disks $p \in \mathcal{T}(r)$ with $p \notin \mathcal{X} \cup \mathcal{T}'$, i.e., $p \in \mathcal{S}[E_0'] \setminus (\mathcal{S}[E_1'] \cup \mathcal{S}[E_2'])$.

We now show that $\mathcal{Y}$ is also a clique in $G(\mathcal{X}, \Unit)$.
Assume, for the sake of a contradiction, that there are two points $p,q \in \mathcal{Y}$ such that $pq$ is not an edge of $G(\mathcal{X}, \Unit)$.
We now observe that $2 < \Dist{p - q} \leq 1 + \rmax$ holds.
Otherwise, the disks $p$ and $q$ would either intersect in $G(\mathcal{S}, \Unit)$ or we have $p,q\in E_2'$, which implies $p,q\in\mathcal{T}'$ (and thus $p,q \notin\mathcal{X}$).
Both cases immediately yield a contradiction. %
As $pq$ is not an edge of $G(\mathcal{X}, \Unit)$, $p$ and $q$ do not intersect with radius one.
However, since $p,q \in \mathcal{Y}$, we have $p,q \notin \mathcal{T}(r)$, i.e., $r(p) = r(q) = 1$ in the solution.
This implies that they are not connected in $G(\mathcal{S}, r)$ either.
This is a contradiction to the assumption that $r$ is a solution.
Consequently, $\mathcal{Y}$ is a clique in $G(\mathcal{X}, \Unit)$.

\proofsubparagraph*{($\boldsymbol{\Leftarrow}$)}
For the backward direction, we assume \NewText{$E_{\bot}' = \emptyset$, $\Size{\mathcal{T}'} \leq k$, and} that there exists a clique~$\mathcal{Y}$ of size $\max(n_{\mathcal{X}} - k', 0)$ in $G(\mathcal{X}, \Unit)$.
We now show that the radius assignment $r$ with $r(p) = \rmax$ for all $p \in \mathcal{T}' \cup (\mathcal{X} \setminus \mathcal{Y})$ and $r(p) = 1$ otherwise is a solution. %
Regarding the size of $\mathcal{T}(r)$, we recall $\mathcal{T}' \cap \mathcal{X} = \emptyset$.
Consequently, we get $\Size{\mathcal{T}(r)} = \Size{\mathcal{T}'} + \Size{\mathcal{X} \setminus \mathcal{Y}} = (k - k') + (n_{\mathcal{X}} - (n_{\mathcal{X}} - k')) = k$.

We now show that $r$ is a solution and assume that this is not the case.
This means that there are two disks $p,q \in \mathcal{S}$ such that $pq$ is not an edge in $G(\mathcal{S}, r)$.
We now distinguish the following cases depending on the membership of $pq$ in the partition of $E'$. 
We immediately observe that \NewText{$pq \in E_0' \cup E_{\bot}'$} is not possible since otherwise $pq$ would trivially be an edge in $G(\mathcal{S}, r)$ or we would have rejected immediately using \Cref{obs:scale-complete-trivial-observations}, respectively\NewText{; recall that we assume $E_{\bot}' = \emptyset$.}
What remains are the cases $pq \in E'_1$ and $pq \in E'_2$.
Let us first analyze the latter case.
There, we have $p,q \in \mathcal{T}'$.
Thus, $r(p) = r(q) = \rmax$, i.e., we scale both disks, and since $\Dist{p - q} \leq 2\rmax$, $pq \notin E(G(\mathcal{S}, r))$ is not possible.
It remains to consider the case $pq \in E'_1$. %
By the definition of~$E_1'$, we have $pq \notin E(G(\mathcal{S},\Unit))$.
Since we have $pq \notin E(G(\mathcal{S},r))$ we get $p,q \notin \mathcal{T}(r)$.
Otherwise, i.e., if at least one disk would be in $\mathcal{T}(r)$, we get $pq \in E(G(\mathcal{S},r))$ by the definition of $E_1'$.
Moreover, as $p,q \notin \mathcal{T}(r)$, we also have $p,q \notin \mathcal{T}'$ and thus $p,q \in \mathcal{X}$.
However, $p,q \notin \mathcal{T}(r)$ implies $p,q \in \mathcal{Y}$ since we have $r(p) = \rmax$ for every $p \in \mathcal{T}' \cup (\mathcal{X} \setminus \mathcal{Y})$.
Now recall $pq \notin E(G(\mathcal{S},\Unit))$, which implies $pq \notin E(G(\mathcal{X},\Unit))$ either.
This means that $\mathcal{Y}$ is not a clique.
Since all cases led to a contradiction, we conclude that $r$ is a solution, i.e., \instance is a yes-instance.
\end{prooflater}
Using \Cref{lem:scale-complete-clique}, %
we can find the remaining $k'$ disks to scale via a sufficiently large clique in $G(\mathcal{X}, \Unit)$.
Computing a maximum clique in a unit disk graph with given geometric representation takes \BigO{n^{2.5} \log n} time~\cite{EKM.FMC.2023}.
All other steps \NewText{and the construction of the sets~$E_\bot'$ and $E_i'$, $i = 0,1,2$,} take $\BigO{n^2}$ time.
\onlyLong{Hence, \Complete can be solved in polynomial time.}
\begin{theorem}
\label{thm:complete}
\Complete can be solved in \BigO{n^{2.5} \log n} time.
\end{theorem}
\onlyLong{
Note that after determining the set $\mathcal{T}'$, it would also suffice to check for the existence of a clique of size $n - k'$ in the disk graph $G(\mathcal{S}, r')$ with $r'(p) = \rmax$ for $p \in \mathcal{T}'$ and $r(p) = 1$ otherwise.
Recently, an $\BigO{n^{6.5}}$-time algorithm for the \probname{Clique} problem on disk graphs with two different radii was established~\cite{KM.MCP.2025,HK.nAM.1973}. %
Although this yields a (conceptually) simpler polynomial-time algorithm, the running time of \Cref{thm:complete} supersedes it. 
}

\section{Scaling to Connected Graphs is \textsf{W}[1]-hard}
\label{sec:connected}
We now consider connected graphs and show that \Connected is \W[1]-hard when parameterized by the number $k$ of allowed scaling operations.
This rules out an \FPT-algorithm, and, in particular, generalizes the \W[1]-hardness by Fomin et al.~\cite{FG00Z25}:
Their hardness holds in the variant where only disks from a given set $\mathcal{A} \subseteq \mathcal{S}$ can be enlarged, i.e., $\mathcal{T}(r) \subseteq \mathcal{A}$ holds.
In contrast, our reduction allows any disk to be scaled, even to a fixed radius $\rmax$, which prevents a straightforward adaptation of Fomin et al.'s proof and instead requires a new construction.
\shortLong{W}{To this end, w}e reduce from the variant of \probname{Grid Tiling} introduced below.

Let \GridTiling[$\oplus$] be defined as follows.
Given two integers $\eta, \kappa \geq 1$, and a family $\mathcal{X}$ of $\kappa^2$ sets $X_{i,j} \subseteq [\eta]^2$, choose a tuple $x_{i,j} = (a, b) \in X_{i,j}$ for each $i,j\in[\kappa]$ such that if $x_{i,j} = (a,b)$ and $x_{i+1,j} = (a', b')$ then $a \oplus a'$, and if $x_{i,j+1} = (a'',b'')$ then $b \oplus b''$ for all (possible) $i,j \in [\kappa]$, i.e., the selected tuples fulfill the relation $\oplus$ on the first and second value, respectively.
\shortLong{We}{We let $\mathcal{Y} = \{x_{i,j} \in X_{i,j} : i,j \in [\kappa]\}$ denote the set of selected tuples, i.e., a \emph{solution}. Furthermore, we} call $X_{i,j}$ the \emph{tile} of the $i$th \emph{column} and $j$th \emph{row}; the deviation of the usual indexation facilitates the following description and is common for geometric reductions~\cite{CFK+.PA.2015}.
Two tiles $X_{i,j}$ and $X_{i',j'}$ are \emph{adjacent} if and only if either $\Size{i' - i} = 1$ or $\Size{j' - j} = 1$ (but not both). 
Textbook-variants use $\oplus \in \{=, \leq\}$ and are known to be \W[1]-hard parameterized by~$\kappa$~\cite[Theorems~14.28 and 14.30]{CFK+.PA.2015}.
One can also extend the hardness to the case where we use $>$ for $\oplus$.

\begin{restatable}\restateref{lem:our-grid-tiling-hardness}{lemma}{lemmaGridTilingHardness}
\label{lem:our-grid-tiling-hardness}
\GridTiling\ parameterized by $\kappa$ is \W\textup{[1]}-hard.
\end{restatable}
\begin{prooflater}{plemmaGridTilingHardness}
To establish \W[1]-hardness, we first reduce \GridTiling[$\leq$] to \GridTiling[$<$]; observe that in the latter problem variant we require strict inequality.
\begin{claim}
\label{claim:our-grid-tiling-hardness-intermediate}
\GridTiling[$<$] parameterized by $\kappa$ is \W[1]-hard.
\end{claim}
\begin{claimproof}
Let $\instance_{\leq} = (\eta, \kappa,\mathcal{X} = \{X_{i,j} : i,j \in [\kappa]\})$ be an instance of \GridTiling[$\leq$].
We construct an instance $\instance_{<} = (\eta^3 + \kappa, \kappa,\mathcal{X}' = \{X_{i,j}' : i,j \in [\kappa]\})$ of \GridTiling[$<$] where $X_{i,j}' = \{(a\eta^2 + i, b\eta^2 + j) : (a, b) \in X_{i,j}\}$, for every $i,j \in [\kappa]$.
The construction can be carried out in polynomial time and the parameter value does not change.
It remains to show correctness.

For the forward direction ($\Rightarrow$), let the tuples $\mathcal{Y} = \{x_{i,j} : i,j \in [\kappa]\}$ be a solution of $\instance_{\leq}$.
We construct a solution $\mathcal{Y}'$ of $\instance_{<}$ that consists of the tuples $x_{i,j}' = (a\eta^2 + i, b\eta^2 + j)$ for every $x_{i,j} = (a,b) \in \mathcal{Y}$, for every $i,j \in [\kappa]$.
As $\mathcal{Y}$ is a solution of $\instance_{\leq}$, we have that $x_{i,j} = (a,b)$ and $x_{i+1,j} = (a', b')$ implies $a \leq a'$ for all $i \in [\kappa - 1], j \in [\kappa]$.
Consequently, for $x_{i,j}'$ and $x_{i+1,j}'$, it holds %
$a\eta^2 + i \leq a'\eta^2 + i < a'\eta^2 + i + 1$.
Similar arguments for the second value show that $\mathcal{Y}'$ is indeed a solution of $\instance_{<}$.
For the backward direction ($\Leftarrow$), let the tuples $\mathcal{Y}' = \{x_{i,j}' : i,j \in [\kappa]\}$ be a solution of $\instance_{<}$.
We construct a solution $\mathcal{Y}$ of $\instance_{\leq}$ that consists of the tuples $x_{i,j} = (a, b)$ for every $x_{i,j}' = (a\eta^2 + i, b\eta^2 + j) \in \mathcal{Y}'$, for every $i,j \in [\kappa]$.
As $\mathcal{Y}'$ is a solution of $\instance_{<}$, we have that \NewText{$x_{i,j}' = (a\eta^2 + i, b\eta^2 + j)$} and \NewText{$x_{i+1,j}' = (a'\eta^2 + i + 1, b'\eta^2 + j)$} implies $a\eta^2 + i < a'\eta^2 + i + 1$ for all $i \in [\kappa - 1], j \in [\kappa]$.
This can be reformulated to $\eta^2(a - a') < 1$, which is only true if $a \leq a'$ (as we can assume $\eta \geq 1$).
Similar arguments for the second value show that $\mathcal{Y}$ is indeed a solution of $\instance_{\leq}$.
\end{claimproof}

We now use \Cref{claim:our-grid-tiling-hardness-intermediate} to show the statement from the lemma.
Let $\instance_{<} = (\eta, \kappa,\mathcal{X} = \{X_{i,j} : i,j \in [\kappa]\})$ be an instance of \GridTiling[$<$].
We construct an instance $\instance_{>} = (\eta, \kappa,\mathcal{X}' = \{X_{i,j}' : i,j \in [\kappa]\})$ of \GridTiling where $X_{i,j}' = \{(\eta + 1 - a, \eta + 1 - b) : (a, b) \in X_{i,j}\}$, for every $i,j \in [\kappa]$.
The construction can be carried out in polynomial time and the parameter value does not change.
As before, it remains to show correctness.

For the forward direction ($\Rightarrow$), let the tuples $\mathcal{Y} = \{x_{i,j} : i,j \in [\kappa]\}$ be a solution of $\instance_{<}$.
We construct a solution $\mathcal{Y}'$ of $\instance_{>}$ that consists of the tuples $x_{i,j}' = (\eta + 1 - a, \eta + 1 - b)$ for every $x_{i,j} = (a,b)$, for every $i,j \in [\kappa]$.
As $\mathcal{Y}$ is a solution of $\instance_{<}$, we have that $x_{i,j} = (a,b)$ and $x_{i+1,j} = (a', b')$ implies $a < a'$ for all $i \in [\kappa - 1], j \in [\kappa]$.
Consequently, $-a > -a'$ and thus $\eta + 1 - a > \eta + 1 - a'$.
Analogous arguments for the second value show that $\mathcal{Y}'$ is indeed a solution to $\instance_{>}$.
The backward direction ($\Leftarrow$) follows from symmetric arguments.	
\end{prooflater}

\newcommand{\showConnectedBaseLayoutFigure}{
\begin{figure}
\centering
\includegraphics[page=1]{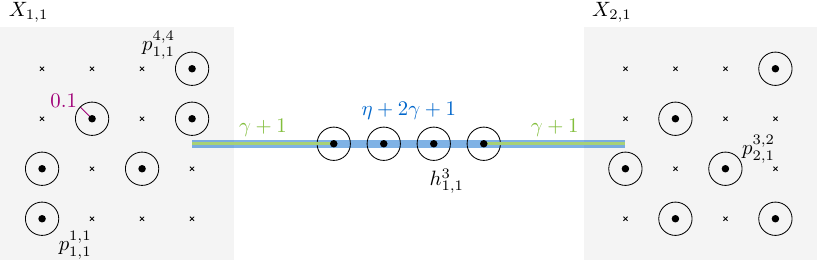}
\caption{
	Gadgets encoding the tiles and constraints. The disks \onlyLong{with gray background }represent the elements of \onlyLong{the tiles }$X_{1,1}$ and $X_{2,1}$\shortLong{, and}{. Moreover,} crosses indicate \onlyLong{potential }disk positions for elements not present in the respective tiles. 
}
\label{fig:connected-base-layout}
\end{figure}
}

\newcommand{\showConnectedDummyDisksFigure}{
\begin{figure}
\centering
\includegraphics[page=1]{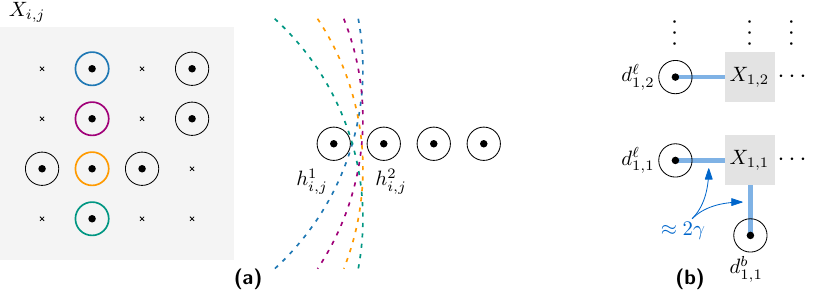}
\caption{
	\textbf{\textsf{(a)}} We must ensure \Cref{prop:connected-scale-h-disks} no matter which disk $p_{i,j}^{a,b}$ we scale, indicated by color, i.e., independent of the value of $a$ and $b$. 
	\textbf{\textsf{(b)}} The placement of the dummy disks.
}
\label{fig:connected-constraints}
\end{figure}
}

\shortLong{
\subsection{An Overview of the Construction}
\label{sec:connected-overview}
}{
\subparagraph*{An Overview of the Construction.}
We are now ready to describe the high-level idea of our reduction.
}
Let $\instance_{>} = (\eta,\kappa,\mathcal{X})$ be an instance of \GridTiling.
\NewText{We assume $\eta > 1$ and that $X_{i,j} \neq \emptyset$ for every $i,j \in [\kappa]$.}
The constructed instance \instance of \Connected closely resembles the (imagined) layout of $\instance_{>}$, i.e., contains a $\kappa \times \kappa$-grid of tiles, each containing a subset of the $\eta\times \eta$ grid.
To facilitate the following presentation, we rescale without loss of generality the plane such that \NewText{unit disks have radius~$0.1$}.
Our instance is composed of several gadgets of three different kinds.
\onlyShort{%
The first gadget represents a tile $X_{i,j} \in \mathcal{X}$ and contains a disk $p_{i,j}^{a,b}$ for every $(a,b) \in X_{i,j}$.
Scaling the disk~$p_{i,j}^{a,b}$ corresponds to choosing the tuple $x_{i,j} = (a,b)$.
The second gadget consists of a series of disks to ensure that the scaled disks, i.e., selected tuples, fulfill the $>$-constraints along the rows and columns of $\instance_{>}$.
Finally, the third gadget is a set of dummy disks that enforce that we scale one disk per tile.
}%
To preserve the equivalence between solutions, we ensure that our gadgets fulfill certain properties\onlyShort{ that we use in \Cref{sec:connected-correctness} to establish correctness}.

\shortLong{
\subparagraph*{Encoding the Tiles.}
Let $X_{i,j} \in \mathcal{X}$ be a tile.
We introduce one disk $p_{i,j}^{a,b}$ for every tuple $(a,b) \in X_{i,j}$ and set $x(p_{i,j}^{a,b}) = i(2\eta + 2\gamma) + a$ and $y(p_{i,j}^{a,b}) = j(2\eta + 2\gamma) + b$, where $\gamma$ is a \emph{displacement-factor} that we specify in the end.
Observe that we arrange the disks for each tile along an $\eta \times \eta$-integer grid and the tiles along a $\kappa \times\kappa$-grid.
We visualize this base layout in \Cref{fig:connected-base-layout} and provide further details in \Cref{sec:connected-base-layout}. 
Let $\mathcal{S}[X{i,j}]$ denote the disks representing the tuples of the tile $X_{i,j}$, $i,j \in[\kappa]$.
We observe that $G(\mathcal{S}[X{i,j}], \Unit)$ is edge-less; recall that one unit equals $0.1$.
To guarantee that our intended equivalence between the solutions is well-defined, we ensure that our final instance \instance has the following property.
}
{
The first gadget (\Cref{sec:connected-base-layout}) represents a tile $X_{i,j} \in \mathcal{X}$.
More specifically, we introduce one disk $p_{i,j}^{a,b}$ for every tuple $(a,b) \in X_{i,j}$ and place them in an $\eta \times \eta$-grid depending on the value of $i,j,a,b$.
This way, we arrange the disks for each tile along an $\eta \times \eta$-integer grid and the tiles along a $\kappa \times\kappa$-grid.
We visualize this base layout in \Cref{fig:connected-base-layout}. %
Let $\mathcal{S}[X_{i,j}]$ be the set of disks constructed for the tile $X_{i,j}$.
Scaling the disk $p_{i,j}^{a,b}$ should corresponding to choosing the tuple $x_{i,j} = (a,b)$: 
}
\begin{restatable}{property}{connectedScaleOnePerTileProperty}
\label{prop:connected-scale-one-per-set}
Let $r$ be a solution to \instance. We have $\Size{\mathcal{S}[X_{i,j}] \cap \mathcal{T}(r)} = 1$ for every $i,j \in [\kappa]$.
\end{restatable}
\medskip
\shortLong{
Moreover, to represent the (non-)adherence to the $>$-constraints, we ensure that two scaled disks $p_{i,j}^{a,b}$ and $p_{i',j'}^{a',b'}$ from different tiles intersect if and only if the corresponding tiles $X_{i,j}$ and $X_{i',j'}$ are adjacent and $a>a'$ (if $i' = i + 1$) or $b>b'$ (if $j' = j + 1$) holds:
}{
The second gadget (\Cref{sec:connected-adjacent-sets}) consists of a series of disks to ensure that the scaled disks, i.e., selected tuples, fulfill the $>$-constraints along the rows and columns of $\instance_{>}$.
More specifically, let $X_{i,j}$ and $X_{i + 1, j}$, $i \in [\kappa-1], j \in [\kappa]$ be two adjacent tiles.
We introduce $\eta$ disks $h_{i,j}^1, h_{i,j}^2, \ldots, h_{i,j}^{\eta}$ between the (disks for the) tiles $X_{i,j}$ and \NewText{$X_{i + 1,j}$}.
These disks are evenly spaced, maintaining a sufficient distance to the disks in the two tiles; see also \Cref{prop:connected-scale-h-disks,fig:connected-base-layout}. %
We introduce analogous disks $v_{i,j}^1, \ldots, v_{i,j}^{\eta}$ between the adjacent tiles $X_{i,j}$ and $X_{i, j + 1}$, $i \in [\kappa], j \in [\kappa-1]$.
Let $\mathcal{S}[>]$ be the above-introduced disks.
With our budget, we will ensure that we cannot scale any of these disks.
Thus, to ensure that the final disk graph is connected, we must scale disks in the adjacent tiles that fulfill the $>$-constraint; otherwise, at least one of the disks between the tiles remains isolated.
}
\begin{restatable}{property}{connectedScaleRelationsProperty}
\label{prop:connected-scale-relations}
Let $p_{i,j}^{a,b}, p_{i',j'}^{a',b'} \in \mathcal{S}$ be disks and $r$ a solution to \instance with $r(p_{i,j}^{a,b}) = \rmax$.
It holds:
\begin{enumerate}
\item if $i' = i$, $j' = j$, then the disks intersect;
\item if $i' = i + 1$, $j' = j$, then the disks intersect if and only if $p_{i',j'}^{a',b'} \in \mathcal{T}(r)$ and $a > a'$;
\item if $i' = i - 1$, $j' = j$, then the disks intersect if and only if $p_{i',j'}^{a',b'} \in \mathcal{T}(r)$ and $a' > a$;
\item if $i' = i$, $j' = j + 1$, then the disks intersect if and only if $p_{i',j'}^{a',b'} \in \mathcal{T}(r)$ and $b > b'$; 
\item if $i' = i$, $j' = j - 1$, then the disks intersect if and only if $p_{i',j'}^{a',b'} \in \mathcal{T}(r)$ and $b' > b$; and
\item if the tiles $X_{i,j}$ and $X_{i',j'}$ are not adjacent, then the disks do not intersect.
\end{enumerate} 
\end{restatable}
\medskip	
\shortLong{
\smallskip
Towards establishing \Cref{prop:connected-scale-one-per-set,prop:connected-scale-relations}, observe that %
the distance between any pair $p_{i,j}^{a,b}, p_{i,j}^{a',b'} \in \mathcal{S}[X_{i,j}]$ is at most $\sqrt{2}\eta$.
Hence, we ensure $\rmax \geq 2\eta$ \NewText{(for $\eta > 1$)}.
The distance between same-valued tuples in adjacent tiles is $2\eta+2\gamma$, and the corresponding disks should never intersect. %
Thus, we ensure $2\rmax<2\eta+2\gamma$.
This alone does not guarantee that~\instance admits only solutions $r$ that correspond to ones of $\instance_{>}$.
In particular, $G(\mathcal{S},r)$ could be connected although two scaled disks from adjacent tiles do not intersect.
We %
now prevent this.
\showConnectedBaseLayoutFigure

\subparagraph*{Ensuring the $\boldsymbol{>}$-Constraints.}
Let $X_{i,j}$ and $X_{i + 1, j}$, $i \in [\kappa-1], j \in [\kappa]$ be two adjacent tiles.
We introduce $\eta$ disks $h_{i,j}^1, h_{i,j}^2, \ldots, h_{i,j}^{\eta}$ between the (disks for the) tiles $X_{i,j}$ and $X_{i,j + 1}$.
These disks are evenly spaced, maintaining a distance of at least $\gamma+1$ to the disks in the two tiles; see \Cref{fig:connected-base-layout} and \Cref{sec:connected-adjacent-sets} for a formal construction.
We introduce analogous disks $v_{i,j}^1, \ldots, v_{i,j}^{\eta}$ between the adjacent tiles $X_{i,j}$ and $X_{i, j + 1}$, $i \in [\kappa], j \in [\kappa-1]$.
Let $\mathcal{S}[>]$ be the above-introduced disks.
Using the disks $\mathcal{S}[>]$, we aim to fulfill the following property.
}{
We furthermore aim to have the following property.
}
\begin{restatable}{property}{connectedScaleHDisksProperty}
\label{prop:connected-scale-h-disks}
Let $p_{i,j}^{a,b}$ be a disk and let $r$ be a solution to \instance such that $r(p_{i,j}^{a,b}) = \rmax$ and $\mathcal{S}[>]\cap\mathcal{T}(r)=\emptyset$.
For every $q, t \in [\eta]$, it holds:
\begin{enumerate}
\item if $1 \leq i < \kappa$, then $p_{i,j}^{a,b}$ and $h_{i,j}^q$ intersect if and only if $q < a$;
\item if $1 < i \leq \kappa$ then $p_{i,j}^{a,b}$ and $h_{i-1,j}^q$ intersect if and only if $q > a$;
\item if $1 \leq j < \kappa$, then $p_{i,j}^{a,b}$ and $v_{i,j}^t$ intersect if and only if $t < b$; and
\item if $1 < j \leq \kappa$, then $p_{i,j}^{a,b}$ and $v_{i,j-1}^t$ intersect if and only if $t > b$.
\end{enumerate}
\end{restatable}

\medskip	

\shortLong{
Assume for a moment $p_{i,j}^{a,b}, p_{i+1,j}^{a',b'} \in \mathcal{T}(r)$ with $a \leq a'$.
Neither of $p_{i,j}^{a,b}$ and $p_{i+1,j}^{a',b'}$ intersect~$h_{i,j}^q$ for $a \leq q \leq a'$ by \Cref{prop:connected-scale-h-disks}.
Combined with \Cref{prop:connected-scale-one-per-set} and an appropriately-chosen budget $k$, these $h_{i,j}^q$ witness that $G(\mathcal{S}, r)$ is not connected.
Observe that $x_{i,j} = (a,b)$ and $x_{i+1,j} = (a', b')$ with $a \leq a'$ can also not be part of a solution to $\instance_{>}$. %
Fulfilling \Cref{prop:connected-scale-h-disks} heavily depends on our values for $\gamma$, i.e., the displacement-factor, and $\rmax$.
In particular, we must select $\gamma$ such that the property is fulfilled no matter which disks $p_{i,j}^{a,b}$ we choose to scale, i.e., we must compensate the possible value of $a$ and $b$; see also \Cref{fig:connected-constraints}a.
Since $a,b\in [\eta]$, we can set $\gamma = \eta^2$.
Using geometry, we can compute the value for $\rmax$. %

\showConnectedDummyDisksFigure

\subparagraph*{The Dummy Disks.}
The above-introduced additional disks make establishing \Cref{prop:connected-scale-one-per-set} difficult.
In particular, we must ensure that no disk $p \in \mathcal{S}[>]$ is scaled. %
We introduce $4\kappa$ \emph{dummy disks} $d_{i,j}$, (at least) one for every tile $X_{i,j}$ with $i \in \{1, \kappa\}$ or $j \in \{1, \kappa\}$.
They are placed around the overall $\kappa\times\kappa$-grid, maintaining a distance of (roughly) $2\gamma$ to the base layout; see \Cref{fig:connected-constraints}b and \Cref{sec:connected-dummy-points}.
Due to their placement, %
we have the following.%
}{
The purpose of the third gadget (\Cref{sec:connected-dummy-points}) is to ensure that we only scale disks from tiles.
To enforce this, we introduce a set of dummy disks.
More specifically, we introduce~$4\kappa$ \emph{dummy disks} $d_{i,j}$, that is, at least one for every tile $X_{i,j}$ with $i \in \{1, \kappa\}$ or $j \in \{1, \kappa\}$.
They are placed around the overall $\kappa\times\kappa$-grid, maintaining a distance of more than $\rmax$ to the base layout; see also the example from \Cref{fig:connected-example}.
Due to their distance from the tiles, we must scale each dummy disk and at least one disk from each tile to ensure that the graph is connected, as also summarized with the following property.
}
\begin{restatable}{property}{connectedDummyDisksProperty}
\label{prop:connected-dummy-disks}
Let $r$ be a solution to \instance, and let $d_{i,j}$ be a dummy disk.
We have $d_{i,j}, p_{i,j}^{a,b} \in \mathcal{T}(r)$ for some tuple $(a,b) \in X_{i,j}$.
\end{restatable}
\medskip
\shortLong{
\Cref{prop:connected-dummy-disks} enforces us to scale each of the $4\kappa$ dummy disks.
Together with our intended semantic (recall also \Cref{prop:connected-scale-one-per-set}), we thus set $k = 4 \kappa + \kappa^2$.
We can set $\rmin = 0.1$.
}{
To complete the reduction, we compute a value for \rmax\ such that all of the above properties are fulfilled, set $\rmin=0.1$, and $k = 4\kappa + \kappa^2$.
}

\begin{statelater}{connectedConstructionDetails}
\subsection{Encoding the Tiles $\boldsymbol{X_{i,j}}$: The Base Layout}
\label{sec:connected-base-layout}
Let $X_{i,j} \in \mathcal{X}$, $i,j\in [\kappa]$ be a tile of the family $\mathcal{X}$.
Recall that we have $X_{i,j} \subseteq [\eta]^2$ for each $i, j \in [\kappa]$.
We introduce one point $p_{i,j}^{a,b}$ in $\mathcal{S}$ for every tuple $(a,b) \in X_{i,j}$ and place it at 
\begin{align*}
x(p_{i,j}^{a,b}) = i(2\eta + 2\gamma) + a\quad\quad\text{and}\quad\quad
y(p_{i,j}^{a,b}) = j(2\eta + 2\gamma) + b,
\end{align*}
where $\gamma$ is an instance-dependent \emph{displacement factor} that we fix towards the end of the construction. %
We let $\mathcal{S}[X_{i,j}] \subseteq \mathcal{S}$ denote the points corresponding to tuples from the tile $X_{i,j}$ for every $i,j \in [\kappa]$; see \Cref{fig:connected-base-layout} for a visualization of the construction.
As $a, b \in \mathbb{N}$, we make the following observation that we will later use in our correctness arguments; recall that we \NewText{scaled the instance down by a factor of 10.} %
\begin{observation}
\label{obs:grid-tiling-edgeless}
We have $E(G(\mathcal{S}[X_{i,j}], \Unit)) = \emptyset$ for every $i,j\in[\kappa]$.
\end{observation}

\onlyLong{\showConnectedBaseLayoutFigure}

We will use the following equivalence between a solution $\mathcal{Y}$ to $\instance_{>}$ and a solution $r$ to our constructed instance of \Connected.
\begin{align}
\label{eq:connected-equivalence}
x_{i,j} = (a,b) \in \mathcal{Y} \Longleftrightarrow p_{i,j}^{a,b} \in \mathcal{T}(r)
\end{align}
To use Equivalence~\ref{eq:connected-equivalence}, we will ensure that our construction has the following property.%
\connectedScaleOnePerTileProperty*

Towards realizing \Cref{prop:connected-scale-one-per-set}, we will ensure that $G(\mathcal{S}[X_{i,j}, r])$ is connected whenever $r$ scales a single disk from $\mathcal{S}[X_{i,j}]$.
To this end, observe that $\Dist{p_{i,j}^{a,b} - p_{i,j}^{a',b'}} \leq \sqrt{2}\eta$ holds for every $(a,b), (a',b') \in X_{i,j}$, for all $i,j \in [\kappa]$.
Consequently, we will set $\rmax$ to a value such that $\rmax \geq \sqrt{2}\eta$ holds.

Of course, with this alone our construction does not yet fulfill \Cref{prop:connected-scale-one-per-set}.
In particular, $\rmax \geq \sqrt{2}\eta$ alone not even guarantees the existence of any solution at all.
However, we will show in \Cref{sec:connected-correctness} that our finished construction fulfills \Cref{prop:connected-scale-one-per-set} (and admits a solution if and only if $\instance_{>}$ does).
Moreover, our constructed instance fulfills the following property, as we will show at the end of the next section.
\connectedScaleRelationsProperty*

\medskip

Points 3.\ and 5.\ are equivalent to points 2.\ and 4., respectively, if $p_{i,j}^{a,b}$ takes the role of~$p_{i',j'}^{a', b'}$ and vice-versa.

\subsection{Encoding the Value Constraints: Connecting Adjacent Sets}
\label{sec:connected-adjacent-sets}
Until now, the instance consists of $\BigO{\kappa^2\eta^2}$ disks arranged in $\kappa^2$ grids, which itself are arranged as a $\kappa \times \kappa$ grid.
With this base layout in place, we now ensure that we only scale disks $p_{i,j}^{a,b}$ and \NewText{$p_{i+1,j}^{a',b'}$} such that $a > a'$ for every column $i \in [\kappa - 1]$ (and likewise for every row $j \in [\kappa - 1]$).
In the following, we focus on modeling these constraints along the columns of $\instance_{>}$; the procedure for the rows is analogous.

Let $X_{i,j}, X_{i+1,j} \in \mathcal{X}$, $i \in [\kappa - 1]$, $j \in [\kappa]$ be two tiles of $\instance_{>}$.
Recall our overall idea, which consisted of placing $\eta$ evenly spaced disks $h_{i,j}^1, h_{i,j}^2, \ldots, h_{i,j}^{\eta}$ between the (grids of) disks from $\mathcal{S}[X_{i,j}]$ and $\mathcal{S}[X_{i + 1,j}]$ such that at least one of them remains disconnected if we violate a $>$-constraint.
More formally, our construction will satisfy the following property:	\connectedScaleHDisksProperty*

\medskip
We place the disk $h_{i,j}^q$, $q \in [\eta]$, at the following coordinates.
Observe that they are placed on a horizontal line, i.e., have identical $y$-coordinates; see also \Cref{fig:connected-base-layout}.
\begin{align*}
x(h_{i,j}^q) = i(2\eta + 2\gamma) + \eta+\gamma + q\quad\quad\text{and}\quad\quad
y(h_{i,j}^{q}) = j(2\eta + 2\gamma) + (\eta+1)/2
\end{align*}
We let $\mathcal{S}[h_{i,j}] = \{h_{i,j}^q : q \in [\eta]\}$ denote the above-introduced disks.

It is now time to find a value for $\gamma$.%
\onlyLong{\showConnectedDummyDisksFigure}
See \Cref{fig:connected-constraints}a for a visualization of the following discussion.
We first observe that the disk $p_{i,j}^{a,b}$ should intersect with all disks $h_{i,j}^q$ for $q < a$ if the former is scaled to \rmax.
This allows us to define first a value for \rmax, which we later use to specify $\gamma$.
We now aim to choose $\rmax$ such that the scaled disk ends near the midpoint between $h_{i,j}^{a - 1}$ and $h_{i,j}^a$.
For $a > 1$, we observe that the ``most extreme'' case corresponds to $q = a - 1$ (for $a = 1$, there is nothing to observe as there exists no disk $h_{i,j}^q$ for $q \leq 0$).
In this case, the \emph{$x$-distance} $d_x$ between $p_{i,j}^{a,b}$ and the \emph{midpoint} between $h_{i,j}^{a - 1}$ and~$h_{i,j}^{a}$ corresponds to $d_x = i(2\eta + 2\gamma) + \eta + \gamma + a - 0.5 - (i(2\eta + 2\gamma) + a) = \eta + \gamma - 0.5$.
The $y$-distance $d_y$ is independent of the choice of $a$ and $q$ and equals $d_y = \frac{\eta - 1}{2}$ at most (e.g., for $b = 1$).
Thus, we can set $\rmax = \sqrt{{d_x}^2 + {d_y}^2}$ and observe that we intersect in this case all disks $h_{i,j}^q$ with $q < a$ since we enclose the center of the disk $h_{i,j}^{a - 1}$.
Recall that we must not intersect the disk~$h_{i,j}^{a}$.
This constraint allows us to define a value for $\gamma$.
In particular, the distance $d$ between $p_{i,j}^{a,b}$ and $h_{i,j}^{a}$ is minimized if $y(p_{i,j}^{a,b}) = y(h_{i,j}^{a})$, which is the case for $b = \frac{(\eta + 1)}{2}$ (if this value exists).
In this case, the distance equals the $x$-distance between these points, which equals $d = i(2\eta + 2\gamma) + \eta+\gamma + a - (i(2\eta + 2\gamma) + a) = \eta + \gamma$.
To ensure that~$p_{i,j}^{a,b}$ and $h_{i,j}^{a}$ do not intersect (if only one of them is scaled), we must have $\rmax + 0.1 < d$.
We obtain,
\begin{align*}
{(\eta + \gamma - 0.5)}^2 + \left(\frac{\eta - 1}{2}\right)^2 &< (\eta + \gamma - 0.1)^2\\
\frac{\eta^2 - 2\eta + 1}{4} -0.8 \eta + 0.5^2 - 0.1^2 &< 0.8 \gamma\\
\frac{\eta^2 - 2\eta + 1}{4} -0.8 \eta + \frac{3}{4}&<0.8\gamma,
\end{align*}
where we increase for the last line the left side (as $0.5^2 - 0.1^2 < 3/4$) of the inequality.
This operation is safe since if we satisfy the inequalities in this case, it is also satisfied without this simplification.
We now set $\gamma = \eta^2$ and observe
\begin{align*}
\frac{\eta^2 - 2\eta + 1}{4} -0.8 \eta + \frac{3}{4}&<0.8\gamma\\
\frac{\eta^2 - 2\eta + 1}{4} -0.8 \eta + \frac{3}{4}&<0.8\eta^2\\
\frac{\eta^2 - 2\eta}{4} -0.8 \eta + 1&<0.8\eta^2\\
-2.2\eta^2 - 5.2\eta &<-4,
\end{align*}
which is true for all $\eta \geq 1$.
Hence, the desired inequality is achieved for $\gamma = \eta^2$.
Note that this implies $\rmax \geq 2\eta$, and hence $\rmax \geq \sqrt{2}\eta$ (for \NewText{ $\eta \geq 1$, which we assume}).

The construction of the points $v_{i,j}^q$ for $q \in [\eta]$ and two tiles $X_{i,j}, X_{i,j+1} \in \mathcal{X}$, $i \in [\kappa]$, $j \in [\kappa - 1]$ is symmetrical.
Similarly, we define $\mathcal{S}[v_{i,j}] = \{v_{i,j}^q : q \in [\eta]\}$.

\NewText{
Observe that if we do not scale any disk from $\mathcal{S}[X_{i,j}]$, then the disk $h_{i,j}^{1}$ remains isolated unless we scale a disk from $\mathcal{S}[h_{i,j}]$: By our choice of $\rmax$, scaling a disk from $\mathcal{S}[X_{i+1,j}]$ does not help and the disks $\mathcal{S}[v_{i,j}]$ are sufficiently far away from $\mathcal{S}[h_{i,j}]$.
We can make symmetric arguments for $h_{i-1,j}^{\eta}$ and $v_{i,j}^{1}$, and $v_{i,j-1}^{\eta}$ (if these disks exist), which leads to the following observation.
\begin{observation}
	\label{obs:connected-min-scale-tile}
	Let $r$ be a solution to \instance with $\mathcal{S}[X_{i,j}] \cap \mathcal{T}(r) = \emptyset$ for some $i, j \in [\kappa]$.
	Then, $\mathcal{T}(r)$ contains at least one disk from each of the sets $\mathcal{S}[h_{i,j}]$,  $\mathcal{S}[h_{i-1,j}]$, $\mathcal{S}[v_{i,j}]$, and $\mathcal{S}[v_{i,j-1}]$ (that exist).
\end{observation}
}

\subsection{Ensuring That We Scale One Disk Per Tile: The Dummy Disks}
\label{sec:connected-dummy-points}
In this section, we introduce $4\kappa$ \emph{dummy disks} that facilitate the correctness arguments that we will use in \Cref{sec:connected-correctness}.
More concretely, we will use these disks to ensure that we scale at least one disk per tile that we introduced in \Cref{sec:connected-base-layout}.

To this end, let $X_{1,j}$, \NewText{$j \in [\kappa]$} be a tile.
We introduce the dummy disk $d_{1,j}^{\ell}$ at
\begin{align*}
x(d_{1,j}^{\ell}) = 2\eta \quad\quad\text{and}\quad\quad
y(d_{1,j}^{\ell}) = j(2\eta + 2\gamma) + (\eta+1)/2.
\end{align*}
Note that $d_{1,j}^{\ell}$ is (roughly) $2\gamma = 2\eta^2$ to the left of the disks representing tuples in $X_{1,j}$.
Similarly, we introduce the dummy disk $d_{\kappa,j}^{r}$ located (roughly) $2\eta^2$ to the right of the disks representing tuples in $X_{\kappa,j}$, i.e., located at
\begin{align*}
x(d_{\kappa,j}^{r}) = \kappa(2\eta + 2\gamma) + \eta + 2\eta^2\quad\quad\text{and}\quad\quad
y(d_{\kappa,j}^{r}) = j(2\eta + 2\gamma) + (\eta+1)/2.
\end{align*}
The dummy disks $d_{i,1}^{t}$ and $d_{i,\kappa}^{b}$ which are (roughly) $2\eta^2$ to the top and bottom of the tiles $X_{i,1}$ and $X_{i,\kappa}$, \NewText{$i \in [\kappa]$}, respectively, are constructed analogously; see also \Cref{fig:connected-constraints}b.
In the following, we address a dummy disk by the tile $X_{i,j}$ it is placed next to, if, there is no risk of confusion.

Due to their distance to the remaining disks and the value for \rmax\ specified in \Cref{sec:connected-adjacent-sets}, our instance has the following property (see also \Cref{lem:connected-dummy-disks} for the proof that this is the case).
\connectedDummyDisksProperty*

As a consequence of \Cref{prop:connected-dummy-disks}, we must scale each of the $4\kappa$ dummy disks.
Together with our intended semantic (recall \Cref{eq:connected-equivalence}), we thus set $k = 4 \kappa + \kappa^2$.

\end{statelater}

\subsection{Combining the Gadgets: Correctness of the Reduction}
\label{sec:connected-correctness}
\shortLong{
\Cref{fig:connected-example} shows a small example of our reduction.
We now show that an above-constructed instance $\instance$ fulfills our properties, which we later use to argue correctness of the reduction.
}{%
In this section, we establish the correctness of our reduction and begin by briefly summarizing it; see also \Cref{fig:connected-example} for a small example of the reduction.
For a given instance $\instance_{>} = (\eta,\kappa,\mathcal{X})$ of \GridTiling, we created for each tile $X_{i,j}$ precisely $\Size{X_{i,j}}$-many disks $p_{i,j}^{a,b}$, each representing a tuple $(a,b) \in X_{i,j}$.
We arranged disks from the same tile in an $\eta \times \eta$-grid, and the tiles themselves in a $\kappa\times\kappa$-grid. 
Furthermore, we introduced between every pair of adjacent tiles a sequence of $\eta$ evenly spaced horizontal or vertical disks $h_{i,j}^q$ and $v_{i,j}^q$, $q \in [\eta]$, respectively, to ensure that the scaled disks adhere to the $>$-constraints.
Finally, we place $4\kappa$ dummy disks around the so-far placed disks to enforce that we must scale at least one disk from every set $\mathcal{S}[X_{i,j}]$ with $i \in \{1, \kappa\}$ or $j \in \{1, \kappa\}$.
Overall, we obtain an instance $\instance = \instanceLong$ of \Connected, where we set $\rmin = 0.1$ and recall $\rmax = \sqrt{(\eta + \eta^2 - 0.5)^2 + (\frac{\eta - 1}{2})^2}$.
For the following arguments, recall that we \NewText{scaled the instance down by a factor of 10.} %

Our correctness arguments in \Cref{thm:connected-w1-hard} heavily use \Cref{prop:connected-scale-one-per-set,prop:connected-scale-relations,prop:connected-scale-h-disks,prop:connected-dummy-disks}.
Thus, on our way to establish \Cref{thm:connected-w1-hard}, we first show that our instance \instance fulfills these properties.
We start with \Cref{prop:connected-dummy-disks}, since it will be used as a building block later.
}

\begin{figure}
\centering
\includegraphics[page=1]{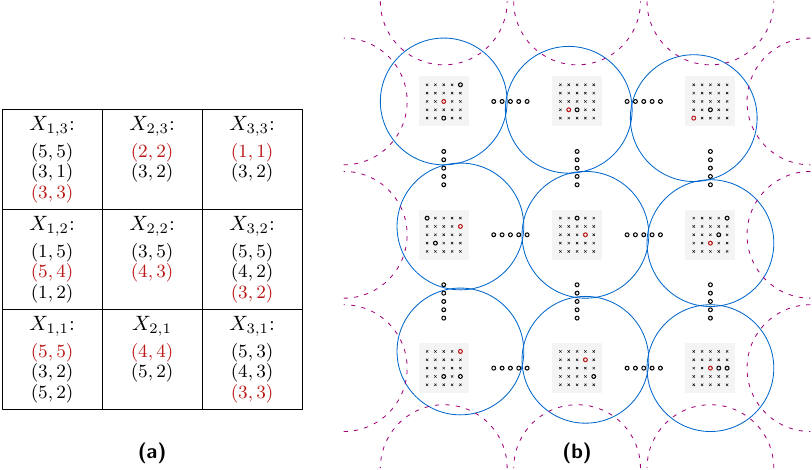}
\caption{
\textbf{\textsf{(a)}} An instance $(5,3,\mathcal{X})$ of \GridTiling[>] and \textbf{\textsf{(b)}} the corresponding instance of \Connected. Small disks and crosses in the tiles have the same meaning as in \Cref{fig:connected-base-layout}. We highlight a solution in red, the corresponding enlarged disks in blue, and the scaled dummy disks with purple dashed lines.
}
\label{fig:connected-example}
\end{figure}

\begin{statelater}{connectedLemmasForProperties}
\begin{lemma}
\label{lem:connected-dummy-disks}
Our instance \instance of \Connected fulfills \Cref{prop:connected-dummy-disks}.
\end{lemma}
\begin{proof}
Let $r$ be a solution and let $d_{i,j}$ be a dummy disk.
We assume $i = 1$ and $j \in [\kappa]$ for the remainder of the proof.
This is without loss of generality, the arguments for the other combinations of $i$ and $j$ are analogous.
We first show that $d_{i,j} \in \mathcal{T}(r)$ must hold and afterwards establish $p_{i,j}^{a,b} \in \mathcal{T}(r)$ for some tuple $(a,b) \in X_{i,j}$.

\proofsubparagraph*{Step 1: $\boldsymbol{d_{1,j} \in \mathcal{T}(r)}$.}
Assume for the sake of a contradiction $d_{1,j} \notin \mathcal{T}(r)$.
Observe that the disk $p_{1,j}^{1,b}$ with $b = (\eta+1)/2$ is closest to $d_{1,j}$ (if it exists).
In the following, we show that $d_{1,j}$ remains isolated if $d_{1,j} \notin \mathcal{T}(r)$, even if this disk exists and is scaled to $\rmax$.
Consequently, no matter which disk we scale, $d_{1,j} \notin \mathcal{T}(r)$ implies that $r$ is not a solution.
To this end, observe that we have $\Dist{p_{1,j}^{1,b} - d_{1,j}} \geq (2\eta + 2\gamma) + 1 - 2\eta = 2\gamma + 1$. 
Since $\rmax = \sqrt{(\eta + \gamma - 0.5)^2 + (\frac{\eta - 1}{2})^2} < 2\gamma + 0.9$ and $r(d_{1,j}) = 0.1$, $d_{1,j}$ is an isolated vertex in $G(\mathcal{S}, r)$, i.e., $G(\mathcal{S}, r)$ is not connected.
Hence, we have $d_{1,j} \in \mathcal{T}(r)$.

\proofsubparagraph*{Step 2: $\boldsymbol{p_{i,j}^{a,b} \in \mathcal{T}(r)}$.}
This second part of the statement can be shown using ideas analogous to above.
We first observe that the closest disk $p \in \mathcal{S} \setminus \mathcal{S}[X_{1,j}]$ \NewText{are the disks $v_{1,j}^1$ and $v_{1,j-1}^{\eta}$} and at least one of these disks exists, say $v_{1,j}^1$ \NewText{(other dummy disks are farther away).}.
Similar to above, it holds \NewText{$\Dist{v_{1,j}^1 - d_{1,j}} > 2\rmax$}, i.e., even if both disks are fully scaled $d_{1,j}$ remains an isolated vertex in $G(\mathcal{S}, r)$.
On the other hand, for any disk $p_{1,j}^{a,b} \in \mathcal{S}[X_{1,j}]$ we have $\Dist{p_{1,j}^{a,b} - d_{1,j}} \leq \sqrt{(2\gamma + \eta)^2 + (\frac{\eta - 1}{2})^2} \leq 2\rmax$ (for $\eta \geq 2$), i.e., scaling any disk in $\mathcal{S}[X_{1,j}]$ together with $d_{1,j}$ suffices to connect the latter disk with the remainder of the graph.
\end{proof}
\smallskip

Let $\mathcal{S}[i,\cdot]$ denote all disks in column $i \in [\kappa]$, including the disks $v_{i,j}^q$, $j \in [\kappa], q \in [\eta]$, and the respective dummy disks, i.e., $\mathcal{S}[i,\cdot] \coloneqq \bigcup_{j\in\kappa} \mathcal{S}[X_{i,j}] \cup \bigcup_{j\in\kappa} \mathcal{S}[v_{i,j}] \cup\{d_{i,1}^t, d_{i,\kappa}^b\}$.
Similarly, let $\mathcal{S}[\cdot,j]$ denote all disks in row $j\in [\kappa]$.
\Cref{prop:connected-dummy-disks} together with the placement of the disk $h_{i,j}^1$ ($v_{i,j}^1$) and $h_{i,j}^{\eta}$ ($v_{i,j}^{\eta}$) allows us to establish:

\begin{lemma}
\label{lem:connected-scale-per-row-column}
Let $r$ be a solution to \instance.
We have $\Size{\mathcal{T}(r) \cap \mathcal{S}[i,\cdot]},\Size{\mathcal{T}(r) \cap \mathcal{S}[\cdot, j]} \geq \kappa + 2$ for every $i,j \in [\kappa]$, i.e., we scale at least $\kappa + 2$ disks from every row/column.
\end{lemma}
\begin{proof}
Let $r$ be a solution to our instance and consider the set $\mathcal{S}[i,\cdot]$ for some $i \in [\kappa]$; the arguments for $\mathcal{S}[\cdot,j]$ will be symmetric.

By \Cref{prop:connected-dummy-disks} (see also \Cref{lem:connected-dummy-disks}), we know that we must scale the two dummy disks for column $i$, i.e., $d_{i,1}, d_{i,\kappa} \in \mathcal{T}(r) \cap \mathcal{S}[i,\cdot]$.
Similarly, we must scale at least one disk from the tile $X_{i,1}$ and $X_{i,\kappa}$.
Hence, we already have $\Size{\mathcal{T}(r) \cap \mathcal{S}[i,\cdot]} \geq 4$.

Observe that there are $\kappa - 2$ tiles in column $i$, ignoring $X_{i,1}$ and $X_{i,\kappa}$.
Applying \Cref{obs:connected-min-scale-tile} to each of them, we derive that in $r$, there are at least $\kappa - 2$ additional disks from column $i$ scaled.
In particular, note that if $r$ does not scale disks from the tiles, it must scale at least $\kappa - 1$ disks from the $\kappa - 1$ disjoint sets \NewText{$\mathcal{S}[v_{i,j}]$}, $j \in [\kappa - 1]$.
Together with the four scaling operations derived above, we conclude $\Size{\mathcal{T}(r) \cap \mathcal{S}[i,\cdot]} \geq \kappa + 2$.
\end{proof}

\smallskip
\Cref{lem:connected-scale-per-row-column} allows us to show \Cref{prop:connected-scale-one-per-set}, and the remaining properties can be established by estimating the largest and smallest possible distance between the disks in question.

\begin{lemma}
\label{lem:fulfill-prop-connected-scale-one-per-set}
Our instance \instance of \Connected fulfills \Cref{prop:connected-scale-one-per-set}.
\end{lemma}
\begin{proof}
Recall \Cref{prop:connected-scale-one-per-set}, which stated that we have $\Size{\mathcal{S}[X_{i,j}] \cap \mathcal{T}(r)} = 1$ for every $i,j \in [\kappa]$ and solution $r$.
Towards a contradiction, assume that there exists some $i,j \in [\kappa]$ where this does not hold.
There are two cases that we can consider.

\proofsubparagraph*{Case 1: $\boldsymbol{\Size{\mathcal{S}[X_{i,j}] \cap \mathcal{T}(r)} = 0}$.}
\NewText{Recall that we have $k = \kappa^2 + 4\kappa$.
By \Cref{prop:connected-dummy-disks}, $4\kappa$ scaling operations are already ``reserved'' for the dummy disks.
Thus,~$\kappa^2$ scaling operations ``remain''. %
Since $\Size{\mathcal{S}[X_{i,j}] \cap \mathcal{T}(r)} = 0$, we have $h_{i-1,j}^q, h_{i,j}^{q'} \in \mathcal{T}(r)$ or $v_{i,j-1}^q, v_{i,j}^{q'} \in \mathcal{T}(r)$, for some $q,q' \in [\eta]$; otherwise, we can find at least one disk in $\mathcal{S}[X_{i,j}]$ which is isolated in $G(\mathcal{S}, r)$.
Assume $h_{i-1,j}^q, h_{i,j}^{q'} \in \mathcal{T}(r)$, the other case is symmetric.
We observe that $1 < i < \kappa$ must hold, otherwise one of these disks would not exist.
Ignoring the tiles $X_{1,j}$, $X_{i,j}$, and $X_{\kappa,j}$, there are $\kappa - 3$ further tiles remaining in row $j$.
To each of them, we can apply \Cref{obs:connected-min-scale-tile}.
There are $\kappa - 1$ sets $\mathcal{S}[h_{i',j}]$, $i' \in [\kappa - 1]$.
After ignoring the sets containing the disks $h_{i-1,j}^q$ and $h_{i,j}^{q'}$, $\kappa - 3$ remain.
Hence, no matter to which sets $\mathcal{S}[h_{i',j}]$ or $\mathcal{S}[X_{i',j}]$ the remaining disks in $\mathcal{T}(r) \cap \mathcal{S}[\cdot,j]$ belong, there must be at least $\kappa - 3$ of them.
Consequently, we have $\Size{\mathcal{T}(r) \cap \mathcal{S}[\cdot,j]} \geq \kappa + 3$, i.e., at least $\kappa + 3$ scaling operations must be associated with row~$j$.
Since every other row needs at least $\kappa + 2$ scaling operations, and there are $\kappa - 1$ other rows, we need at least $(\kappa - 1)(\kappa + 2) + \kappa  + 3 = \kappa(\kappa + 2) + 1$ scaling operations across all rows.
Furthermore, $2\kappa$ our of the $4\kappa$ scaling operations for the dummy disks are not yet covered, namely those for the dummy disks above and below the tiles.
Thus, we would have more than $\kappa^2 + 4\kappa$ scaling operations.
Consequently, $\Size{\mathcal{S}[X_{i,j}] \cap \mathcal{T}(r)} = 0$ cannot hold.
}

\proofsubparagraph*{Case 2: $\boldsymbol{\Size{\mathcal{S}[X_{i,j}] \cap \mathcal{T}(r)} \geq 2}$.}
By \Cref{prop:connected-dummy-disks}, we know that for every dummy point $d_{i',j'}$ we have $d_{i',j'} \in \mathcal{T}(r)$.
Hence, $\Size{\mathcal{T}(r) \cap \bigcup_{i',j' \in [\kappa]} \mathcal{S}[X_{i',j'}]} \leq \kappa^2$.
Since $\Size{\mathcal{S}[X_{i,j}] \cap \mathcal{T}(r)} \geq 2$, there exists, by the pigeonhole principle, some $i',j' \in [\kappa]$ such that $\Size{\mathcal{S}[X_{i',j'}] \cap \mathcal{T}(r)} = 0$.
We can apply the arguments from Case~1 to this set $\mathcal{S}[X_{i',j'}]$ to derive that it cannot exist in a solution $r$.
Hence, we can conclude that also $\Size{\mathcal{S}[X_{i,j}] \cap \mathcal{T}(r)} \geq 2$ cannot hold. 
Combining both cases yields the statement.
\end{proof}

We continue with showing that our choice for $\gamma$ and $\rmax$ fulfills \Cref{prop:connected-scale-relations,prop:connected-scale-h-disks}.
\begin{lemma}
\label{lem:distance-properties-scale-relations}
Our instance \instance of \Connected fulfills \Cref{prop:connected-scale-relations}.
\end{lemma}
\begin{proof}
Recall $\gamma = \eta^2$ and $\rmax = \sqrt{(\eta + \eta^2 - 0.5)^2 + (\frac{\eta - 1}{2})^2}$, which can be simplified to $\sqrt{4\eta^4+8\eta^3+\eta^2-6\eta+2}/2$.
Observe that $\rmax < \eta^2 + \eta$ since 
\begin{align*}
	{\rmax}^2 
	= \eta^4 + 2\eta^3 + \frac{\eta^2}{4} - \frac{6\eta}{4} + 0.5
	< \eta^4 + 2\eta^3 + \eta^2 = (\eta^2 + \eta)^2.
\end{align*}

Towards showing that our instance \instance fulfills \Cref{prop:connected-scale-relations}, let $p_{i,j}^{a,b}$ and $p_{i',j'}^{a',b'}$ be two disks and $r$ a radius assignment such that $r(p_{i,j}^{a,b}) = \rmax$. 
We show each item separately.

\begin{enumerate}
	\item We first assume $i' = i$ and $j' = j$.
	In this case, we can observe $\Dist{p_{i,j}^{a,b} - p_{i',j'}^{a',b'}} \leq \Dist{p_{i,j}^{\eta,\eta} - p_{i',j'}^{1,1}} \leq \eta\sqrt{2} < 2\eta \leq \rmax$\NewText{; recall that we assume $\eta > 1$}.
	Thus, the disks intersect.
	
	\item \label{case:distance-properties-scale-relations} Next, we assume $i' = i + 1$ and $j' = j$.
	Recall that we need to show that $p_{i,j}^{a,b}$ and $p_{i',j'}^{a',b'}$ intersect if and only if $p_{i',j'}^{a',b'} \in \mathcal{T}(r)$ and $a > a'$.
	First, observe that 
	\begin{align*}
		\Dist{p_{i,j}^{a,b} - p_{i+1,j}^{a',b'}} \geq \Dist{p_{i,j}^{\eta,b} - p_{i+1,j}^{1,b}} = \eta + 2\gamma + 1 = 2\eta^2 + \eta + 1.
	\end{align*}
	Since $\rmax < \eta^2 + \eta$, we conclude that $\rmax + 0.1 < 2\eta^2 + \eta + 1$.
	Hence, the disks can only intersect if $p_{i',j'}^{a',b'} \in \mathcal{T}(r)$.
	We next show that we must have in addition $a > a'$.
	On the one hand, if we have $a > a'$, we obtain
	\begin{align*}
		\Dist{p_{i,j}^{a,b} - p_{i+1,j}^{a',b'}} \leq \Dist{p_{i+1,j}^{a-1,\eta} - p_{i,j}^{a,1}} = \sqrt{(2\eta + 2\eta^2 - 1)^2 + (\eta - 1)^2}.
	\end{align*}
	Above term can be simplified to $\sqrt{4\eta^4+8\eta^3+\eta^2-6\eta+2} = 2\rmax$.
	Hence, if $a > a'$, then both disks intersect (if $r(p_{i,j}^{a,b}) = r(p_{i+1,j}^{a',b'}) = \rmax$).
	On the other hand, if we have $a \leq a'$, we obtain
	\begin{align*}
		\Dist{p_{i,j}^{a,b} - p_{i+1,j}^{a',b'}} \geq \Dist{p_{i+1,j}^{a,b} - p_{i,j}^{a,b}} = 2\eta + 2\eta^2.
	\end{align*}
	Since $\rmax < \eta^2 + \eta$, we have $\Dist{p_{i,j}^{a,b} - p_{i+1,j}^{a',b'}} > 2\rmax$.
	Thus, the disks $p_{i,j}^{a,b}$ and $p_{i+1,j}^{a',b'}$ do not intersect, even if their radii equals $\rmax$.
	Combining all, we obtain the desired property.
	
	\customEnumerateItem{3.\ -- 5.} These cases are symmetric to Case~\ref{case:distance-properties-scale-relations}.

	\customEnumerateItem{6.} Finally, we assume that the tiles $X_{i,j}$ and $X_{i',j'}$ are not adjacent.
	We need to show that the two disks $p_{i,j}^{a,b}$ and $p_{i+1,j}^{a',b'}$ do not intersect in this case.
	Since the tiles are not adjacent, we have %
	\NewText{(a) $\Size{i' - i} \geq 2$, or (b) $\Size{j' - j} \geq 2$, or (c) $\Size{i' - i} = \Size{j' - j} = 1$.}
	The first two cases follow from arguments analogous to the above, and it remains to consider the last case, i.e., \NewText{$\Size{i' - i} = \Size{j' - j} = 1$}.
	\NewText{Without loss of generality, we assume $i' = i + 1$ and $j' = j + 1$}.
	We can make the following observation.
	\begin{align*}
		\Dist{p_{i,j}^{a,b} - p_{i+1,j+1}^{a',b'}} \geq \Dist{p_{i+1,j+1}^{1,1} - p_{i,j}^{\eta,\eta}} = \sqrt{2}(2\eta^2 + \eta + 1).
	\end{align*}
	On the other hand, we have $2\rmax < 2\eta^2 + 2\eta$.
	To see that $\Dist{p_{i,j}^{a,b} - p_{i+1,j+1}^{a',b'}} > 2\rmax$, we now combine both bounds:
	\begin{align*}
		\Dist{p_{i,j}^{a,b} - p_{i+1,j+1}^{a',b'}} \geq \sqrt{2}(2\eta^2 + \eta + 1) &> 2\eta^2 + 2\eta\\
		2\sqrt{2}\eta^2 + \sqrt{2}\eta + \sqrt{2} &> 2\eta^2 + 2\eta\\
		2\eta^2(\sqrt{2} - 1) + \eta(\sqrt{2}-2) + \sqrt{2} &> 0
	\end{align*}
	The last inequality is satisfied as %
	\NewText{
		\begin{align*}
			2\eta^2(\sqrt{2} - 1) + \eta(\sqrt{2}-2) + \sqrt{2} \geq 2\cdot 2\eta(\sqrt{2} - 1) + \eta(\sqrt{2}-2)=\eta(4\sqrt{2} - 4) > 0
		\end{align*}
		since $\sqrt{2} > 1$ and $\eta \geq 2$.
	}
	
\end{enumerate}
Combining all points, we conclude that our instance satisfies \Cref{prop:connected-scale-relations}.
\end{proof}
\begin{lemma}
\label{lem:distance-properties-h-disks}
Our instance \instance of \Connected fulfills \Cref{prop:connected-scale-h-disks}.
\end{lemma}
\begin{proof}
Recall once more $\gamma = \eta^2$ and $\rmax = \sqrt{(\eta + \eta^2 - 0.5)^2 + (\frac{\eta - 1}{2})^2}$, which can be simplified to $\sqrt{4\eta^4+8\eta^3+\eta^2-6\eta+2}/2$.
Recall $\rmax < \eta^2 + \eta$.
\NewText{We now show that even $\rmax + 0.1 < \eta^2 + \eta$ holds.
To this end, we consider the difference between ${\rmax}^2$ and $(\eta^2 + \eta - 0.1)^2 = \eta^4+2\eta^3+0.8\eta^2-0.2\eta+0.1^2$: 
\begin{align*}
	&{\rmax}^2 - (\eta^4+2\eta^3+0.8\eta^2-0.2\eta+0.1^2)\\
	&= \eta^4 + 2\eta^3 + \frac{\eta^2}{4} - \frac{6\eta}{4} + 0.5 - \eta^4-2\eta^3-0.8\eta^2+0.2\eta-0.1^2\\
	&= \eta^2(0.25 - 0.8) - 1.3\eta + 0.49\\
	&< 0
\end{align*}
As $\eta > 1$, we can indeed conclude $\rmax + 0.1 < \eta^2 + \eta$.} 

Now, towards showing that our instance \instance fulfills \Cref{prop:connected-scale-h-disks}, let $p_{i,j}^{a,b}$ be a disk and $r$ a solution to \instance such that $r(p_{i,j}^{a,b}) = \rmax$ and $\mathcal{S}[>]\cap\mathcal{T}(r)=\emptyset$.

Assume that the disks $h_{i,j}^q$ and $h_{i-1,j}^q$ exist for some $q \in [\eta]$,  i.e., $1 < i < \kappa$; otherwise there is nothing to show.
We concentrate on the first two points, i.e., that the disks $p_{i,j}^{a,b}$ and~$h_{i,j}^q$ intersect if and only if $q < a$ and $p_{i,j}^{a,b}$ and $h_{i-1,j}^q$ intersect if and only if $q > a$.
The last two points follow from analogous arguments.

\begin{enumerate}
	\item We first show that if $q < a$, then the two disks $p_{i,j}^{a,b}$ and $h_{i,j}^q$ intersect.
	We observe
	\begin{align*}
		\Dist{h_{i,j}^q - p_{i,j}^{a,b}} \leq \Dist{h_{i,j}^{a - 1} - p_{i,j}^{a,{1}}} = \sqrt{(\eta^2 + \eta - 1)^2 + (\frac{\eta + 1}{2} - 1)^2} = d.
	\end{align*}
	Squaring and simplifying $d$ and $\rmax$ yields
	\begin{align*}
		d^2 = \frac{4\eta^4 + 8\eta^3-3\eta^2-10\eta + 5}{4} \quad\quad\text{and}\quad\quad
		{\rmax}^2 = \frac{4\eta^4+8\eta^3+\eta^2-6\eta+2}{4}.
	\end{align*}
	We now observe that $d^2 \leq {(\rmax)}^2$ implying, as $d$ and $\rmax$ are positive, $d \leq \rmax$.
	Hence, the disks intersect.
	
	We now show that if $q \geq a$, then the two disks $p_{i,j}^{a,b}$ and $h_{i,j}^q$ do not intersect.
	Similar to before, we observe
	\begin{align*}
		\Dist{h_{i,j}^q - p_{i,j}^{a,b}} \geq \Dist{h_{i,j}^{a} - p_{i,j}^{a,{(\eta+1)/2}}} = \eta^2 + \eta.
	\end{align*}
	As \NewText{$\rmax + 0.1 < \eta^2 + \eta$}, we conclude that the disks do not intersect; recall $r(h_{i,j}^q) = 0.1$ for all $q \in [\eta]$.
	
	\item We now turn our attention to the disk $h_{i-1,j}^q$.
	First, consider the case \NewText{$q > a$}.
	We observe
	\begin{align*}
		\Dist{h_{i-1,j}^q - p_{i,j}^{a,b}} \leq \Dist{p_{i,j}^{a,{1}} - h_{i-1,j}^{a + 1}} = \sqrt{(\eta^2 + \eta - 1)^2 + (\frac{\eta + 1}{2} - 1)^2} = d.
	\end{align*}
	From here on, we can proceed analogously to the above argument and see that the two disks intersect.
	Second, for the case $q \leq a$, we observe
	\begin{align*}
		\Dist{h_{i-1,j}^q - p_{i,j}^{a,b}} \geq \Dist{p_{i,j}^{a,{(\eta+1)/2}} - h_{i-1,j}^{a}} = \eta^2 + \eta.
	\end{align*}
	As before, observing \NewText{$\rmax + 0.1 < \eta^2 + \eta$} shows that the disks do not intersect.	
\end{enumerate}
With analogous derivations for points three and four in \Cref{prop:connected-scale-h-disks}, we complete the proof.
\end{proof}
Combining \Cref{lem:fulfill-prop-connected-scale-one-per-set,lem:distance-properties-scale-relations,lem:distance-properties-h-disks,lem:connected-dummy-disks} yields \Cref{lem:connected-properties}.

\end{statelater}

\begin{restatable}\restateref{lem:connected-properties}{lemma}{lemmaConnectedProperties}
\label{lem:connected-properties}
Our instance \instance of \Connected fulfills \Cref{prop:connected-scale-one-per-set,prop:connected-scale-relations,prop:connected-scale-h-disks,prop:connected-dummy-disks}.
\end{restatable}
\smallskip
We use \Cref{lem:connected-properties} in the (more involved) backward direction of the correctness argument of the reduction. 
In particular, \Cref{prop:connected-scale-one-per-set,prop:connected-dummy-disks} ensures that a solution to $\instance_{>}$ is well-defined, and \Cref{prop:connected-scale-relations,prop:connected-scale-h-disks} ensure that the tuples corresponding to the scaled disks adhere to the $>$-constraints; note that we can assume without loss of generality that $r(p) = \rmax$ holds in a solution $r$ for every disk $p \in \mathcal{T}(r)$ since connectivity is preserved under enlarging disks.

\begin{restatable}\restateref{thm:connected-w1-hard}{theorem}{theoremConnectedHard}
\label{thm:connected-w1-hard}
\Connected is \W\textup{[1]}-hard when parameterized by $k$.
\end{restatable}
\begin{prooflater}{ptheoremConnectedHard}
Let $\instance_{>} = (\eta,\kappa,\mathcal{X})$ be an instance of \GridTiling.
Furthermore, let $\instance = \instanceLong$ be the instance of \Connected as constructed above; recall $\rmin = 0.1$, $\rmax = \sqrt{(\eta + \eta^2 - 0.5)^2 + (\frac{\eta - 1}{2})^2}$, and $k = 4\kappa + \kappa^2$.
Furthermore, recall that we assume one unit to equal to $0.1$, i.e., we scaled down the instance by \NewText{a factor of~10.} %
Observe that $\mathcal{S}$ contains $\BigO{\kappa^2\eta^2}$-many disks.
Hence, the size of \instance is polynomial in $\eta + \kappa$ and our parameter $k$ is polynomially bounded in $\kappa$.
Note that \instance can be constructed in polynomial time.
We now establish correctness of our reduction; recall \Cref{lem:connected-properties}, where we showed that \instance fulfills \Cref{prop:connected-scale-one-per-set,prop:connected-scale-relations,prop:connected-scale-h-disks,prop:connected-dummy-disks}.

\proofsubparagraph*{($\boldsymbol{\Rightarrow}$)}
Assume that $\instance_{>}$ is a yes-instance and let $\mathcal{Y} = \{x_{i,j} : i,j \in [\kappa]\}$ be a solution.
We now construct a radius assignment $r$ as follows.
For every dummy disk $d_{i,j}$, we set $r(d_{i,j}) = \rmax$.
Furthermore, for every $i,j \in [\kappa]$, there exists an $x_{i,j} = (a,b) \in \mathcal{Y}$.
We set $r(p_{i,j}^{a,b}) = \rmax$.
Finally, we set $r(p) = 0.1$ for the remaining disks.
We observe $\Size{\mathcal{T}(r)} = \Size{\mathcal{Y}} + 4\kappa = \kappa^2 + 4\kappa = k$.
To see that $G(\mathcal{S}, r)$ is connected, we scale, on the one hand, the necessary disks according to \Cref{prop:connected-dummy-disks}.
On the other hand, since $\mathcal{Y}$ is a solution to $\instance_{>}$, the selected tuples fulfill the $>$-constraints.
Therefore, we can deduce from \Cref{prop:connected-scale-relations} that the disks $p_{i,j}^{a,b}$ for \NewText{$x_{i,j} = (a,b) \in \mathcal{Y}$} induce a $\kappa\times\kappa$-grid in $G(\mathcal{S}, r)$, i.e., disks from neighboring tiles are adjacent, which is a connected graph.
Further, \Cref{prop:connected-scale-relations} ensures that every disk $p_{i,j}^{a',b'} \in \mathcal{S}[X_{i,j}]$ is adjacent to $p_{i,j}^{a,b}$ for every $i,j \in [\kappa]$.
What remains to show is that the disks $\mathcal{S}[h_{i,j}]$ and \NewText{$\mathcal{S}[v_{i,j}]$} are connected to the rest of the graph, for every $i\in[\kappa], j \in [\kappa]$ (if they exist).
Consider the tuples $x_{i,j} = (a,b),x_{i + 1, j}=(a',b') \in \mathcal{Y}$.
Using \Cref{prop:connected-scale-relations}, we derive that $p_{i,j}^{a,b}$ intersects with the disks $h_{i,j}^{q}$, $q < a$, and $p_{i + 1,j}^{a',b'}$ intersects with the disks $h_{i,j}^{q'}$, $q' > a'$.
Since $\mathcal{Y}$ is a solution, we have $a > a'$.
Consequently, every disk $h_{i,j}^{q}$ intersects with $p_{i,j}^{a,b}$ or $p_{i + 1,j}^{a',b'}$ (or both).
In particular, they are connected to the $\kappa \times \kappa$-grid induced by the disks scaled according to $\mathcal{Y}$.
Hence, all disks $\mathcal{S}[h_{i,j}]$ are connected with the rest of the graph.
Similar arguments apply for the disks $\mathcal{S}[v_{i,j}]$ and combining all, we conclude that $r$ is a solution to \instance.

\proofsubparagraph*{($\boldsymbol{\Leftarrow}$)}
Assume that $\instance$ is a yes-instance and consider a solution $r$.
Since connectivity is maintained under enlarging disks, we can assume without loss of generality that $r(p) = \rmax$ for every $p \in \mathcal{T}(r)$.
We construct a set $\mathcal{Y}$ as follows.
For every tile $X_{i,j}$, we add the unique entry $x_{i,j} = (a,b)$ with $r(p_{i,j}^{a,b}) = \rmax$ to $\mathcal{Y}$.
By \Cref{prop:connected-scale-one-per-set}, there is a single disk $p_{i,j}^{a,b}$ with this property, so this operation is well-defined.
Moreover, together with \Cref{prop:connected-dummy-disks}, we obtain that no other disk is scaled except the $4\kappa$ dummy disks.

We now show that $\mathcal{Y}$ is a solution to $\instance_{>}$ and assume, towards a contradiction, that this is not the case.
This implies that there exist two neighboring tiles $X_{i,j}, X_{i',j'} \in \mathcal{X}$ such that $x_{i,j} = (a,b)$ and $x_{i',j'} = (a',b')$ do not fulfill the $>$-constraints.
We consider the case $i' = i + 1$ and $j = j'$, the remaining cases are symmetric.
Hence, we have $a \leq a'$.
From \Cref{prop:connected-scale-relations}, we get that $p_{i,j}^{a,b}$ intersects with the disks $h_{i,j}^{q}$, $q < a$, and $p_{i + 1,j}^{a',b'}$ intersects with the disks $h_{i,j}^{q'}$, $q' > a'$
In particular, since $a \leq a'$, this implies that neither of the two disks intersects the disk $h_{i,j}^{a}$.
Moreover, $\mathcal{T}(r) \cap (\mathcal{S}[X_{i,j}] \cup \mathcal{S}[h_{i,j}] \cup \mathcal{S}[X_{i + 1,j}]) = \{p_{i,j}^{a,b}, p_{i',j'}^{a',b'}\}$ as derived above from combining \Cref{prop:connected-scale-one-per-set,prop:connected-dummy-disks}.
Thus, $h_{i,j}^a$ is an isolated disk in $G(\mathcal{S}, r)$, contradicting the fact that $r$ is a solution.
Combining all, we conclude that $\mathcal{Y}$ is a solution to $\instance_{>}$.
\end{prooflater}

\section{Generalizations and Concluding Remarks}
\label{sec:conclusion}
We contributed to the study of geometric-aware modification operations for \onlyLong{geometric} intersection graphs, focusing on disk-scaling in unit disk graphs where each modified disk can choose a radius from \onlyLong{the interval }$[\rmin, \rmax]$.
Our work opens up several promising avenues for future research:
\onlyLong{%
\subparagraph*{Extension to Disk Graphs.}

Most results extend to (non-unit) disk graphs by adapting the LP formulation (\Cref{sec:frameworks}).
The exception is \Cref{thm:complete}, since the complexity of \probname{Clique} in disk graphs is unknown~\cite{KM.MCP.2025}.
Overcoming this barrier is a promising direction to pursue.

\subparagraph*{Optimization Variants.}
}%
Fomin et al.~\cite{FG00Z25} also studied an optimization variant where we scale (at most) $k$~disks and minimize $\sum_{p \in \mathcal{S}} \Size{1 - r(p)}$.
Generalizing our algorithmic results to this \emph{min}-variant is non-trivial and remains open.
In particular, %
we cannot simply adapt the optimization function of the LP due to our (maximized) ``slack-variable'' $\varepsilon$ (and only determining the disks to scale is not sufficient since disks might require different radii). %
\shortLong{%
It is also open whether \Cluster admits a polynomial kernel for $k$ or is in \NP.
Finally, it would be interesting to combine scaling with other modification operations, such as disk dispersion~\cite{FG00Z23}.
}{
\subparagraph*{Other Open Questions.}
The most prominent open question from our work is whether we can obtain a polynomial kernel for \Cluster.
Furthermore, \NP-containment of \Cluster\ remains open.
Finally, it would be interesting to combine scaling with other modification operations, such as disk dispersion~\cite{FG00Z23}, or study the impact of structural graph parameters, such as the vertex cover number, on our problem.
}

\bibliography{bibliography}

@InProceedings{FG00Z23,
  author    = {Fedor V. Fomin and Petr A. Golovach and Tanmay Inamdar and Saket Saurabh and Meirav Zehavi},
  booktitle = {Proc. 31st European Symposium on Algorithms (ESA'23)},
  title     = {Kernelization for Spreading Points},
  year      = {2023},
  editor    = {Inge Li G{\o}rtz and Martin Farach{-}Colton and Simon J. Puglisi and Grzegorz Herman},
  pages     = {48:1--48:16},
  publisher = {Schloss Dagstuhl - Leibniz-Zentrum f{\"{u}}r Informatik},
  series    = {LIPIcs},
  volume    = {274},
  doi       = {10.4230/LIPICS.ESA.2023.48},
}

@InProceedings{Fishkin03,
  author    = {Aleksei V. Fishkin},
  booktitle = {Proc. 1st Workshop on Approximation and Online Algorithms (WAOA'03)},
  title     = {Disk Graphs: {A} Short Survey},
  year      = {2003},
  editor    = {Klaus Jansen and Roberto Solis{-}Oba},
  pages     = {260--264},
  publisher = {Springer},
  series    = {Lecture Notes in Computer Science},
  volume    = {2909},
  doi       = {10.1007/978-3-540-24592-6\_23},
}

@Article{HK01,
  author  = {Petr Hlinen{\'{y}} and Jan Kratochv{\'{\i}}l},
  journal = {Discrete Mathematics},
  title   = {Representing graphs by disks and balls (a survey of recognition-complexity results)},
  year    = {2001},
  number  = {1-3},
  pages   = {101--124},
  volume  = {229},
  doi     = {10.1016/S0012-365X(00)00204-1},
}

@InProceedings{DKHO24,
  author    = {Nicol{\'{a}}s Honorato Droguett and Kazuhiro Kurita and Tesshu Hanaka and Hirotaka Ono},
  booktitle = {Proc. 18th International Joint Conference on Frontiers in Algorithmics (IJTCS-FAW'24)},
  title     = {Algorithms for Optimally Shifting Intervals Under Intersection Graph Models},
  year      = {2024},
  editor    = {Bo Li and Minming Li and Xiaoming Sun},
  pages     = {66--78},
  publisher = {Springer},
  series    = {Lecture Notes in Computer Science},
  volume    = {14752},
  doi       = {10.1007/978-981-97-7752-5\_5},
}

@Article{ABUKHZAM2025,
  author  = {Faisal N. Abu{-}Khzam and Judith Egan and Serge Gaspers and Alexis Shaw and Peter Shaw},
  journal = {Discrete Applied Mathematics},
  title   = {Cluster Editing with Vertex Splitting},
  year    = {2025},
  pages   = {185-195},
  volume  = {371},
  doi     = {https://doi.org/10.1016/j.dam.2025.04.013},
}

@Article{CDFG23,
  author  = {Christophe Crespelle and P{\aa}l Gr{\o}n{\aa}s Drange and Fedor V. Fomin and Petr A. Golovach},
  journal = {Computer Science Review},
  title   = {A survey of parameterized algorithms and the complexity of edge modification},
  year    = {2023},
  pages   = {100556},
  volume  = {48},
  doi     = {10.1016/J.COSREV.2023.100556},
}

@InProceedings{FG00Z25,
  author    = {Fedor V. Fomin and Petr A. Golovach and Tanmay Inamdar and Saket Saurabh and Meirav Zehavi},
  booktitle = {Proc. 16th Innovations in Theoretical Computer Science (ITCS'25)},
  title     = {Parameterized Geometric Graph Modification with Disk Scaling},
  year      = {2025},
  editor    = {Raghu Meka},
  pages     = {51:1--51:17},
  publisher = {Schloss Dagstuhl - Leibniz-Zentrum f{\"{u}}r Informatik},
  series    = {LIPIcs},
  volume    = {325},
  doi       = {10.4230/LIPICS.ITCS.2025.51},
}

@InProceedings{DKHOW25,
  author    = {Nicol{\'{a}}s Honorato Droguett and Kazuhiro Kurita and Tesshu Hanaka and Hirotaka Ono and Alexander Wolff},
  booktitle = {Proc. 20th International Conference and Workshops on Algorithms and Computation (WALCOM'26)},
  title     = {Further Results on Rendering Geometric Intersection Graphs Sparse by Dispersion},
  year      = {2025},
  editor    = {Di Giacomo, Emilio and Mondal, Debajyoti},
  pages     = {451--466},
  publisher = {Springer},
  series    = {Lecture Notes in Computer Science},
  volume    = {16444},
  doi       = {https://doi.org/10.1007/978-981-95-7127-7_30},
}

@Article{EvdHM.SGN.2024,
  author  = {Jeff Erickson and Ivor {van der Hoog} and Tillmann Miltzow},
  journal = {SIAM Journal on Computing},
  title   = {Smoothing the Gap Between {NP} and {ER}},
  year    = {2024},
  number  = {6},
  pages   = {S20--102},
  volume  = {53},
  doi     = {10.1137/20M1385287},
}

@InProceedings{DKHO25,
  author    = {Nicol{\'{a}}s Honorato Droguett and Kazuhiro Kurita and Tesshu Hanaka and Hirotaka Ono},
  booktitle = {Proc. 19th International Symposium on Algorithms and Data Structures (WADS'25)},
  title     = {On the Complexity of Minimising the Moving Distance for Dispersing Objects},
  year      = {2025},
  editor    = {Pat Morin and Eunjin Oh},
  pages     = {36:1--36:14},
  publisher = {Schloss Dagstuhl - Leibniz-Zentrum f{\"{u}}r Informatik},
  series    = {LIPIcs},
  volume    = {349},
  doi       = {10.4230/LIPICS.WADS.2025.36},
}

@Book{Die.GT4.2012,
  author    = {Reinhard Diestel},
  publisher = {Springer},
  title     = {{G}raph Theory, {4th} {E}dition},
  year      = {2012},
  isbn      = {978-3-642-14278-9},
  series    = {Graduate {T}exts in {M}athematics},
  volume    = {173},
}

@InProceedings{HS95,
  author       = {Huson, M.L. and Sen, A.},
  booktitle    = {Proceedings of the Military Communications Conference ({MILCOM}~'95)},
  title        = {Broadcast scheduling algorithms for radio networks},
  year         = {1995},
  organization = {IEEE},
  pages        = {647--651},
  volume       = {2},
  doi          = {https://doi.org/10.1109/MILCOM.1995.483546},
}

@Article{XZ25,
  author  = {Jie Xue and Meirav Zehavi},
  journal = {Computer Science Review},
  title   = {Parameterized algorithms on geometric intersection graphs},
  year    = {2025},
  pages   = {100796},
  volume  = {58},
  doi     = {10.1016/J.COSREV.2025.100796},
}

@InProceedings{Karp72,
  author    = {Richard M. Karp},
  booktitle = {Proc. Computer Science Conferences \& Workshops},
  title     = {Reducibility {A}mong {C}ombinatorial {P}roblems},
  year      = {1972},
  editor    = {Raymond E. Miller and James W. Thatcher},
  pages     = {85--103},
  publisher = {Plenum Press, New York},
  series    = {The {IBM} Research Symposia Series},
  doi       = {10.1007/978-1-4684-2001-2\_9},
}

@InProceedings{Eppstein09,
  author    = {David Eppstein},
  booktitle = {Proc. 35th International Workshop on Graph-Theoretic Concepts in Computer Science (WG'09)},
  title     = {Graph-Theoretic Solutions to Computational Geometry Problems},
  year      = {2009},
  editor    = {Christophe Paul and Michel Habib},
  pages     = {1--16},
  publisher = {Springer},
  series    = {Lecture Notes in Computer Science},
  volume    = {5911},
  doi       = {10.1007/978-3-642-11409-0\_1},
}

@Article{CCJ90,
  author  = {Brent N. Clark and Charles J. Colbourn and David S. Johnson},
  journal = {Discrete Mathematics},
  title   = {Unit disk graphs},
  year    = {1990},
  number  = {1-3},
  pages   = {165--177},
  volume  = {86},
  doi     = {10.1016/0012-365X(90)90358-O},
}

@Article{BBKMZ18,
  author  = {Mark de Berg and Hans L. Bodlaender and S{\'{a}}ndor Kisfaludi{-}Bak and D{\'{a}}niel Marx and Tom C. van der Zanden},
  journal = {SIAM Journal on Computing},
  title   = {A Framework for Exponential-Time-Hypothesis-Tight Algorithms and Lower Bounds in Geometric Intersection Graphs},
  year    = {2020},
  number  = {6},
  pages   = {1291--1331},
  volume  = {49},
  doi     = {10.1137/20M1320870},
}

@Article{FLPSZ19,
  author  = {Fedor V. Fomin and Daniel Lokshtanov and Fahad Panolan and Saket Saurabh and Meirav Zehavi},
  journal = {Discrete \& Computational Geometry},
  title   = {Finding, Hitting and Packing Cycles in Subexponential Time on Unit Disk Graphs},
  year    = {2019},
  number  = {4},
  pages   = {879--911},
  volume  = {62},
  doi     = {10.1007/S00454-018-00054-X},
}

@InProceedings{FLS12,
  author    = {Fedor V. Fomin and Daniel Lokshtanov and Saket Saurabh},
  booktitle = {Proc. 23rd ACM-SIAM Symposium on Discrete Algorithms (SODA'12)},
  title     = {Bidimensionality and geometric graphs},
  year      = {2012},
  editor    = {Yuval Rabani},
  pages     = {1563--1575},
  publisher = {{SIAM}},
  doi       = {10.1137/1.9781611973099.124},
}

@Article{LY80,
  author  = {John M. Lewis and Mihalis Yannakakis},
  journal = {Journal of Computer and System Sciences (JCSS)},
  title   = {The Node-Deletion Problem for Hereditary Properties is {NP}-Complete},
  year    = {1980},
  number  = {2},
  pages   = {219--230},
  volume  = {20},
  doi     = {10.1016/0022-0000(80)90060-4},
}

@Article{Yannakakis81,
  author  = {Mihalis Yannakakis},
  journal = {SIAM Journal on Computing},
  title   = {Edge-Deletion Problems},
  year    = {1981},
  number  = {2},
  pages   = {297--309},
  volume  = {10},
  doi     = {10.1137/0210021},
}

@Article{Santi05,
  author  = {Paolo Santi},
  journal = {ACM Computing Surveys},
  title   = {Topology control in wireless ad hoc and sensor networks},
  year    = {2005},
  number  = {2},
  pages   = {164--194},
  volume  = {37},
  doi     = {10.1145/1089733.1089736},
}

@Article{CHAN18,
  author  = {Timothy M. Chan},
  journal = {ACM Transactions on Algorithms (TALG)},
  title   = {Improved Deterministic Algorithms for Linear Programming in Low Dimensions},
  year    = {2018},
  number  = {3},
  pages   = {30:1--30:10},
  volume  = {14},
  doi     = {10.1145/3155312},
}

@Article{Komusiewicz18,
  author  = {Christian Komusiewicz},
  journal = {ACM Transactions on Computation Theory (ToCT)},
  title   = {Tight Running Time Lower Bounds for Vertex Deletion Problems},
  year    = {2018},
  number  = {2},
  pages   = {6:1--6:18},
  volume  = {10},
  doi     = {10.1145/3186589},
}

@Book{FLSZ19,
  author    = {Fomin, Fedor V and Lokshtanov, Daniel and Saurabh, Saket and Zehavi, Meirav},
  publisher = {Cambridge University Press},
  title     = {Kernelization: theory of parameterized preprocessing},
  year      = {2019},
  doi       = {https://doi.org/10.1017/9781107415157},
}

@Book{CFK+.PA.2015,
  author    = {Marek Cygan and Fedor V. Fomin and Lukasz Kowalik and Daniel Lokshtanov and D{\'{a}}niel Marx and Marcin Pilipczuk and Michal Pilipczuk and Saket Saurabh},
  publisher = {Springer},
  title     = {Parameterized Algorithms},
  year      = {2015},
  isbn      = {978-3-319-21274-6},
  doi       = {10.1007/978-3-319-21275-3},
}

@Article{PSZ24,
  author  = {Fahad Panolan and Saket Saurabh and Meirav Zehavi},
  journal = {ACM Transactions on Algorithms (TALG)},
  title   = {Contraction Decomposition in Unit Disk Graphs and Algorithmic Applications in Parameterized Complexity},
  year    = {2024},
  number  = {2},
  pages   = {15},
  volume  = {20},
  doi     = {10.1145/3648594},
}

@InProceedings{EKM.FMC.2023,
  author    = {Espenant, Jared and Keil, J. Mark and Mondal, Debajyoti},
  booktitle = {Proc. 39th International Symposium on Computational Geometry (SoCG'23)},
  title     = {Finding a Maximum Clique in a Disk Graph},
  year      = {2023},
  editor    = {Erin W. Chambers and Joachim Gudmundsson},
  pages     = {30:1--30:17},
  publisher = {Schloss Dagstuhl - Leibniz-Zentrum f{\"{u}}r Informatik},
  series    = {LIPIcs},
  volume    = {258},
  doi       = {10.4230/LIPICS.SOCG.2023.30},
}

@InProceedings{KM.MCP.2025,
  author    = {Keil, J. Mark and Mondal, Debajyoti},
  booktitle = {Proc. 41st International Symposium on Computational Geometry (SoCG'25)},
  title     = {The Maximum Clique Problem in a Disk Graph Made Easy},
  year      = {2025},
  editor    = {Oswin Aichholzer and Haitao Wang},
  pages     = {63:1--63:16},
  publisher = {Schloss Dagstuhl – Leibniz-Zentrum für Informatik},
  series    = {LIPIcs},
  volume    = {332},
  doi       = {10.4230/LIPICS.SOCG.2025.63},
}

@Book{GJ.CIG.1979,
  author    = {Michael R. Garey and David S. Johnson},
  publisher = {W. H. Freeman},
  title     = {{C}omputers and Intractability: {A} Guide to the Theory of {NP}-{C}ompleteness},
  year      = {1979},
  booktitle = {Computers and intractability : a guide to the theory of NP-completeness},
}

@Article{LMS.LA2.1998,
  author  = {Yanpei Liu and Aurora Morgana and Bruno Simeone},
  journal = {Discrete Applied Mathematics},
  title   = {A Linear Algorithm for 2-bend Embeddings of Planar Graphs in the Two-dimensional Grid},
  year    = {1998},
  number  = {1–3},
  pages   = {69--91},
  volume  = {81},
  doi     = {10.1016/s0166-218x(97)00076-0},
}

@Article{HK.nAM.1973,
  author  = {John E. Hopcroft and Richard M. Karp},
  journal = {SIAM Journal on Computing},
  title   = {An n\({}^{\mbox{5/2}}\) Algorithm for Maximum Matchings in Bipartite Graphs},
  year    = {1973},
  number  = {4},
  pages   = {225--231},
  volume  = {2},
  doi     = {10.1137/0202019},
}

@InProceedings{JSWZ.fas.2021,
  author    = {Jiang, Shunhua and Song, Zhao and Weinstein, Omri and Zhang, Hengjie},
  booktitle = {Proc. 53rd Symposium on the Theory of Computing (STOC'21)},
  title     = {A faster algorithm for solving general {LP}s},
  year      = {2021},
  editor    = {Samir Khuller and Virginia Vassilevska Williams},
  pages     = {823--832},
  publisher = {ACM},
  doi       = {10.1145/3406325.3451058},
}

@Article{Ren.pta.1988,
  author  = {Renegar, James},
  journal = {Mathematical Programming},
  title   = {A polynomial-time algorithm, based on {N}ewton’s method, for linear programming},
  year    = {1988},
  number  = {1–3},
  pages   = {59--93},
  volume  = {40–40},
  doi     = {10.1007/bf01580724},
}

\onlyShort{

\appendix

\newpage
\section{Omitted Details from \Cref{sec:frameworks}}

\subsection{A Linear Program for Finding Radii in Polynomial Time}
\label{app:linear-programming}
\linearProgrammingFormulation

\subsection{Omitted Proofs}

\lemmaLinearProgramCorrectness*
\label{lem:linear-program-correctness*}
\plemmaLinearProgramCorrectness

\theoremXPFramework*
\label{thm:xp-framework*}
\ptheoremXPFramework

\newpage
\section{Omitted Details from \Cref{sec:cluster-fpt}}
\label{app:cluster-fpt}

\clusterFPTDetails

\subsection{Omitted Proofs}

\theoremClusterFPT*
\label{thm:cluster-fpt*}

Before we can prove \Cref{thm:cluster-fpt}, we first have to establish correctness of our algorithm.
\clusterFPTCorrectness

We now use \Cref{lem:cluster-fpt-correctness} to show that \Cluster is \FPT\ with respect to $k$, i.e., establish \Cref{thm:cluster-fpt}, which we recall below.
\theoremClusterFPT*
\ptheoremClusterFPT

\newpage
\section{Omitted Details from \Cref{sec:cluster-hardness}}
\label{app:cluster-hardness}

\clusterHardnessShrinkDetails

\clusterHardnessEnlargeDetails

\clusterHardnessCombinedDetails

\subsection{Omitted Proofs}
\propositionClusterHardnessShrink*
\label{prop:cluster-hardness-shrink*}
\ppropositionClusterHardnessShrink

\propositionClusterHardnessEnlarge*
\label{prop:cluster-hardness-enlarge*}
\ppropositionClusterHardnessEnlarge

\newpage
\section{Omitted Proof from \Cref{sec:complete}}
\label{app:complete}

\lemmaScalingCompleteClique*
\label{lem:scale-complete-clique*}
\plemmaScalingCompleteClique

\newpage
\section{Omitted Details from \Cref{sec:connected}}
\label{app:connected}

For a yes-instance $\instance_{>} = (\eta, \kappa, \mathcal{S})$ of \GridTiling, we let $\mathcal{Y} = \{x_{i,j} \in X_{i,j} : i,j \in [\kappa]\}$ denote the set of selected tuples, i.e., a \emph{solution}.

\connectedConstructionDetails

\subsection{Omitted Proofs}

\lemmaGridTilingHardness*
\label{lem:our-grid-tiling-hardness*}
\plemmaGridTilingHardness

\lemmaConnectedProperties*
\label{lem:connected-properties*}

To establish \Cref{lem:connected-properties}, we show for each property individually that our instance fulfills it.

\connectedLemmasForProperties
\lemmaConnectedProperties*

\theoremConnectedHard*
\label{thm:connected-w1-hard*}
\ptheoremConnectedHard
}

\end{document}